%% file: main.tex
\documentclass[custom]{imsart}%

\RequirePackage[authoryear]{natbib}
\RequirePackage{graphicx}

\usepackage{amsmath, amsfonts, amsthm, amssymb}
\usepackage{mathtools}

\usepackage[table]{xcolor}
\usepackage{multirow}
\usepackage{makecell}
\usepackage{hhline}

\usepackage{todonotes}

\usepackage{pgfplots}
\usepackage{pgfplotstable}
\pgfplotsset{compat=1.18}
\usepackage{dsfont}
\usepackage[titlenumbered,ruled,noend,vlined]{algorithm2e}

\makeatletter
\renewcommand*{\@algocf@post@ruled}{}
\makeatother

\usepackage{stmaryrd}
\usepackage{relsize}
\usepackage{subcaption}
\usetikzlibrary{spy}
\usepackage{nicematrix}
\usepackage{enumitem}
\usepackage{accents}
\usepackage{algpseudocode}

\newlist{equivcond}{enumerate}{1}
\setlist[equivcond,1]{
    label=\normalfont\bfseries(\roman*),
    leftmargin=4em,
    noitemsep,
    ref=\roman*,
}

\usepackage[toc,page]{appendix}
\usepackage[colorlinks,citecolor=blue,urlcolor=blue,filecolor=blue,backref=page]{hyperref}
\usepackage[capitalise]{cleveref}

\startlocaldefs

\DeclarePairedDelimiter{\norm}{\lVert}{\rVert}
\DeclarePairedDelimiter{\abs}{\lvert}{\rvert}

\newcommand{\TV}[0]{\textup{TV}}

\newcommand{\ambientspace}[0]{\mathcal{X}}
\newcommand{\ambientalgebra}[0]{\mathbb{X}}

\renewcommand{\Pr}[0]{\textup{Pr}}

\newcommand{\bfA}[0]{\mathbf{A}}
\newcommand{\bfB}[0]{\mathbf{B}}
\newcommand{\bfC}[0]{\mathbf{C}}
\newcommand{\bfr}[0]{\mathbf{r}}

\newcommand{\dd}[0]{\mathrm{d}}
\newcommand{\bbE}[0]{\mathbb{E}}
\newcommand{\bbR}[0]{\mathbb{R}}
\newcommand{\bbRp}[0]{\bbR^{+}}

\newcommand{\OO}{\mathcal{O}}  
\newcommand{\oo}{o}  
\newcommand{\Var}{\mathrm{Var}}

\newcommand{\newres}{\cellcolor{blue!15}}
\newcommand{\obviouscor}{\cellcolor{green!15}}

\newcommand\munderbar[1]{%
    \underaccent{\bar}{#1}}

\theoremstyle{plain}
\newtheorem{corollary}{Corollary}[section]
\newtheorem{theorem}{Theorem}[section]
\newtheorem{proposition}{Proposition}[section]
\newtheorem{remark}{Remark}[section]
\newtheorem{lemma}[theorem]{Lemma}

\theoremstyle{definition}

\newtheorem{assumption}[theorem]{Assumption}
\newtheorem*{example}{Example}

\endlocaldefs

\begin{document}

    \begin{frontmatter}
        \title{On the complexity of standard and waste-free SMC samplers}

        \begin{aug}
            \author[A]{\fnms{Yvann}~\snm{Le Fay}\ead[label=e1]{yvann.lefay@ensae.fr}}
            \author[A]{\fnms{Nicolas}~\snm{Chopin}\ead[label=e2]{nicolas.chopin@math.cnrs.fr}}
            \author[B]{\fnms{Matti}~\snm{Vihola}\ead[label=e3]{matti.s.vihola@jyu.fi}}
            \address[A]{CREST, ENSAE, IP Paris\printead[presep={ ,\ }]{e1,e2}}

            \address[B]{Department of Mathematics and Statistics,
                University of Jyväskylä\printead[presep={,\ }]{e3}}
        \end{aug}

        \begin{abstract}
            We establish finite sample bounds for the error of standard and waste-free
            SMC samplers. Our results cover estimates of both expectations and
            normalising constants of the target distributions. We consider first an
            arbitrary sequence of distributions, and then specialise our results to
            tempering sequences. We use our results to derive the complexity of SMC
            samplers with respect to the parameters of the problem, such as $T$, the
            number of target distributions, in the general case, or $d$, the dimension
            of the ambient space, in the tempering case.
            We use these bounds to derive practical recommendations for the implementation
            of SMC samplers for end users.
        \end{abstract}



    \end{frontmatter}

    \maketitle

    \section{Introduction}\label{sec:intro}

    \subsection{Background}\label{sub:background}

    SMC (Sequential Monte Carlo) samplers are a class of numerical algorithms used
    to approximate recursively a sequence of distributions $\pi_0, \dots, \pi_T$
    defined on a common probability space $(\ambientspace, \ambientalgebra)$.
    In some applications, each $\pi_t$ is of practical interest; for instance, in
    Bayesian on-line learning, $\pi_t$ may be the posterior distribution of the
    parameter given the $t$ first data points.
    In other applications, the sequence is artificial, and designed only to
    interpolate between $\pi_0$, a distribution that is easy to sample from, and
    $\pi_T=\pi$, the distribution of interest. A popular interpolation strategy is
    tempering (also known as annealing), where $\pi_t \propto
    q^{(1-\lambda_t)} \pi^{\lambda_t}$, and the exponent $\lambda_t$ grows
    from $0$ to $1$.

    SMC samplers have gained popularity in recent years for a variety of reasons.
    First, they provide at no extra cost an estimate of the normalising constants
    of the $\pi_t$; these quantities are useful in particular in Bayesian model
    choice. Second, the most expensive steps of SMC samplers are `embarrassingly
    parallel', making it possible to leverage parallel hardware for efficient
    computation. Third, the fact that SMC samplers carry forward a particle sample
    approximating the current target $\pi_t$ makes it easier to automatically calibrate
    tuning parameters, such as those of the Markov kernels, to obtain
    optimal performance.

    \subsection{Waste-free versus standard Sequential Monte Carlo}

    Assume that  \[ \pi_t(\dd x) = \frac{1}{Z_t}g_t(x)\nu(\dd x) \] where $\nu(\dd x)$ is a
    dominating probability measure, $g_t:\ambientspace\to\bbRp$ may be evaluated
    pointwise, and $Z_t\in(0, +\infty)$ is a possibly intractable normalising
    constant. Let $G_t \coloneq g_t / g_{t-1}$ for $t\geq 1$, $G_0\coloneq g_0$,
    and let $M$, $P\geq 1$ be integers.

    \Cref{alg:smc,alg:wf_smc}  describe the two types of SMC samplers we are
    considering in the paper. They essentially perform the same operations: at
    iteration $t\geq 1$, $M$ Markov chains of length $P$ are generated,
    $X_t^{m,p}$, $p=1,\dots,P$, using a $\pi_{t-1}$-invariant Markov kernel $K_t$.
    (At time $0$, the $X_0^{m,p}$ are generated independently from $\nu$.) Some
    of the resulting particles are reweighted according to function $G_t$, and
    resampled. The resampled particles are then used as starting points of the $M$
    Markov chains generated at time $t+1$, and so on.

    \begin{algorithm}[!ht]
        \DontPrintSemicolon
        \caption{Standard SMC}
        \label{alg:smc}
        \SetKwInOut{Input}{input}
        \SetKwInOut{Output}{output}
        \

        \Input{Integers $M, P\geq 1$}
        \For{$t\gets 0$ \KwTo $T$}{
            \If{$t=0$}{
                $X_0^{1:M}\sim \nu^{\otimes M}$\;
                $\widehat{Z}_{-1} \gets 1$\;
            }\Else {
                $A_t^{1:M}\sim \text{resample}(M, W_{t-1}^{1:M})$
                \tcp*{IID draws from $\mathrm{Cat}(W_{t-1}^{1:M})$}\;
                \For{$m\gets 1$ \KwTo $M$}{
                    $\tilde{X}_{t}^{m, 1}\gets X_{t-1}^{A_t^{m}}$\;
                    \For{$p\gets 2$ \KwTo $P$}{
                        $\tilde{X}_t^{m, p}\sim K_t(\tilde{X}_t^{m, p-1}, \dd x)$\;
                    }
                    $X_t^m\gets \tilde{X}_t^{m,P}$  \tcp*{discarding intermediate samples $\tilde{X}_t^{1:M, 1:P-1}$}
                }
            }
            \For{$m\gets 1$ \KwTo $M$}{
                $w_t^m\gets G_t(X_t^m)$\;
            }
            \For{$m\gets 1$ \KwTo $M$}{
                $W_t^m \gets w_t^m/\sum_{p=1}^{M}w_t^{p}$\;
            }
            $\widehat{Z}_t \gets \widehat{Z}_{t-1} \times
            \left\{\frac{1}{M}\sum_{m=1}^M w_t^m\right\}$\;
        }
    \end{algorithm}
    \begin{algorithm}[!ht]
        \DontPrintSemicolon
        \caption{waste-free SMC}
        \label{alg:wf_smc}
        \SetKwInOut{Input}{input}
        \SetKwInOut{Output}{output}
        \

        \Input{Integers $M, P\geq 1$, $N\gets MP$}
        \For{$t\gets 0$ \KwTo $T$}{
            \If{$t=0$}{
                $X_0^{1:N}\sim \nu^{\otimes N}$\;
                $\widehat{Z}_{-1} \gets 1$\;
            }\Else {
                $A_t^{1:M}\sim \text{resample}(M, W_{t-1}^{1:N})$
                \For{$m\gets 1$ \KwTo $M$}{
                    $X_{t}^{m, 1}\gets X_{t-1}^{A_t^{m}}$\;
                    \For{$p\gets 2$ \KwTo $P$}{
                        $X_t^{m,p}\sim K_t(X_t^{m,p-1}, \dd x)$
                    }
                    Gather $X_{t}^{1:M,1:P}$ to form $X_{t}^{1:N}$ \tcp*{keeping
                    all intermediary samples}
                }
            }
            \For{$n\gets 1$ \KwTo $N$}{
                $w_t^n\gets G_t(X_t^n)$\;
            }
            \For{$n\gets 1$ \KwTo $N$}{
                $W_t^n \gets w_t^n/\sum_{p=1}^{N}w_t^p$\;
            }
            $\widehat{Z}_t \gets \widehat{Z}_{t-1} \times
            \left\{\frac{1}{N}\sum_{n=1}^N w_t^n\right\}$\;
        }
    \end{algorithm}

    The key difference between the two algorithms is how they define the pool of
    candidates for the reweighting and resampling steps.
    In standard SMC, only the end points $X_t^{m,P}$ of the chains are reweighted and
    resampled. In waste-free SMC, all the $N=MP$ iterates of these chains are
    reweighted; then $M$ particles are resampled out of these $N$ candidates.

    \citet{MR4400392}, who introduced waste-free SMC samplers, show in their numerical
    experiments that they tend to outperform standard SMC samplers. Our objectives
    are to derive finite sample properties for both types of SMC samplers.

    \subsection{Summary of main results}\label{sub:motiv}

    As in~\citet{MR4630952,marion2025finitesampleboundssequential}, we are
    interested in deriving conditions on $M$ and $P$ that guarantee a certain
    error level for moment estimates, i.e.,
    \begin{equation}
        \label{eq:garantee}
        \left|\widehat{\pi}_{T-1}(f) - \pi_{T-1}(f) \right|
        < \varepsilon
    \end{equation}
    with probability $1-\eta$, for given $\varepsilon>0$ and $\eta\in(0, 1)$
    for some functions $f:\ambientspace \to \bbR$, or, alternatively, for normalising
    constant estimates:
    \begin{equation}
        \label{eq:garantee_Z}
        \left|\frac{\widehat{Z}_T}{Z_T} - 1 \right| < \varepsilon
    \end{equation}
    again with probability $1-\eta$. There,
    $\widehat{\pi}_{T-1}(f)$ is an estimate of $\pi_{T-1}(f)$ that may be
    computed at the last iteration $T$ of the algorithm;
    e.g. $\widehat{\pi}_{T-1}(f) = N^{-1} \sum_{n=1}^{N} f(X_{T}^n)$ for
    \cref{alg:wf_smc}
    and $\widehat{Z}_T$ is an estimate of the normalising constant $Z_T$,
  generated by either \cref{alg:smc} or \cref{alg:wf_smc}.
    (For the sake of simplicity, we will not derive bounds for weighted estimates,
    i.e. $\sum_{n=1}^N W_T^n \varphi(X_T^n)$.)
    Of course, one may obtain similar guarantees for intermediate quantities
    $\pi_{t-1}(f)$ and $Z_t$ by pretending that the algorithm is stopped after
    iteration $t$, and replacing $T$ by $t$ in all the stated results.

    As displayed in Table~\ref{tab:smc_comparison}, the complexity bounds obtained
    in this paper depend on the type of sequence of distributions (arbitrary, or
    tempering), the type of the estimates, i.e.~\eqref{eq:garantee} or
    ~\eqref{eq:garantee_Z}, and the assumptions on the Markov kernels $K_t$.
    Complexity (a.k.a. query complexity) means number of Markov steps in our
    context, i.e. the product $TMP$ for \cref{alg:smc,alg:wf_smc}. We would obtain
    the same bounds if we were to define complexity as the number of evaluations of
    a function $g_t$, in case the $K_t$ are Metropolis kernels.

    \begin{table}[h]
        \centering
        \small
        \setlength{\tabcolsep}{4pt}
        \begin{NiceTabular}{c|c|c|c|c}
            \textbf{Sequence} & \textbf{Estimate} & \makecell{\textbf{standard SMC}\\($\tau_\beta$)} & \makecell{\textbf{waste-free SMC}\\($\gamma$)} & \makecell{\textbf{greedy}\\\textbf{waste-free SMC}\\($\gamma$)} \\
            \hline
            \Block{2-1}{\textbf{Arbitrary}} &
            $\widehat{\pi}_{T-1}(f)$
            & $\frac{T}{\varepsilon^2}\log\left(\frac{T}{\eta}\right)\tau_\beta\left(\frac{\eta\varepsilon^2}{T\log(T/\eta)}\right)$ &
            \newres $\frac{T}{\gamma \varepsilon^2} \log\left(\frac{T}{\eta}\right)$
            & \newres $\frac{T + \varepsilon^{-2}}{\gamma}\log\left(\frac{T}{\eta}\right)$ \\
            \cline{2-5}
            & $\widehat{Z}_T$
            & \newres $\frac{T^2}{\varepsilon^2}\tau_{\beta}\left(\frac{\varepsilon^2}{T^2}\right)$
            & \newres $\frac{T^3}{\gamma \varepsilon^2}$
            &
            \\
            \hline
            \Block{2-1}{\textbf{Tempering}} &
            $\widehat{\pi}_{T-1}(f)$
            & $\frac{d^{1/2}}{\varepsilon^2}\tau_\beta\left(\frac{\varepsilon^2}{d^{1/2}}\right)$
            & \newres $\frac{d^{1/2}}{\gamma\varepsilon^2}$
            & \newres $\frac{1}{\gamma}\left(  d^{1/2}
                                           + \frac{1}{\varepsilon^2} \right)$ \\
            \cline{2-5}
            & $\widehat{Z}_T$ &
            \obviouscor $\frac{d}{\varepsilon^2}\tau_\beta\left(\frac{\varepsilon^2}{d}\right)$
            & \newres $\frac{d^{3/2}}{\gamma\varepsilon^2}$
            & \\
            \hline
            \makecell{\textbf{Tempering}\\(ULMC + MALA)} & $\widehat{Z}_T$
            & \newres $\frac{d^{3/2}}{\varepsilon^2}$ & & \\
            \CodeAfter
            \tikz \draw (3-|5) -- (4-|6) (3-|6) -- (4-|5)
            (5-|5) -- (6-|6) (5-|6) -- (6-|5)
            (6-|4) -- (7-|5) (6-|5) -- (7-|4)
            (6-|5) -- (7-|6) (6-|6) -- (7-|5);
        \end{NiceTabular}
        \caption{Summary of $\OO(\cdot)$ complexities for SMC samplers in terms of $T$
            (number of target distributions), $d$ (ambient dimension, $\ambientspace=\bbR^d$), and mixing properties of the Markov kernels:
            $\tau_\beta$ ($\beta$-mixing time) for standard SMC, $\gamma$ (spectral gap) for waste-free SMC.
            The two quantities are not directly comparable (Section~\ref{sec:spectral_gap_limitation}).
            The last row instantiates the standard SMC bound for log-concave targets with $\tau_\beta=\tilde\OO(\sqrt{d})$ (Section~\ref{subsec:warm_start_logconcave}).
            The background color distinguishes new results (blue, general
            results, or green, direct corollary of more general results) from results
            of~\cite{MR4630952, marion2025finitesampleboundssequential} (no color).
            Log factors
            have been omitted (except in the first row).}
        \label{tab:smc_comparison}
    \end{table}

    Let us focus first on moment estimates.
    For an arbitrary sequence of length $T$, we obtain (Theorem~\ref{th:main_wf}) a complexity for waste-free SMC comparable
    to that of standard SMC up to the mixing properties of the kernels. Results for waste-free SMC are derived
    under an $\mathcal{L}^2$ spectral gap assumption
    $\geq\gamma>0$; existing SMC results~\citep{MR4630952, marion2025finitesampleboundssequential} instead require the
    kernels to mix rapidly in $\TV$ from starting points drawn from a warm distribution, as measured by the $\beta$-mixing time $\tau_\beta$.
    The latter is typically implied by uniform pointwise contraction in $\TV$ (i.e., geometric ergodicity),
    say at rate $1-\rho<1$, with $\rho\leq\gamma$ in general~\citep{10.1214/ECP.v2-981, 2afca349-66e5-3cf9-b028-1997198e7630}, but can also be
    established directly, sometimes beating what the spectral gap gives (Section~\ref{sec:spectral_gap_limitation}).

    Again, for moment estimates, we show it is possible to
    design a greedy variant of waste-free SMC, where $P$ is constant at times
    $t<T$, and then scales like $\OO(\varepsilon^{-2})$ at iteration $T$ (Theorem~\ref{th:wf_smc_greedy}).
    In this way, the first $T$ iterations serve as a warm-start provider for the final long MCMC chain started near stationarity,
    whose cost $\tilde\OO(\gamma^{-1}\varepsilon^{-2})$ cannot be improved in general ($\tilde\OO$ for ignoring log factors).
    For tempering, we show we can take $T=\Theta\left(d^{1/2}\right)$
    (Theorem~\ref{th:chi_log_concave},
    extending~\citealt{marion2025finitesampleboundssequential}, Corollary~1),
    and, using the greedy variant, we obtain complexity $\tilde\OO(\gamma^{-1}(d^{1/2}+\varepsilon^{-2}))$ for moments with respect to the target distribution $\pi$
    rather than $\tilde\OO(d^{1/2}\gamma^{-1}\varepsilon^{-2})$.

    The rest of the paper is devoted to normalising constant estimates.
    Possibly the sharpest result (Theorem~\ref{th:normalising_smc}) is for tempering: for
    log-concave and log-smooth targets, standard SMC with composite
    ULMC (Underdamped Langevin Monte Carlo) and MALA (Metropolis-adjusted Langevin) kernels
    achieves complexity $\tilde\OO(d^{3/2}\varepsilon^{-2})$.
    We also consider a product-of-medians estimator
    $\hat{Z}^{\textup{med}}_T$ obtained from independent copies of waste-free SMC (Theorem~\ref{th:boosted}).

    We shall discuss the practical implications of these results to end users
    at the end of the paper (\cref{sec:discussion}).

    \subsection{Related work}

    SMC samplers were introduced by~\citet{MR2278333} as a generalisation of the
    samplers of~\citet{neal1998annealedimportancesampling} and~\citet{MR1929161}.
    For an overview of SMC samplers and their numerous applications, see~\citet[Chap.~17]{Chopin2020} or~\citet{invitation_smc}.

    Asymptotic convergence results, including
    consistency and central limit theorems as the number of particles goes to
    infinity,
    are available for both standard SMC~\citep{DelMoral2004,mr2153989}
    and waste-free SMC~\citep{MR4400392}.
    However, these results are largely non-explicit with respect to the specific
    problem features, namely the target distributions, the ambient dimension, or
    the mixing properties of the kernels. As a consequence, they offer limited
    guidance on how to select the particle system size in practice, or how to compare SMC
    performance with alternative sampling schemes under finite computational
    budget. Another approach is to study the stability of SMC samplers
    as $d\to\infty$~\citep{Beskos2014stability,
        Beskos2014_normconstants}.

    Non-asymptotic finite-sample error bounds for $\pi(f)$ for SMC were recently developed by~\citet{MR4630952}.
    Bounds derived in that paper exhibited unfavourable dependence on ratios of normalising constants;
    sharper bounds were derived in~\citet{marion2025finitesampleboundssequential}.
    We restate their results and discuss proof techniques that are relevant to the
    waste-free setting in Section~\ref{sec:smc}.


    As for the non-asymptotic guarantees on the normalising constant estimate
    $\hat{Z}_T$, the Gaussian-like guarantees derived
    in~\citet{marion2025finitesampleboundssequential} for moment estimation are
    insufficient to derive sharp complexity rates for $\hat{Z}_T$ as the
    weights are typically heavy-tailed. There also exist bounds on relative
    variance of $\hat{Z}_T$~\citep{AIHPB_2011__47_3_629_0,WHITELEY20121840} or
    concentration inequalities for $\log
    \hat{Z}_T/Z_T$~\citep{10.3150/14-BEJ680}, building on the inequalities
    derived in \cite{DelMoralRio2011}. These bounds are derived using
    Feynman-Kac semigroup formalism with strong assumptions on $(G_t, K_t)$.
    For example,~\citet{10.3150/14-BEJ680} assume a uniform Dobrushin
    contraction together with bounded weights oscillations, conditions that are
    essentially unavailable in most practical cases.

    SMC-like annealing ideas are central to randomised approximation schemes for estimating the volume of a convex body (e.g. a compact convex set in $\bbR^d$),
    including classical approaches based on sequences of intermediate truncated Gaussian distributions~\citep{4031343,doi:10.1137/15M1054250,10.1145/3795687}.

    Finite-sample complexity bounds for estimating normalising constants of log-concave distributions were established by~\citet{10.1214/18-EJS1411}, yielding polynomial-time complexities of the form
    $\mathcal{O}(\mathrm{poly}(d,\varepsilon^{-1}))$.
    The dependence on the dimension varies with the smoothness and curvature assumptions made on the potential, ranging from $d^3$
    for strongly log-concave and smooth targets to $d^{5/2}$ under the additional Lipschitz Hessian assumption, and higher complexities
    in the (non-strongly) log-concave case.
    The dependence on the accuracy $\varepsilon$ is typically non-optimal, scaling between $\varepsilon^{-3}$ and $\varepsilon^{-4}$.
    In a different line of work,~\citet{ge2020estimatingnormalizingconstantslogconcave} establish the complexity lower bound $\Omega(d^{1-\oo(1)}\varepsilon^{-2})$
    for this class of distributions, and propose combining multilevel Monte Carlo
    with Gaussian annealing to reach $\tilde\OO(d^{4/3}\varepsilon^{-2})$.

    Tempering strategies for warm-starting MCMC have been analysed in the literature under the name annealed MCMC~\citep{Guo2025,chehab2025provable}.
    Principled approaches to designing the tempering schedule motivated by the Gaussian case suggest the optimal schedule is a geometric sequence
    with horizon $T = \Theta(d^{1/2})$~\citep{10.5555/3692070.3692421}, while~\citet{invitation_smc} have proven the general polynomial dependency of the bridging length $T$ under a suitable functional inequality on $\pi$.
    Other algorithmic warm-up strategies based on non-adjusted kernels are available for log-concave sampling and come with better
    complexity bounds for ULMC or HMC~\citep{AltChe24HighAcc, zhang2026algorithmicwarmstartshamiltonian}.

    There is also some work to tackle the thorny case of SMC for
    multimodal distributions~\citep{schweizer2012nonasymptoticerrorboundssequential,Paulin2019,10.1214/23-AAP1989,
        lee2026convergenceboundssequentialmonte}.
    However, those results typically rely on strong assumptions: the target admits a mixture decomposition over disjoint sets,
    or the ratios between successive normalised densities are uniformly bounded.
    Within tempered SMC with Langevin diffusion,~\citet{han2026timecomplexitysamplingmultimodaldistribution}
    show the diffusion only has to mix within each well because reweighting and resampling already move mass between wells,
    but this requires the wells to have non-negligible mass at every temperature.
    In general, sampling from a non-log-concave distribution is provably hard~\citep{pmlr-v291-he25a}.
    We leave this direction aside.

    Finally, so-called sequential MCMC, or MCMC-PFs (MCMC particle filters) are
    close in spirit to SMC samplers, but their study seems to require a distinct
    approach; see \cite{Finke_Doucet_Johansen_2020} and references therein.

    \subsection{Organization of the paper}

    In Section~\ref{sec:smc},
    we recall existing results for the finite-sample complexity of standard SMC for estimating expectations.
    In Section~\ref{sec:main}, we present our main finite-sample guarantees for waste-free SMC for estimating expectations and normalising constants.
    We also discuss the need for concentration bounds dependent on the relative variance of the reweighting functions
    to derive sharp bounds on the normalising constant estimators.
    Section~\ref{sec:proof} outlines the proof strategy, highlighting the coupling strategy and the concentration bounds used to handle correlations
    within chains for waste-free SMC.
    Section~\ref{sec:application} specialises the general results to geometric tempering paths for log-concave targets, with an attention
    to dimension and condition number dependencies.
    Section~\ref{sec:spectral_gap_limitation} discusses the applicability of the concentration bounds for the normalising constant
    estimator of standard SMC.
    In particular, we show how trading the spectral gap assumption for a pointwise-start mixing time bound
    can yield improved complexity guarantees for standard SMC with fast-mixing kernels.
    \Cref{sec:discussion} discusses how our findings may translate into practical
    guidance for end users.
    Postponed proofs and auxiliary technical results used throughout the paper are given in Appendix~\ref{sec:proof_technical}.

    \subsection{Notation and definitions}

    Let $(\ambientspace, \ambientalgebra)$ be a measurable space, $f:\ambientspace \to \bbR$ a measurable
    function, and $\norm{f}_{\infty}$ the supremum norm.
    Results related to tempering assume that $\ambientspace \subset \bbR^d$, with
    $d\geq 1$.
    We write $a\lesssim b$ (or $a=\OO(b)$) if $\abs{a}\le Cb$ for some $C>0$, $a\gtrsim b$ (or $a=\Omega(b)$) if $a\ge c \abs{b}$ for some $c>0$, and $a\asymp b$ (or $a=\Theta(b)$) if both hold.
    The notation $\tilde\OO(f)$ hides logarithmic factors: $\tilde\OO(f)=\OO(f\log^{\OO(1)} f)$.
    The median of $J$ real numbers $x_1,\ldots, x_J$ is denoted by $\textup{median}((x^{(j)})_{j=1,\ldots,J})$.
    We denote by $\succcurlyeq$ (resp. $\succ$) the partial order on the positive semi-definite (resp. definite) matrices, i.e. for $A, B$ two symmetric matrices, $A\succcurlyeq B$ (resp. $A\succ B$) means that $A-B$ is positive semi-definite (resp. definite).
    For a probability measure $\nu$, let $\nu(f)=\int f \dd \nu$ and $\Var_{\nu}[f]$ the variance of $f$ under $\nu$.
    We denote by $\mathcal{B}_r(x)$ the Euclidean ball of $\bbR^d$ centered at $x$ with radius $r$, and $\mathcal{B}_r=\mathcal{B}_r(0)$.
    For probability measures $\mu,\pi$, $\mu \otimes \pi$ is the product measure, $\mu\ll \pi$ denotes absolute continuity.
    The $\chi^2$-divergence (for $\mu\ll \pi$) and total variation are
    \begin{equation*}
        \chi^2(\mu\mid \pi)=\int \left(\frac{\dd\mu}{\dd\pi}-1\right)^2 \dd \pi, \quad
        \textup{TV}(\mu, \pi)=\sup_{A\in\ambientalgebra}|\mu(A)-\pi(A)|.
    \end{equation*}
    We say $\mu$ is $\omega$-warm in $\mathcal{L}^\infty(\pi)$-norm if $\sup_{A:\pi(A)>0}\mu(A)/\pi(A)\le \omega$, and in $\mathcal{L}^2(\pi)$-norm if $\chi^2(\mu\mid\pi)\le \omega$; we omit the norm when it is clear
    from the context.
    A Markov kernel $K(x,\dd y)$ is a map $K:\ambientspace\times \ambientalgebra\to[0,1]$, such that for any $x\in \ambientspace$, $K(x,\dd y)$ is a probability measure.
    It acts from the right on measures with $\mu K(A)=\int K(x,A)\mu(\dd x)$, and from the left on functions with $Kf(x)=\int K(x,\dd y) f(y)$. $K$ leaves $\pi$ invariant if $\pi K=\pi$, and is $\pi$-reversible if $\pi(\dd x)K(x,\dd y)=\pi(\dd y)K(y,\dd x)$.
    Let $\mathcal{L}^2_\pi$ be the set of square-integrable functions under $\pi$ equipped with the usual inner product $\langle f,g\rangle=\int f g \dd \pi$, $\norm{\cdot}_{\pi}$ the induced norm, and $\mathcal{L}^2_{0,\pi}$ the zero-mean subspace.
    A $\pi$-reversible Markov kernel $K$ has (absolute) spectral gap $\gamma = 1-\lambda>0$ if $\lambda = \sup\{\norm{Kh}_{\pi} : \norm{h}_{\pi} = 1, h\in \mathcal{L}^2_{0,\pi}\}<1$.
    For $\xi>0$ and measure $\mu$, the mixing time in $\text{TV}$-distance is
    \begin{equation*}
        \tau(\xi,\mu,K)=\min\{n : \textup{TV}(\mu K^n,\pi)\leq \xi\},
    \end{equation*}
    where $\pi$ denotes the stationary distribution of $K$, while the $\beta$-mixing time is
    \begin{equation*}
        \tau_{\beta}(\xi, \mu, K)=\min\left\{n: \int \textup{TV}(\delta_x K^n, \pi)\mu(\dd x)\leq \xi\right\},
    \end{equation*}
    and in general $\tau_{\beta}\geq \tau$.
    Using warmness $\omega$, we write $\tau(\xi,\omega,K)=\sup_{\mu\text{ is }\omega\text{-$\mathcal{L}^{\infty}$-warm}}\tau(\xi,\mu,K)$, and similarly for $\tau_\beta$.
    For the sake of notation consistency, we let $\pi_{-1}=\nu$ and $Z_{-1}=1$.

    \section{Known results and assumptions\label{sec:smc}}

    We recall the finite sample complexity results of
    \citet{MR4630952, marion2025finitesampleboundssequential}.
    \begin{assumption}
        \label{ass:chi}
        There exists a constant $\bar{\chi}^2<\infty$ such that the sequence of distributions $(\pi_t)$ satisfies $1+\chi^2(\pi_t\mid \pi_{t-1})\leq \bar{\chi}^2$, for all $0\leq t\leq T$.
    \end{assumption}

    For the sake of notation, we write $\tau_\beta(\xi, \omega)=\max_{t=1,\ldots, T}\tau_\beta(\xi, \omega, K_t)$ the uniform $\beta$-mixing times
    over kernels $K_1,\ldots, K_T$, and $\tau_\beta(\xi) = \tau_\beta(\xi, \OO(1))$.

    Except where stated otherwise, every test function $f:\ambientspace \to
    \bbR$ is assumed to be such that  $\norm{f}_\infty=\OO(1)$, with a constant
    that does not depend on $T$ or $d$.

    The theorem below is a straightforward generalisation of Th.~3.1 of
    \cite{MR4630952} and Th.~1 of
    \cite{marion2025finitesampleboundssequential} to an arbitrary $\eta$ (while the
    original theorems assume $\eta=1/4$). The proof is essentially the same,
    and is omitted.

    \begin{theorem}
        \label{th:main_thm}
        Assume~\ref{ass:chi}.
        Fix $T\geq 1$, $\varepsilon>0$, $\eta\in (0, 1)$, and let $(M, P)$ satisfy
        \begin{equation*}
            \label{eq:condition_n_p_r_}
            \begin{split}
                M \geq \log(8T\eta^{-2}) \max\left(18\bar{\chi}^2,
                                                 \varepsilon^{-2}/2\right),\indent P\geq
                \tau_\beta(\eta/(2 M T), 2).
            \end{split}
        \end{equation*}
        Then for any $f:\ambientspace\to\bbR$ with $\abs{f}\leq 1$,
        \begin{equation*}
            \label{eq:concentration}
            \Pr(\abs{\widehat{\pi}_{T-1}(f)-\pi_{T-1}(f)}< \varepsilon)\geq 1-\eta,
        \end{equation*}
        where $\widehat{\pi}_{T-1}(f) \coloneq M^{-1}\sum_{m=1}^M f(X_{T}^{m})$, and $X_T^{1:M}$ are the particles produced
        by standard SMC (Algorithm~\ref{alg:smc}).
    \end{theorem}
    \begin{remark}
        Since $M$ appears in the lower bound on $P$, the query complexity is not
        tractable without further assumptions on the Markov kernels $K_t$.
        A standard (albeit strong) assumption that implies a bound on
        $\tau_\beta$ is (uniform) geometric ergodicity:
        there exist $0<\rho \leq 1$, $C>0$ such that for any $t=1,\ldots, T$,
        \begin{equation*}
            \textup{TV}(\delta_x K^n_t, \pi_{t-1})\leq C(1-\rho)^n,\qquad \forall x\in \ambientspace,
        \end{equation*}
        in which case $\tau_{\beta}(\xi, \omega)
        =\OO(\rho^{-1}\log(C\xi^{-1}))$ for sufficiently small $\xi>0$;
        then $M$ appears only as a logarithmic factor in the lower bound on $P$; 
        a weaker version of this assumption is obtained by allowing $C$ to depend on $x$, in which case $\tau_\beta(\xi, \omega, K_t)=\OO(\rho^{-1}\log(\omega \pi_{t-1}(C)\xi^{-1}))$.
        However, outside very specific cases, neither sharp estimates of the rate $\rho$ nor of $C$
        are available in closed form.
    \end{remark}

    \section{Main results} \label{sec:main}

    \subsection{Moments}\label{sub:moments}
    All proofs are postponed to Section~\ref{sec:proof}.
    \begin{assumption}
        \label{ass:spg}
        For any $1\leq t\leq T$, $K_t$ is $\pi_{t-1}$-reversible and admits a spectral gap $\gamma_t>0$.
    \end{assumption}
    We let $\gamma = \min_{t=1,\ldots, T} \gamma_t$ be the minimal spectral gap,
    and $\gamma_0\coloneq 1$.
    For reversible kernels, a spectral gap is equivalent to geometric ergodicity with a worst pointwise $\textup{TV}$-rate
    in general~\citep[Theorem 2.1; Theorem 2][]{10.1214/ECP.v2-981, 2afca349-66e5-3cf9-b028-1997198e7630}.
    In particular, the ergodicity assumption implies Assumption~\ref{ass:spg} with $\gamma_t\geq \rho_t$, which in turn
    implies a bound on the warm-mixing time $\tau$ but does not transfer to $\tau_\beta$.

    \begin{assumption}
        \label{ass:chisquare_is_small}
        Assumption~\ref{ass:chi} is satisfied with $\bar{\chi}^2=2$.
    \end{assumption}
    The feasibility of this assumption is discussed in Sections~\ref{sec:application} and~\ref{sub:practical}.

    The theorem below is the counterpart of Theorem~\ref{th:main_thm} for
    waste-free SMC.
    \begin{theorem}
        \label{th:main_wf}
        Assume~\ref{ass:spg} and~\ref{ass:chisquare_is_small}.
        Fix $M\geq 1$, and take
        \begin{equation*}
            P \geq
            \max\left(\frac{128}{\gamma}\log\left(\frac{128MT}{\eta}\right),
                    \frac{4}{\gamma \varepsilon^2}\log\left(\frac{256T}{\eta}\right)\right).
        \end{equation*}
        Then for any $f$ with $\abs{f}\leq 1$,
        \begin{equation*}
            \Pr(\abs{\widehat{\pi}_{T-1}(f)-\pi_{T-1}(f)}<\varepsilon)\geq 1-\eta,
        \end{equation*}
        where $\widehat{\pi}_{T-1}(f)\coloneq N^{-1}\sum_{n=1}^N f(X_T^n)$, and $X_T^{1:N}$ are the particles produced by waste-free SMC (Algorithm~\ref{alg:wf_smc}).
    \end{theorem}
    The previous theorem yields the following complexity upper bound for waste-free SMC:
    \begin{equation}
        \label{eq:total_cost_wfsmc}
        \OO\left(\max\left(\frac{MT}{\gamma \varepsilon^{2}}\log\left(\frac{T}{\eta}\right), \frac{MT}{\gamma}\log\frac{MT}{\eta}\right)\right).
    \end{equation}
    Unless stated otherwise, we assume the number of parallel
    chains $M$ does not grow with the parameters of the problem; i.e. $M=\OO(1)$.
    In this case, the first term in~\eqref{eq:total_cost_wfsmc} prevails.

    We now consider a greedy variant of waste-free SMC, where the length of the
    chains, $P$ is allowed to vary across iterations, see~\cref{alg:smc_size}.

    \begin{algorithm}[!ht]
        \DontPrintSemicolon
        \caption{Greedy waste-free SMC}
        \label{alg:smc_size}
        \SetKwInOut{Input}{input}
        \SetKwInOut{Output}{output}
        \

        \Input{Integers $P_0,\ldots, P_T\geq 1$, $M\geq 1$}
        \For{$t\gets 0$ \KwTo $T$}{
            $P\gets P_t$\;
            /* Same operations as in the loop of~\cref{alg:wf_smc}*/\;
        }
    \end{algorithm}
    The next theorem shows that the greedy variant makes it possible to reduce the
    complexity further, to
    \begin{equation*}
        \OO\left(\frac{T}{\gamma}\log\left(\frac{T}{\eta}\right) + \frac{1}{\gamma\varepsilon^2}\log\left(\frac{T}{\eta}\right)\right).
    \end{equation*}
    Note in particular that, in the $\varepsilon\ll 1$ regime, the leading term
    scales logarithmically in $T$ (rather than linearly).
    \begin{theorem}
        \label{th:wf_smc_greedy}
        The guarantees of~\cref{th:main_wf} hold for Algorithm~\ref{alg:smc_size} as soon as
        \begin{equation}
            \label{eq:wf_smc_greedy_requirement}
            P_{0:T-1} \geq  \frac{128}{\gamma}\log\left(\frac{128MT}{\eta}\right),
            \indent P_T \geq \frac{4}{\gamma\varepsilon^2}\log\left(\frac{256T}{\eta}\right).
        \end{equation}
    \end{theorem}
    In words, in order to estimate accurately (i.e. $\varepsilon\ll 1$) moments relative to
    $\pi_{T-1}$, only $P_T$ needs to scale like $\varepsilon^{-2}$, while
    the previous $P_t$, $t<T$, may stay constant (relative to $\varepsilon$).
    We interpret the first $T$ iterations of Algorithm~\ref{alg:smc_size} as a warm-start provider
    for the final iteration, where each of the $M$ long MCMC chains is started near stationarity.

    \subsection{Normalising constant}
    \label{subsec:normalising_constant}
    We now turn to the problem of estimating the normalising constant $Z_T$.
    Sequential Monte Carlo samplers readily provide the user with an estimate of the ratio $Z_t/Z_{t-1} = \pi_{t-1}(G_t)$.
    As a consequence, an estimate for the final normalising constant $Z_{T}$ is
    given by (since $Z_{-1}=1$):
    \begin{equation*}
        \widehat{Z}_T= \prod_{t=0}^{T} \hat{\pi}_{t-1}(G_{t}).
    \end{equation*}
    Our objective is to determine how many Markov steps are required for
    $\widehat Z_T$ to achieve a prescribed relative accuracy,
    $\abs{\widehat Z_T/Z_T - 1} < \varepsilon$,
    with probability at least $1-\eta$.

    A first, naive approach consists in directly combining the concentration
    bounds obtained in the previous section with a union bound over
    $t=0,\ldots,T$.
    Indeed, bounding the error of $\widehat Z_T/Z_T$ reduces to bounding
    each factor $\widehat\pi_{t-1}(G_t)/\pi_{t-1}(G_t)$ with accuracy of order
    $\varepsilon/T$ and probability $\eta/T$.
    However, the concentration results of the previous sections rely on a boundedness assumption on the functional to derive sub-Gaussian like concentration,
    but in general, the normalised weights
    \begin{equation*}
        \frac{G_t(X_t)}{\pi_{t-1}(G_t)}, \indent X_t\sim \pi_{t-1},
    \end{equation*}
    need not be bounded, and may in fact fail to be sub-Gaussian.
    This typically happens when the target distribution has heavier tails than the base measure.
    We give below a concrete example.
    \begin{example}
        \label{example:special_case_gaussian}
        Let $q$ be the density of a standard Gaussian distribution $\mathcal{N}(0, I_d)$, and $\pi(\dd x) = \mathcal{N}(0, \sigma^2 I_d)$ with $\sigma^2 > 1$.
        Let $\lambda_{-1} = 0$, and $\lambda_t = \min(1, \lambda_{t-1}+\delta/\sqrt{d})$ be a linearly increasing schedule for any fixed $\delta > 0$ with $d>\delta^2(\sigma^2-1)^2$
        and let $\pi_t\propto q^{1-\lambda_t}\pi^{\lambda_t}$ be the
        associated geometric tempering sequence.
        Then for any $t$, $G_t(X)/\pi_{t-1}(G_t)$ with $X\sim \pi_{t-1}$ is not sub-Gaussian, but has finite variance $\chi^2(\pi_t \mid \pi_{t-1})$.
    \end{example}
%

    If we only assume that the variance of the normalised weights exists (i.e., $\chi^2(\pi_t\mid\pi_{t-1})$ is finite for all $0\leq t\leq T$), then
    the bound implied by Bienaymé-Chebyshev's inequality, that is, with probability at least $1-\eta/T$,
    \begin{equation*}
        \left|\frac{N^{-1}\sum_{n=1}^N G_t(X^n)}{\pi_{t-1}(G_t)}-1\right|\leq \sqrt{ \frac{ T\chi^2(\pi_t\mid \pi_{t-1})}{N\eta}}, \indent X^{1:N}\sim \pi_{t-1}^{\otimes N},
    \end{equation*}
    is essentially the best we can hope for.
    Notice the exponentially worse dependency in the desired level of confidence compared to a sub-Gaussian regime.
    The next theorem addresses the insufficiency of the previous approach; in particular, it applies to
    Example~\ref{example:special_case_gaussian}.
    \begin{theorem}
        \label{th:strong_normalising_constant}
        Assume~\ref{ass:spg} and~\ref{ass:chisquare_is_small}.
        Let $\varepsilon\in (0,2)$, $M = 1$, and $P \geq \frac{20480T^2}{\gamma\varepsilon^2}$.
        Then waste-free SMC (Algorithm~\ref{alg:wf_smc}) returns $\widehat{Z}_T$ such that $\abs{\widehat{Z}_T/Z_T-1}< \varepsilon$ with probability at least $3/4$.
        In particular, it requires at most a total of
        \begin{equation*}
            \OO\left(\frac{T^3}{\gamma\varepsilon^2}\right)
        \end{equation*}
        Markov transitions for the previous bound to hold.
    \end{theorem}
    The theorem above is stated for $M=1$ but the same guarantee holds for $M\geq 1$, provided $P\geq 32\log(256MT)/\gamma$.

    The $3/4$ level of confidence can be boosted to $1-\eta$ at the cost of a logarithmic
    factor by using median-type estimators, which are known to exhibit sub-Gaussian performance~\citep{JERRUM1986169, Lugosi2019,10.1214/18-EJS1411}.
    For example, running approximately $\log(1/\eta)$ independent SMC samplers with the sample-size requirements of the previous theorem,
    and computing the median of the estimated normalising constants reduces the failure probability to $\OO(\eta)$.
    Similarly, one can replace each empirical mean factor in $\hat{Z}_T$
    by a median-of-means estimator, thereby ensuring per-factor high-probability guarantees.
    \begin{algorithm}[!ht]
        \DontPrintSemicolon
        \caption{Product-of-medians for $\widehat{Z}_T$}
        \label{alg:boosted}
        \SetKwInOut{Input}{input}
        Integers $J$, $M$ and $P$\;
        \SetKwInOut{Output}{output}
        \For{$j\gets 1$ \KwTo $J$}{
            Run waste-free SMC
            as given by Algorithm~\ref{alg:wf_smc}
            to obtain $\{\widehat{\pi}_{t-1}(G_t)^{(j)}\}_{t=0,\ldots,T}$\;
        }
        \For{$t\gets 0$ \KwTo $T$}{
            $\widehat{Z_t/Z_{t-1}}^{\textup{med}} \gets \textup{median}(\widehat{\pi}_{t-1}(G_t)^{(j)})_{j=1,\ldots, J}$\;
        }
        $\widehat{Z}^{\textup{med}}_T \gets \prod_{t=0}^T \widehat{Z_t/Z_{t-1}}^{\textup{med}}$\;
    \end{algorithm}
    The product-of-medians estimator $\widehat{Z}^{\textup{med}}_T$ returned by Algorithm~\ref{alg:boosted} with $J=\OO(\log(T/\eta))$
    satisfies $\abs{\widehat{Z}^{\textup{med}}_T/Z_T-1}< \varepsilon$ with probability at least $1-\eta$ at cost
    $\OO(T^3\gamma^{-1}\varepsilon^{-2}\log(T/\eta))$ (Theorem~\ref{th:boosted} in the appendix),
    that is, the same complexity as $\widehat{Z}_T$ up to a logarithmic factor.
    These median estimators exhibit greater robustness to heavy-tailed weights
    compared to $\hat{Z}_T$; see \cref{sec:discussion} for numerical evidence
    on this phenomenon. 

    \subsubsection{Lower bound}
    We now derive a lower bound on $P$ for a guarantee on $\widehat{Z}_T/Z_T$ to hold via an
    asymptotic argument based on a central limit theorem for $\log(\widehat{Z}_T/Z_T)$ derived by~\citet[][Theorem 2]{MR4400392}, and which we recall below.
    \begin{theorem}
        \label{th:clt}
        Let $M=\OO(1)$, and assume the Markov kernels are uniformly geometrically ergodic.
        Then the log-normalising constant $\log \widehat{Z}_T$ obtained by Algorithm~\ref{alg:wf_smc} satisfies the following convergence in law as $P\to\infty$,
        \begin{equation*}
            \sqrt{N}\log \left(\frac{ \widehat{Z}_T}{Z_T}\right)\xrightarrow{\mathcal{L}}\mathcal{N}(0, \sigma_{\infty}^2),\indent \sigma^2_{\infty} = \sum_{t=0}^T \sigma_{\infty}^2(G_t)/\pi_{t-1}(G_t)^2,
        \end{equation*}
        where $\sigma_{\infty}^2(G_t)$ is the asymptotic variance of the ergodic average $P^{-1}\sum_{p=1}^P G_t(\bar{X}_t^p)$, with $(\bar{X}_t^p)_{p\geq 1}$ a stationary Markov chain with kernel $K_t$.
    \end{theorem}
    This suggests that for $\Pr(\abs{\widehat{Z}_T/Z_T-1}< \varepsilon)$ to remain bounded away from $0$ as $\varepsilon\downarrow 0$ and $P\to \infty$, we must have
    $P=\Omega(\sigma_{\infty}^2/\varepsilon^2)$.
    Provided the reweighting functions are aligned with the slowest mode of the Markov kernels, we have $\sigma_{\infty}^2 = \Omega(T/\gamma)$, which gives
    the lower bound complexity $MTP=\Omega\left(T^2 \gamma^{-1} \varepsilon^{-2}\right)$.
    \begin{proposition}
        \label{prop:lower_bound}
        Let $M\geq 1$ be fixed and $\eta\in(0, 1]$.
        \begin{enumerate}[label=(\roman*)]
            \item Assume the kernels $(K_t)$ are uniformly geometrically ergodic. If $N_\varepsilon\to\infty$ and $N_\varepsilon\leq \sigma_\infty^2\varepsilon^{-2}\log(1/\eta)/4$ as $\varepsilon\to 0$, then waste-free SMC (Algorithm~\ref{alg:wf_smc}) with $N_\varepsilon$ particles returns $\hat{Z}_T$ such that
            \begin{equation*}
                \liminf_{\varepsilon\to 0}\Pr\left(\abs{\hat{Z}_T/Z_T-1}>\varepsilon\right)\geq \sqrt{\frac{2}{\pi}} \frac{\eta^{1/2}\sqrt{\log(1/\eta)}}{1+\log(1/\eta)}.
            \end{equation*}
            \item For any sequence of distributions $\nu,\pi_0,\ldots,\pi_T$ on $(\ambientspace, \ambientalgebra)$ and any $\gamma\in(0,1]$, there exists a sequence of uniformly geometrically ergodic kernels $(K_t)$ with common spectral gap $\gamma$ such that
            $\sigma^2_\infty\geq T\munderbar{\chi^2}/\gamma$, where $\munderbar{\chi^2} = \min_{t=0,\ldots, T-1}\chi^2(\pi_t\mid\pi_{t-1})$.
        \end{enumerate}
    \end{proposition}
    The gap of order $T$ between this lower bound and the upper bound remains an open problem.

    \section{Proof scheme}
    \label{sec:proof}

    \subsection{Overview}

    Before proceeding to describe the main proof techniques and steps of the
    theorems in Section~\ref{sec:main}, we summarise the proof strategy
    of~\citet{MR4630952} for the finite-sample analysis of standard SMC
    (Theorem~\ref{th:main_thm}), on which our proofs are built, and highlight
    the important differences.

    First, note that both standard SMC and waste-free SMC generate $M$ Markov
    chains of length $P$ at iteration $t$, using a $\pi_{t-1}$-invariant kernel
    $K_t$. If the $M$ starting points of these chains were sampled exactly from
    $\pi_{t-1}$, these algorithms would be trivial to analyse. The main difficulty
    is therefore to compare in some way the empirical resampling distribution from
    which these starting points are drawn with $\pi_{t-1}$.  \citet{MR4630952}
    achieve that by coupling the $M$ resampled particles with $M$ i.i.d.\ particles
    from $\pi_{t-1}$. In standard SMC, these resampled particles are drawn from an
    empirical distribution over the end points of the $M$ Markov chains generated
    at the previous iteration.
    Taking $P$ large enough relative to $\tau_\beta$ makes the coupling
    probability $\geq 1-\delta$. The coupling construction is then applied
    recursively.

    In waste-free SMC, the support of the resampling distribution contains all the
    $N=M\times P$ states of the $M$ Markov chains, rather than only the final
    states. Thus, to adapt the approach of~\citet{MR4630952}, we construct a coupling on
    the complete chains. In contrast to~\citet{MR4630952}, our coupling is
    performed marginally so that each chain is coupled to a stationary path at random time $R_t$,
    and the $M$ chains are not made independent conditionally on the coupling
    event.
    In exchange, the meeting
    probability depends only on the (warm)-mixing time $\tau$,
    which the spectral gap controls. We discard the factor $M$ but recover concentration
    of ergodic averages over the $P$
    states under the spectral gap assumption.
    This leads to two technical hurdles.

    First, the marginal distribution of the resampled particles will be warm,
    conditional on the sequence of previous meeting events $\cap_{t'<t}\{R_{t'}\leq r_{t'}\}$.
    This further imposes conditions on the meeting times for the warmness parameter
    to remain small.

    Second, as the ergodic averages are computed over non-stationary Markov chains,
    we will resort to concentration techniques for Markov chains with spectral gap,
    while \citet{MR4630952} dealt with empirical means over identically and
    independently distributed particles. In particular, to obtain sharp dependency
    with respect to the ambient dimension, we will need Gaussian concentration
    bounds for the ergodic average of the importance sampling weights. A naive
    approach based on Hoeffding's or Bernstein's inequality fails to deliver a
    proxy variance independent of the supremum norm of the (normalised) reweighting
    function $G_t/\pi_{t-1}(G_t)$. We will derive a Gaussian concentration lemma
    for $\{\widehat{\pi}_{t-1}(G_t)\geq \pi_{t-1}(G_t)-s\}$ specifically for
    deviation $s$ at least of order $\sqrt{\Var_{\pi_{t-1}}(G_t)}$, for which the proxy variance
    is the $\chi^2$-divergence between $\pi_t$ and $\pi_{t-1}$. Then
    Assumption~\ref{ass:chisquare_is_small} will ensure the target deviation
    $s\propto \pi_{t-1}(G_t)$ falls within this regime.

    As for the guarantees for the normalising constant estimator,
    our analysis uses the same key ingredients i.e., the introduction of coupling events and the warmness of the starting distributions.
    However, the previously mentioned Gaussian concentration lemmas are replaced by moment-based concentration restricted to the coupling events.
    A sharper dependency in $\eta$ is recovered afterward by using median-based estimators.
    Guarantees for $\hat{Z}$ hold for SMC provided we control the mixing rates of the kernels as defined by $\tau_\beta$.
    We discuss in Section~\ref{sec:spectral_gap_limitation} the implications
    of this when using kernels with no known gap in the case of SMC.

    Section~\ref{sec:coupling} constructs the coupling events and warm-start bounds.
    Section~\ref{subsec:concentration_spgap} establishes concentration for ergodic averages and lower-tail control for $\widehat\pi_{t-1}(G_t)$.
    Section~\ref{subsec:induction} completes the induction and yields Theorems~\ref{th:main_wf} and~\ref{th:wf_smc_greedy}.
    Finally, we sketch the proofs of Theorems~\ref{th:strong_normalising_constant} and~\ref{th:boosted}.

    \subsection{Coupling construction}
    \label{sec:coupling}
    We refer the reader to Appendix~\ref{appendix:proof_thm_main_wf} for the proofs
    of the lemmas stated in this section.

    The next lemma states the existence of a maximal coupling between the two measures induced by two Markov chains, and relates the mixing time of the kernel with the meeting probability.
    \begin{lemma}
        \label{lemma:coupling}
        Let $\pi$, $\tilde{\pi}$ be two probability measures on $(\ambientspace, \ambientalgebra)$.
        Let $K$ be a Markov kernel on $(\ambientspace,\ambientalgebra)$ with invariant probability measure $\pi$, let $\Pi_K$ be the path probability measure induced by a Markov chain starting at stationarity, and $\tilde{\Pi}_K$ for the measure induced by a Markov chain starting from $\tilde{\pi}$.
        There exists an exact maximal distributional coupling $(\bar{Y},Y)$ with coupling time $R = \inf\{p\geq 1 : Y_p = \bar{Y}_p \}$ of $(\Pi_K, \tilde{\Pi}_K)$, that is
        \begin{enumerate}
            \item (distributional coupling property) $Y$ has marginal law $\tilde{\Pi}_K$, and $\bar{Y}$ has marginal law $\Pi_K$.
            \item (exact coupling) Conditional on $R<\infty$, $(Y_{R}, Y_{R+1}, \ldots) = (\bar{Y}_R, \bar{Y}_{R+1}, \ldots)$ almost surely.
            \item (maximal coupling) For any $r\geq 1$,
            \begin{equation*}
                \Pr(R>r) = \textup{TV}(\tilde{\pi}K^{r-1}, \pi).
            \end{equation*}
        \end{enumerate}
        Furthermore, for any $\xi > 0$, if $ \tau(\xi, \tilde{\pi}, K)<r$ 
        \begin{equation*}
            \Pr(R>r)\leq \xi.
        \end{equation*}
    \end{lemma}

    Let $X_{t}^{1:M, 1:P}$ be the particles produced at step $t\geq 1$ in
    waste-free SMC (Algorithm~\ref{alg:wf_smc}). To simplify the notation, let
    $Y_{t}^m[p]=X_{t}^{m,p}$, and $Y_{t}^{m} = (Y_{t}^{m}[1],\ldots,
    Y_{t}^{m}[P])$. Then for any $1\leq m\leq M$, $Y_{t}^{m}$ is a chain of
    length $P$ with the following distribution, conditional on
    $\mathcal{G}_{t-1}$, the $\sigma$-algebra induced by chains from all
    previous iterations $Y_{0:t-1}^{1:M}$:
    \begin{equation}
        \begin{split}
            \label{eq:law} Y_{t}^{m} \mid \mathcal{G}_{t-1} &\sim
            \left(\sum_{m'=1}^{M}\sum_{p=1}^{P}
                \frac{G_{t-1}(Y_{t-1}^{m'}[p])}{\sum_{m''=1}^M\sum_{p'=1}^P
                G_{t-1}(Y_{t-1}^{m''}[p'])}\delta_{Y_{t-1}^{m'}[p]}(\dd
                Y_{t}^{m}[1])\right)\\
            &\indent \times
            \prod_{p=2}^P
            K_t(Y_{t}^{m}[p-1], \dd
            Y_{t}^{m}[p]).
        \end{split}
    \end{equation}
    Furthermore, the chains $Y_{t}^{1:M}$ are independent conditional on
    $\mathcal{G}_{t-1}$. At step $t=0$, $Y_0^{1:M}\sim (\nu^{\otimes
    P})^{\otimes M}$ by construction. We will iteratively construct for each iteration
    $t\geq 0$, a set of stationary Markov chains $\bar{Y}_{t}^{1:M}$ with
    kernel $K_t$, ($K_0(x, \dd y) = \nu(\dd y)$ for consistency), and such that conditionally
    on some event $\bfC_{t-1}$, the probability of all chains $Y_{t}^{1:M}$ to have coupled with
    $\bar{Y}_{t}^{1:M}$ at some time $r_{t}$ is large by Lemma~\ref{lemma:coupling}. The event
    $\bfC_{t-1}$ will depend only on previous iterations and will preserve exchangeability
    of $\{(Y_t^{m}, \bar{Y}_t^{m}), m=1,\ldots, M\}$.

    We define the (truncated) coupling time as
    \begin{equation*}
        R_{t} = \inf\{1\leq p\leq P : Y_t^{1:M}[p] = \bar{Y}_t^{1:M}[p]\},
    \end{equation*}
    with $R_t=\infty$ if the coupling did not occur.
    We let $\mathcal{F}_{t} = \sigma(\mathcal{G}_{t-1},
    (Y_t^{1:M},\bar{Y}_t^{1:M}))$ be the extended $\sigma$-algebra induced by $\mathcal{G}_{t-1} = \sigma(Y_{0:t-1}^{1:M})$, the current chains $Y_t^{1:M}$, and the stationary counterparts $\bar{Y}_t^{1:M}$. 
    To better distinguish the resampled particles $X_{t}^{1:M, 1} = Y_{t}^{1:M}[1]$ from the particles produced at the following Markov steps $p=2, \ldots, P$, we write $\tilde{X}_{t}^{1:M}=X_{t}^{1:M, 1}$.

    We let $\tilde{\pi}_{t-1}(\cdot \mid \bfC_{t-1})$ be the marginal distribution of $\tilde{X}_t^m$ (for any $m$) given
    $\bfC_{t-1}\in \mathcal{F}_{t-1}$, and we let $(Y_t^m, \bar{Y}_t^m)$ be an exact and maximal coupling in the sense
    of Lemma~\ref{lemma:coupling} with $\tilde{\pi}\coloneq \tilde{\pi}_{t-1}(\cdot \mid \bfC_{t-1})$ and $K\coloneq K_t$.
    This coupling has the following properties:
    \begin{enumerate}
        \item For all $t\geq 0$ and $R_t \leq r \leq P$, $Y_t^{1:M}[r] = \bar{Y}_t^{1:M}[r]$ almost surely.
        \item For $t\geq 0$, for each $m=1,\ldots, M$, conditionally on $\bfC_{t-1}$, $Y_t^{m}$ has the law of a Markov chain with kernel $K_t$, with $Y_t^m[1]\sim \tilde{\pi}_{t-1}(\cdot\mid \bfC_{t-1})$.
        The chains $(\bar{Y}_t^{1:M})$ all start at stationarity.
        \item For any $t\geq 0$, for any $r>\tau(\delta, \tilde{\pi}_{t-1}(\cdot\mid \bfC_{t-1}), K_t)$ with $r\leq P$, and each $m=1,\ldots, M$
        \begin{equation*}
            \Pr(Y_t^{m}[r] \neq \bar{Y}_t^{m}[r]\mid \bfC_{t-1}) \leq \delta.
        \end{equation*}
    \end{enumerate}
    Furthermore, for $t=0$, $(Y_t^{1:M})$ and $(\bar{Y}_t^{{1:M}})$ meet almost surely at time $R_0=1$.

    Let us fix $1\leq r_{t}\leq P$, which shall later be specified. We define the following $\mathcal{F}_t$-measurable events:
    \begin{equation*}
        \begin{split}
            \bfA_{t}&= \{Y_t^{1:M}[r_t]= \bar{Y}_t^{1:M}[r_t]\},\\
            \bfB_{t}&= \{\widehat{\pi}_{t-1}(G_t)\geq \pi_{t-1}(G_t)/4 \},\\
        \end{split}
    \end{equation*}
    where $\widehat{\pi}_{t-1}(G_t) = N^{-1}\sum_{m=1}^M\sum_{p=1}^P G_{t}(Y_{t}^{m}[p])$.
    Then
    \begin{enumerate}
        \item $\bfA_{t}$ denotes the event that all the $M$ chains $Y_{t}^{1:M}$ have already met with their stationary counterparts $\bar{Y}_t^{1:M}$ at time $r_{t}$.
        This is equivalent to $R_{t}\leq r_{t}$.
        \item $\bfB_{t}$ denotes the event that the estimated ratio of normalising constants $\widehat{\pi}_{t-1}(G_t)$ between $\pi_t$ and $\pi_{t-1}$ underestimates the true ratio $\pi_{t-1}(G_{t})=Z_t/Z_{t-1}$ by at most a factor $4$.
    \end{enumerate}
    In addition, we now specify the chaining of the previous events at times $(r_0,r_1,\ldots, r_t)$:
    \begin{equation*}
        \bfC_{t} = \bfC_{t-1} \cap \bfA_t \cap \bfB_t,
    \end{equation*}
    where, for consistency, we let $\bfC_{-1}$ be an almost-sure event.
    The one-to-one dependency between $\bfC_{t}$ and the sequence $(r_0,\ldots, r_t)$ is made implicit.

    Similarly to~\citet{MR4630952}, we show that, conditional on $\bfC_{t-1}$, the marginal law of a chain $Y_t^m$ at iteration $t$ has a warm path density with respect to the stationary path.
    The next lemma states how the warmness parameter $\Omega_{t-1}$ depends on the previous meeting times $(r_0,\ldots, r_{t-1})$ and $P$.
    \begin{lemma}
        \label{lemma:warmness_path_density}
        Let $t\geq 1$.
        Assume $\Pr(\bfC_{t'})\geq 4/\omega_{t'}$ for some $\omega_{t'}>4$, and all $t'=0,\ldots, t-1$.
        Then, the conditional law of the (truncated) Markov chain $Y_t^m \mid \bfC_{t-1}$, defined for any $B=B_1\times\ldots\times B_P\in \ambientalgebra^P$ by 
        \begin{equation*}
            \Pr(Y_t^{m}\in B\mid \bfC_{t-1}) = \int_{B_1}\tilde{\pi}_{t-1}(\dd y_1 \mid \bfC_{t-1})\int_{B_2}K_t(y_1,\dd y_2)\ldots\int_{B_P} K_t(y_{P-1},\dd y_{P}),
        \end{equation*}
        is warm with respect to the path density of a stationary Markov chain.
        Furthermore, the warmness parameter $\Omega_{t-1}$ satisfies the recurrence relation
        \begin{equation}
            \label{eq:rec_omega}
            \Omega_{t-1} =\frac{4(r_{t-1}-1)}{P}\frac{\Pr(\bfC_{t-2})}{\Pr(\bfC_{t-1})}\Omega_{t-2}+\omega_{t-1},\indent \Omega_{-1} = 1.
        \end{equation}
    \end{lemma}
    In the rest of the proofs, we let $\delta,\delta'\in (0, 1)$, $\eta\in (0, 1)$, $\varepsilon>0$, and $T\geq 1$ is fixed.
    Lemmas are stated for a generic $t\geq 0$, under the implicit assumption that sets $\bfC_0,\ldots, \bfC_{t-1}$ satisfy the assumption of Lemma~\ref{lemma:warmness_path_density}.

    The next lemma shows that all warm chains can be coupled simultaneously with high probability, conditional on $\bfC_{t-1}$.
    This follows immediately from a union bound over $m=1,\ldots, M$ on the individual coupling failure events.
    \begin{lemma}
        \label{lemma:prasrs}
        Take $r_{t}> \tau(\delta/M, \Omega_{t-1}, K_t)$, then $
        \Pr(\bfA_{t} \mid \bfC_{t-1})\geq (1-\delta)$.
    \end{lemma}

    \subsection{Concentration under the spectral gap assumption}
    \label{subsec:concentration_spgap}
    Sub-Gaussian concentration for
    $\widehat{\pi}_{t-1}(f)=N^{-1}\sum_{m=1}^M\sum_{p=1}^P f(Y_{t}^m[p])$ can
    be obtained via a Hoeffding-type concentration bound for Markov chains with
    a spectral gap~\citep{lezaud:tel-01084797, JMLR:v22:19-479}.

    \begin{lemma}
        \label{lemma:concentration}
        Under Assumption~\ref{ass:spg},
        take $P\geq 4\norm{f}_{\infty}^2\log(2\Omega_{t-1}/\delta')/(\gamma_{t}\varepsilon^2)$,
        then
        $\Pr(\abs{\widehat{\pi}_{t-1}(f)-\pi_{t-1}(f)}<\varepsilon) \geq (1-\delta')\Pr(\bfC_{t-1})$.
    \end{lemma}
    Using the previous Hoeffding-type lemma
    with $f=G_t$ and $\varepsilon=\beta\pi_{t-1}(G_t)$ with $\beta=3/4$ to bound $\Pr(\bfB_t\mid \bfC_{t-1})$ yields the requirement
    $P=\Omega\left(1/\pi^2_{t-1}(G_t)\right)$ (omitting the dependencies in other
    parameters), which behaves poorly with respect to the ambient
    dimension $d$, see Example~\ref{example:special_case_gaussian}.

    \citet{marion2025finitesampleboundssequential} improve
    on~\citet{MR4630952} by replacing Hoeffding's inequality with a
    one-sided Bernstein's inequality
    for signed non-centered independent and identically distributed variables.
    Remarkably, this gets rid of any dependency on $1/\pi_{t-1}(G_t)$;
    the resulting bound depends instead on
    the $\chi^2$-divergence between $\pi_t$ and $\pi_{t-1}$:
    \begin{equation*}
        \chi^2_t = \Var_{\pi_{t-1}}[G_t]/\pi_{t-1}^2(G_t).
    \end{equation*}

    Unfortunately, this approach is not directly generalisable to the ergodic averages.
    However,
    under the assumption that the deviation $\beta \pi_{t-1}(G_t)$ is at least of order the standard deviation, i.e.
    $ \beta \pi_{t-1}(G_t) \gtrsim \sqrt{\Var_{\pi_{t-1}}[G_t]}$,
    or equivalently that $\chi^2_t\lesssim \beta^2$
    (Assumption~\ref{ass:chisquare_is_small}), we recover Gaussian concentration with sharp rates.
    \begin{lemma}
        \label{lemma:concentration_of_weights}

        Let $\beta \in (0, 1)$ and assume $\chi_t/2\leq \beta$, then
        \begin{equation*}
            \begin{split}
                &\Pr\left[\widehat{\pi}_{t-1}(G_t)\leq (1-\beta)\pi_{t-1}(G_t)\mid \bfC_{t-1}\right]\\
                &\indent \leq \Omega_{t-1}\exp\left\{-\frac{P\gamma_t}{2}\min\left(\left(\frac{\beta}{2\chi_t}-\frac{1}{4}\right)^2,\frac{\beta-\chi_t/2}{10}\right)\right\}.
            \end{split}
        \end{equation*}
        Under Assumption~\ref{ass:chisquare_is_small}, the above inequality implies that if $P\geq \frac{128}{\gamma_t}\log(\Omega_{t-1}/\delta')$, then $\Pr(\bfB_t\mid  \bfC_{t-1})\geq 1 -\delta'$.
    \end{lemma}

    \subsection{Finishing the induction}
    \label{subsec:induction}
    We use Lemmas~\ref{lemma:prasrs} and~\ref{lemma:concentration_of_weights} to
    recursively construct a sequence $(r_0,\ldots, r_{t-1})$ such that $\Pr(\bfC_{t-1}^{\textup{C}})\lesssim T(\delta+\delta')$. 
    Lemma~\ref{lemma:concentration} implies that, provided that $\delta, \delta' = \Theta(\eta/T)$ and that $P=\tilde{\Omega}(1/(\gamma\varepsilon^2))$ (omitting logarithmic dependencies in $\eta$, $T$ and $\Omega_{T-1}$), the final estimator $\widehat{\pi}_{T-1}(f)$ concentrates around $\pi_{T-1}(f)$ with desired precision and probability.
    \begin{lemma}
        \label{lemma:main_thm_technical}
        Let $\delta = \eta/(2T)$.
        Let $\omega_t=4/((1-3/2\delta)^{t+1})$ for $t=0,\ldots, T-1$.
        Take $r_{t}> \tau(\delta/M, \Omega_{t-1}, K_{t}), t=0,\ldots, T$, and $(\Omega_t)$ defined by~\eqref{eq:rec_omega}.
        Then $\bfC_0,\ldots, \bfC_{T-1}$ satisfy the requirement of Lemma~\ref{lemma:warmness_path_density}.
        Take $P\geq  \frac{128}{\gamma_t} \log \frac{2\Omega_{t-1}}{\delta} $ for $t=0,\ldots, T-1$, and $P\geq  \frac{4}{\gamma_T \varepsilon^2}\log \frac{4\Omega_{T-1}}{\delta}$.
        Then, for any $f:\ambientspace\to\bbR$ with $\abs{f}\leq 1$,
        $           \Pr(\abs{\widehat{\pi}_{T-1}(f)-\pi_{T-1}(f)}< \varepsilon)\geq 1-\eta$.
    \end{lemma}
    Theorem~\ref{th:main_wf} and the refined version Theorem~\ref{th:wf_smc_greedy} follow from Lemma~\ref{lemma:main_thm_technical} by taking $r_{t}=\tau(\delta/M, \Omega_{t-1}, K_t)+1$ and $P\geq 32\log(32M/\delta)/\gamma$ for which $\Omega_{t-1}\leq 32$
    (Lemma~\ref{lemma:bound_warmness}).

    \subsection{Union bounds and second moment bounds for the normalising constant}
    To establish Theorem~\ref{th:strong_normalising_constant}, we split the failure probability according to the good event $\bfC_{T-1}$.
    If the relative error $\abs{\widehat{Z}_T/Z_T-1}$ is $\geq \varepsilon$ then the sum of the relative errors
    of each factor $\sum_{t=0}^T\abs{\hat{\pi}_{t-1}(G_t)/\pi_{t-1}(G_t)-1}$ is $\geq \varepsilon/2$.
    On $\bfC_{T-1}$, we control this sum by applying Markov's inequality on the squared sum
    followed by Minkowski's inequality to split it into a sum of second moments:
    for any $\varepsilon \in (0,2)$ (Lemma~\ref{lemma:union})
    \begin{equation*}
        \Pr(\abs{\widehat{Z}_T/Z_T-1}\geq \varepsilon) \leq \Pr(\bfC_{T-1}^{\textup{C}}) + \frac{4}{\varepsilon^2}\left(\sum_{t=0}^{T} \left\{\frac{\bbE[(\widehat{\pi}_{t-1}(G_t)-\pi_{t-1}(G_t))^2\mid \bfC_{t-1}]}{\pi^2_{t-1}(G_t)}\right\}^{1/2}\right)^2.
    \end{equation*}
    The first term can be made as small as desired provided $P$ and the meeting times are of order the mixing times.
    The second term, as it proceeds on a well-behaved set of trajectories
    (e.g., that couple with known meeting times), can be bounded.

    When restricting ourselves to $\bfC_{t-1}$, an average over any chain $Y^m_t$ essentially behaves like the average over a chain with path density that is $\Omega_{t-1}$-warm with respect to the stationary path density.
    This allows us to closely relate the second order moments of the ergodic average to the moment under stationarity (Lemma~\ref{lemma:bound_mse})
    \begin{equation*}
        \bbE[(\hat{\pi}_{t-1}(G_t)-\pi_{t-1}(G_t))^2\mid \bfC_{t-1}]\leq \frac{2\Omega_{t-1}\pi_{t-1}(G_t)^2}{\gamma P},
    \end{equation*}
    where $\Omega_{t-1} = \OO(1)$ provided $P$ is of order the mixing time (Lemma~\ref{lemma:bound_warmness}).
    In particular, provided $P=\Omega\left(T^2\gamma^{-1}\varepsilon^{-2}\right)$ with constant large enough, the failure probability is upper bounded by $1/4$.

    The guarantee on the product-of-medians estimator (Theorem~\ref{th:boosted}) follows from similar arguments except
    that a union bound is applied over marginal guarantees on each median-estimated ratio of
    normalising constants (Lemma~\ref{lemma:boosted}).

    \section{Application to tempering and other scenarios}
    \label{sec:application}

    \subsection{Geometric tempering on log-concave and smooth distributions}
    \label{subsec:log_concave_distribution}

    We now specialise our results to geometric tempering, where
    \begin{equation*}
        \pi_{t}(\dd x) \propto \exp\left\{-\lambda_t U(x)\right\} q^{1-\lambda_t}(x)\dd x,
    \end{equation*}
    $q$ is a proposal distribution (from which initial particles $X_0^m$ are
    simulated), $\pi\propto e^{-U}$ is the target distribution, and $0=\lambda_{-1} < \dots <
    \lambda_T=1$ is referred to as the tempering schedule. Potential limitations
    of geometric tempering have been pointed out in~\citet{chehab2025provable}.
    \begin{assumption}
        \label{ass:log_concave}
        $q\propto e^{-Q}$ and $\pi\propto e^{-U}$ are $\mathcal{C}^2$, log-concave, $\nabla^2 Q \succcurlyeq \alpha_Q I_d$ for some $\alpha_Q\geq 0$, and
        $V\coloneq U - Q$ satisfies for some $(\alpha_V, \beta_V)\in \bbR\times \bbR^{+}$
        \begin{equation*}
            \alpha_V I_d \preccurlyeq\nabla^2 V\preccurlyeq \beta_V I_d, \indent \abs{\alpha_V}\leq \beta_V, \indent \alpha_Q + \alpha_V > 0.
        \end{equation*}
        Furthermore, they share a common mode at $x=x^\star$, i.e. $\nabla Q(x^\star) = \nabla U(x^\star) = 0$.
    \end{assumption}
    Assumption~\ref{ass:log_concave} is less restrictive than strictly log-concave-smooth potentials $V$.
    In particular, it allows the density ratio $\dd\pi/\dd q \propto e^{-V}$ to be not necessarily log-concave (e.g. $\alpha_V<0$), as long as $Q+\lambda V$ remains uniformly strongly convex along the path (i.e. $\alpha_Q+\lambda\alpha_V>0$).
    If $Q$ (resp. $U$) is $(\alpha_{Q}, \beta_{Q})$-concave-smooth (resp. $(\alpha_U, \beta_U)$-concave-smooth), then $V$ satisfies the previous assumption with $\alpha_V = \alpha_U-\beta_Q$, $\beta_V = \max(\beta_U-\alpha_Q, \beta_Q - \alpha_U)$ provided $\beta_Q< \alpha_Q+\alpha_U$.

    The next theorem states a sufficient condition for the tempering sequence to satisfy Assumption~\ref{ass:chisquare_is_small}.
    \begin{theorem}
        \label{th:chi_log_concave}
        Assume~\ref{ass:log_concave}.
        Let $c\in (0, 1/8)$.
        If the tempering schedule satisfies $\lambda_t-\lambda_{t-1} \leq c\frac{\alpha_Q + \lambda_{t-1} \alpha_V}{\beta_V \sqrt{d}}$, then
        $            \chi^2(\pi_t\mid\pi_{t-1})\leq e^{c^2(1+24/\sqrt{d})}-1$.
        In particular, for $c=1/\{8\sqrt{1+24/\sqrt{d}}\}$, $(\pi_t)$ satisfies Assumption~\ref{ass:chisquare_is_small}.
    \end{theorem}
    The tempering schedule given by the equality case in the previous theorem with initialisation $\lambda_{-1}=0$ and limit value $\lambda_{T} = 1$ is
    \begin{equation}
        \label{eq:schedule}
        \lambda_t = 1 \wedge \left[
                                 \frac{\alpha_Q}{\alpha_V}\left\{\left(1+\frac{c}{\kappa_V
            \sqrt{d}}\right)^{t+1}-1\right\}\right],
    \end{equation}
    where $\kappa_V = \beta_V/\alpha_V$, and the tempering horizon is
    $T\sim \sqrt{d}\abs{\kappa_V \log(1+\alpha_V/\alpha_Q)}/c$ if $\alpha_V\neq 0$ and $T\sim \sqrt{d}\beta_V/( \alpha_Q c)$
    if $\alpha_V = 0$, assuming $\alpha_Q >0$. 
    If $Q$ is non-strongly convex (i.e. $\alpha_Q = 0$), the tempering schedule satisfies $T\sim \kappa_V\sqrt{d}\log(1/\lambda_0)/c$ for some initial temperature $\lambda_0>0$.

    The common-mode assumption is not required for
    explicit feasibility of tempering.
    \begin{theorem}
        \label{th:chi_log_concave_extension}
        Assume~\ref{ass:log_concave} except for its common-mode requirement, and
        let $v_t = \nabla V(x^\star_t)$ where $x^\star_{t}$ is the mode of $\pi_{t}$.
        Then Theorem~\ref{th:chi_log_concave} still holds, provided the
        tempering schedule satisfies 
        \[\lambda_t-\lambda_{t-1}\leq c
        \frac{\alpha_Q + \lambda_{t-1}\alpha_V}{\beta_V
      \sqrt{d}+\norm{v_{t-1}}\sqrt{\alpha_Q + \lambda_{t-1}\alpha_V}}.\]
    \end{theorem}

    In particular, strong log-concavity of the intermediate targets along with smoothness of $V$
    implies that $\norm{v_t}\leq \{1+\frac{\lambda_t \beta_V}{\alpha_{Q}+\lambda_t \alpha_V}\}\norm{\nabla V(x^\star_Q)}$.
    As a consequence, when $U$ is smooth (resp. $V$ admits a critical point), provided the mode mismatch between the proposal $q$ and the target $\pi$ (resp. the distance $\norm{x^\star_Q-x^\star_V}$) is of order $\sqrt{d}$, the dimensional
    dependence of $T$ remains unchanged.

    Theorem~\ref{th:chi_log_concave} along with direct applications of results from Section~\ref{sec:main}, namely Theorems~\ref{th:main_wf} and~\ref{th:wf_smc_greedy} for moment estimates, and Theorems~\ref{th:strong_normalising_constant} and~\ref{th:boosted} for normalising constants, yield
    the complexity bounds given below (Table~\ref{tab:log_concave_complexity}). The same results hold when the modes are misaligned up to the dependence on $\kappa_V$.
    \begin{table}[h]
        \centering
        \begin{tabular}{c|c|c}
            \arrayrulecolor{black}
            \textbf{Estimate} & \textbf{Algorithm} & \textbf{Complexity} \\
            \hline
            \multirow{2}{*}{$\widehat{\pi}(f)$}
            & Alg.~\ref{alg:smc_size}
            & $\frac{1}{\gamma}\log\left(\abs{\kappa_V}\sqrt{d}\right)\left(\abs{\kappa_V}\sqrt{d}+ \frac{1}{\varepsilon^2}\right)$ \\
            \hhline{~--}
            & Alg.~\ref{alg:wf_smc}
            & $\frac{\sqrt{d}\abs{\kappa_V}}{\gamma\varepsilon^2}\log\left(\abs{\kappa_V}\sqrt{d}\right)$ \\
            \hline
            $\hat{Z}_T$
            & Alg.~\ref{alg:wf_smc}
            & $\frac{d^{3/2}\abs{\kappa_V}^3}{\gamma\varepsilon^2}$ \\
            \hline
        \end{tabular}
        \caption{Summary of $\OO(\cdot)$ complexities under Assumptions~\ref{ass:spg} and~\ref{ass:log_concave}. Log factors involving $1/\eta$ or $\alpha_V/\alpha_Q$ have been removed.}
        \label{tab:log_concave_complexity}
    \end{table}

    \subsection{Random-walk Metropolis kernels}
    \label{subsec:rwmh}
    The spectral gap of a well-calibrated RWM (random-walk Metropolis) kernel
    targeting a log-concave-smooth density with condition number $\kappa$
    satisfies $\gamma=\Omega(1/(\kappa
    d))$~\citep[Th.~1 of][]{10.1214/24-AAP2058}.
    Assuming $Q$ is $\beta_Q$-smooth and $\alpha_Q$-strongly convex, and
    writing $\kappa=\max(\kappa_Q , \abs{\kappa_V}, (\beta_Q+\beta_V)/(\alpha_Q+\alpha_V))$ as a uniform bound on the condition
    numbers along the path, we obtain (up to polylogarithmic factors) the
    following complexities:
    for estimating $\pi(f)$, $\tilde\OO(d\varepsilon^{-2}\kappa+\kappa ^2 d^{3/2})$ using Algorithm~\ref{alg:smc_size}, against $\tilde{\OO}(d^{3/2}\kappa^2\varepsilon^{-2})$ using waste-free SMC (Algorithm~\ref{alg:wf_smc});
    for estimating $Z_T$ with failure probability $\eta=1/4$, $\OO(\kappa^4 d^{5/2}\varepsilon^{-2})$ using Algorithm~\ref{alg:wf_smc}.

    \subsection{Latent Gaussian models and preconditioned Crank-Nicolson Metropolis-Hastings (pCN-MH)
        kernels}
    \label{sec:pcn}
    We focus on target distributions $\pi$ that admit a log-concave-smooth density with respect to a Gaussian base measure $q(\dd x)=\mathcal{N}(\dd x, 0, C)$ for some prior covariance matrix $C$.
    Let $V:\bbR^d \to \bbR$ be some potential function, and let $\pi$ be the target distribution given by
    \begin{equation}
        \label{eq:target_pcn}
        \pi(\dd x)\propto \mathcal{N}(\dd x, 0, C)\exp\left(-V(x)\right).
    \end{equation}
    To establish quantitative bounds on the mixing properties of the kernels, we impose restrictions on $\pi$ that are similar to Assumption~\ref{ass:log_concave}.
    \label{sec:spectral_gap_pcn}
    \begin{assumption}
        \label{ass:V}
        $V$ is a $\mathcal{C}^2$ convex function, with $0\preccurlyeq \nabla^2 V\preccurlyeq \beta_V I_d$, and $V$ is minimised at $x^\star=0$.
    \end{assumption}
    A prior informed proposal is given by the preconditioned Crank-Nicolson (pCN) proposal:
    \begin{equation*}
        Y = \rho x + \sqrt{1-\rho^2}C^{1/2}\zeta, \qquad \zeta\sim\mathcal{N}(0, I_d)
    \end{equation*}
    where $\rho\in [0, 1)$ is the autoregressive parameter. We let $K_t(\rho)$ be the associated Metropolis-Hastings kernel targeting $\pi_{t-1}$.

    For simplicity, we assume that $\lambda_{\max}(C)=1$ and $\beta_V\geq 1$, and, when comparing with RWM-MH kernels, that $C$ is well conditioned (i.e. $\lambda_{\max}(C)/\lambda_{\min}(C)=\OO(1)$).
    Under Assumption~\ref{ass:V}, provided the kernels are well calibrated over the chosen tempering schedule, the spectral gap of $K_t$ is $\Omega(1/(1+\lambda_{t-1} \beta_V\textup{Tr}(C)))$ (see Proposition~\ref{prop:rho_pcn} in the appendix, and~\citet[Th. 54, Lemma 58]{10.1214/24-AAP2058}),
    while RWM has a spectral gap $\Omega(1/(d+\lambda_{t-1}\beta_V d))$.

    We assume the sequence of distributions is set by tempering with schedule prescribed by Theorem~\ref{th:chi_log_concave}; then, in earlier iterations,
     $    \lambda_{t}\sim t/\{\beta_V \sqrt{d}\}$, so that Algorithm~\ref{alg:smc_size} has a one-iteration cost reduced to $\tilde{\OO}(1+t\textup{Tr}(C)/\sqrt{d})$ compared to $\tilde{\OO}(d+t\sqrt{d})$ if using calibrated RWM kernels (Theorem~\ref{th:wf_smc_greedy}).
    The total complexity for producing a warm-start using pCN is then $\tilde\OO(\beta_V^2 \sqrt{d} \textup{Tr}(C))$ against $\tilde\OO(\beta_V^2 d^{3/2})$.
    \begin{table}[h]
        \centering
        \begin{tabular}{c|c|c}
            \arrayrulecolor{black}
            & \textbf{RWM-MH} & \textbf{pCN-MH} \\
            \hline
            Spectral gap of $K_t$
            & $\Omega\big(1/\{d(1+\lambda_{t-1}\beta_V)\}\big)$
            & $\Omega\big(1/\{1+\lambda_{t-1}\beta_V\textup{Tr}(C)\}\big)$ \\
            \hline
            Cost of iteration $t$
            & $\tilde{\OO}(d + t\sqrt{d})$
            & $\tilde{\OO}(1 + t\,\textup{Tr}(C)/\sqrt{d})$ \\
            \hline
            Total warm-start cost
            & $\tilde{\OO}(\beta_V^2 d^{3/2})$
            & $\tilde{\OO}(\beta_V^2\sqrt{d}\,\textup{Tr}(C))$ \\
            \hline
        \end{tabular}
        \caption{Cost of Algorithm~\ref{alg:smc_size} with calibrated RWM-MH and pCN-MH kernels for $\pi\propto\mathcal{N}(0, C)e^{-V}$ under Assumption~\ref{ass:V}, with $\lambda_{\max}(C)=1$, $\beta_V\geq 1$, $C$ well conditioned, and tempering schedule $\lambda_t\sim t/\{\beta_V\sqrt{d}\}$, $t\leq T\sim\beta_V\sqrt{d}$.}
        \label{tab:rwm_vs_pcn}
    \end{table}

    The gain of pCN over RWM is larger when the prior measure has a low effective dimension ($\operatorname{Tr}(C)$ small in front of $d$); or
    in the case where $C$ is ill-conditioned (which we have not
    treated). Such posteriors typically arise in Bayesian inverse problems with Gaussian process priors, see~\citet{Cotter_2013}.

    \section{The spectral gap assumption} \label{sec:spectral_gap_limitation}

    \subsection{From spectral gaps to mixing times}

    Assuming that each kernel $K_t$ admits a spectral gap is a strong requirement,
    which may lead to sub-optimal complexity bounds. For instance, under proper
    conditions, MALA (Metropolis adjusted Langevin) kernels have an $\Omega(d^{-1})$
    spectral gap, but their mixing times can be sublinear in
    $d$ when initialised with an $\OO(1)$-warm distribution~\citep{JMLR:v20:19-306, pmlr-v134-chewi21a}.

    When studying waste-free SMC, we used the spectral gap assumption
    (Assumption~\ref{ass:spg}), in particular, to upper bound the variance of sums over Markov
    chains. We do not need to do this for standard SMC; instead the per-particle coupling
    of~\citet{MR4630952} imposes constraints of the form $P\gtrsim \tau_\beta(\xi)$ with $\xi=\eta/(MT)$,
    and
    \begin{equation}
        \begin{split}
            \label{eq:stron_requirement_tau_b}
            \tau_{\beta}(\xi) \textup{ the smallest } n\geq 0
            \textup{ such that } \sup_{t=1,\ldots, T}\int \text{TV}(\delta_x K_t^n,\pi_{t-1})\pi_{t-1}(\dd x) \leq \xi.
        \end{split}
    \end{equation}
    The requirement~\eqref{eq:stron_requirement_tau_b} is in general stronger than warm-start mixing but may yield
    better bounds than the spectral gap assumption; we will discuss two such cases.

    Consequently, while the SMC results do not require the kernels to have a spectral gap or even to leave
    their target invariant, and are in this sense more general than the waste-free results, it remains
    unclear whether the resulting complexities are sharper in general.
    In return for this pointwise fast mixing assumption, the coupling of~\citet{MR4630952} is such that whatever the starting particles $\tilde{X}_t^m$,
    the coupled particles are i.i.d.\ and independent across iterations,
    so that the variance of $\hat{Z}_T$ on the coupling event grows linearly in $T$, and no union bound over iterations is required.

    Below, we state guarantees for the normalising constant estimate $\hat{Z}_T$ produced
    by standard SMC (\cref{alg:smc}) using $\tau_\beta$-mixing times.

    \begin{theorem}
        \label{th:normalising_smc} Assume~\ref{ass:chisquare_is_small}. Let
        $\varepsilon\in (0, 2)$, $M=\Omega(T/\varepsilon^2)$ and
        $P=\Omega(\tau_\beta(1/(M T)))$ with numerical constants large enough. Then standard SMC (Algorithm~\ref{alg:smc})
        returns a $\hat{Z}_T$ such that $\abs{\hat{Z}_T/Z_T-1}< \varepsilon$ with
        probability at least $3/4$.
    \end{theorem}

    The specific per-particle coupling of~\citet{MR4630952} for SMC allows one to save a factor $T$
    compared to our results on waste-free SMC at the price of $\tau_\beta(\varepsilon^2/T^2)$.
    We discuss hereafter specific cases where fast-mixing
    under the requirement~\eqref{eq:stron_requirement_tau_b} is still possible,
    when spectral gap yields suboptimal complexity, or when no gap is known.

    \subsection{Algorithmic warm-start and application to log-concave sampling}\label{subsec:warm_start_logconcave}
    Under log-concavity and smoothness of the target, MALA, under a warm-start with
    appropriate step size, is known to produce a sample $\xi$-close in $\textup{TV}$ to the target
    in time $\sqrt{d}$. This bound does not convert into a pointwise control starting from a deterministic point:
    for any step size, there exist a target and an $\exp(d)$-warm start, hence a deterministic point $x$, from which the Metropolis
    filter requires $\widetilde\Omega(\kappa d)$ iterations~\citep{LeeST21}.

    We propose to slightly modify the execution of SMC to circumvent the per-$x$
    mixing time issue and achieve the fast $\sqrt{d}$-mixing time of MALA.
    At each iteration, we proceed with a warm-start procedure to reach, given a resampled particle $x$, a new point whose law
    is warm.
    For example,~\citet[Theorem 4.1, Remark 4.2]{AltChe24HighAcc} show that repeatedly applying the transition kernel associated
    with a (discretized, properly calibrated) underdamped Langevin Monte Carlo (ULMC) to $\delta_x$, outputs a random point $x'\sim \delta_x K_{\textup{ULMC}}^{n(x)}$ in $n(x)$ steps with
    \begin{equation*}
        \chi^2(\delta_x K_{\textup{ULMC}}^{n(x)}\mid \pi) =\OO(1),\qquad n(x)=\tilde{\OO}\left(\kappa^{3/2} \sqrt{d}\log\left(\sqrt{d/\alpha} + \norm{x-x^{\star}}\right)\right),
    \end{equation*}
    (stated in general for $\pi$ an $\alpha$-strongly log-concave and smooth distribution with conditioning $\kappa$ and mode $x^{\star}$).
    As a consequence, the cost of forming a sample whose law is $\OO(1)$-warm, with input $x$ at distance at most $ \sqrt{d/\alpha}+\sqrt{2\log(1/\xi)/\alpha}$ from the mode, collapses to $\tilde{\OO}(\kappa^{3/2} \sqrt{d}\log\log\xi^{-1})$.
    Under $\pi$, $x$ is at greater distance with probability at most $\OO(\xi)$, squeezing to negligible the inflation of the running time for those points.
    From this warm initialisation, the remaining number of MALA steps sufficient to achieve a sample $\xi$-close in $\textup{TV}$ is then $\tau(\xi) = \OO(\kappa  \sqrt{d}\mathop{\textup{polylog}}\xi^{-1})$.
    \begin{proposition}
        \label{prop:smc_mala_warm_start}
        Let $\pi$ be an $(\alpha,\beta)$-log-concave-smooth distribution, put $\kappa=\beta/\alpha$, let $K_{\textup{ULMC}}$ be the
        transition kernel corresponding to one discretized step of the underdamped Langevin Monte Carlo (ULMC) of
        Theorem 4.1 in~\citet{AltChe24HighAcc}, and let $K_{\textup{MALA}}$ be the MALA transition kernel
        of~\citet[Theorem 1; Theorem 7.3.5]{pmlr-v134-chewi21a, chewi2025logconcave}.
        Take
        \begin{equation*}
            n=\tilde{\OO}(\kappa^{3/2} \sqrt{d}\log(\sqrt{d/\alpha}+\sqrt{\log(1/\xi)/\alpha})),\qquad n' = \tilde\OO(\kappa  \sqrt{d}\mathop{\textup{polylog}}\xi^{-1})
        \end{equation*}
        then
        $
        \int \textup{TV}(\delta_x K_{\textup{ULMC}}^{n}K_{\textup{MALA}}^{n'},\pi)\pi(\dd x)\lesssim \xi.
        $
    \end{proposition}
    Whether MALA alone achieves $\beta$-mixing in $\sqrt{d}$ is, to our knowledge, not established; we
    think of the composite kernel made of a non-Metropolised kernel and its Metropolised
    version as a proof technique for controlling $\tau_\beta$.

    We assume~\ref{ass:log_concave} and use the notation of
    Subsection~\ref{subsec:log_concave_distribution}; the following corollary records
    the SMC entry of Table~\ref{tab:normalising_constant_comparison}.
    \begin{corollary}[Log-concave normalising constant with ULMC{}+MALA kernels]
        \label{cor:zhat_ulmc_mala}
        Assume~\ref{ass:log_concave} and take the geometric tempering
        schedule~\eqref{eq:schedule}, so that Assumption~\ref{ass:chisquare_is_small}
        holds, $T=\tilde\Theta(\sqrt d\,\abs{\kappa_V})$ (Theorem~\ref{th:chi_log_concave}) and
        $\kappa=\max(\kappa_Q , \abs{\kappa_V}, (\beta_Q+\beta_V)/(\alpha_Q+\alpha_V))$.
        Then standard SMC (Algorithm~\ref{alg:smc}), with the folded application of $K_t$ replaced by
        the composite ULMC{}+MALA kernel $K_{\textup{ULMC}}^{n}K_{\textup{MALA}}^{n'}$
        targeting $\pi_{t-1}$, returns $\widehat Z_T$ within
        relative precision $1\pm\varepsilon$ with probability at least $3/4$, using
        \begin{equation*}
            \tilde{\OO}\!\left(d^{3/2}\kappa^{7/2}\varepsilon^{-2}\right)\ \text{Markov steps for }\widehat Z_T.
        \end{equation*}
    \end{corollary}
    Table~\ref{tab:normalising_constant_comparison} is a summary of the query complexities of SMC (with ULMC+MALA) and waste-free SMC
    (with RWM, Subsection~\ref{subsec:rwmh}) for estimating $Z_T$.
    \begin{table}[h]
        \centering
        \begin{tabular}{c|c}
            \arrayrulecolor{black}
            \textbf{}                     & $\hat{Z}_T$                             \\
            \hline
            WF-SMC (RWM, $\gamma$)        & $d^{5/2} \kappa^4 \varepsilon^{-2}$     \\[0.6em]
            \hline
            SMC (ULMC+MALA, $\tau_\beta$) & $d^{3/2} \kappa^{7/2} \varepsilon^{-2}$ \\[0.6em]
        \end{tabular}
        \caption{$\tilde\OO(\cdot)$ complexities for estimating $Z_T$ within relative
        precision $1\pm\varepsilon$ with failure probability $\leq 1/4$ (Theorems~\ref{th:strong_normalising_constant} and~\ref{th:normalising_smc}), under log-concave
        and smooth targets ($T=\tilde\OO(\kappa\sqrt{d})$, Theorem~\ref{th:chi_log_concave}). Log factors omitted.
        For RWM: $\gamma=\Omega(1/(\kappa d))$; for ULMC+MALA: $\tau_\beta=\tilde\OO(\kappa^{3/2} \sqrt{d})$.}
        \label{tab:normalising_constant_comparison}
    \end{table}
    The SMC estimator saves a factor $d\kappa^{1/2}$ over waste-free SMC;
    a factor $T\sim\sqrt{d}\kappa$ is due to the per-particle coupling, and a factor $\sqrt{d}\kappa^{-1/2}$ to the use of MALA(+ULMC) instead of RWM,
    though the latter is only established for standard SMC.


    Further refinement of the exponent for SMC may be achieved by using better bounds that come with faster kernels. For example,
  one expects the exponent $\tilde\OO(d^{1+1/2})$ to be further reduced to $\tilde\OO(d^{1+1/4})$ by replacing ULMC and MALA by HMC and its Metropolised
    version, in the vein of~\citet{zhang2026algorithmicwarmstartshamiltonian} who show that HMC+MHMC yield samples near stationarity in $\OO(d^{1/4})$ queries.
    \citet{ge2020estimatingnormalizingconstantslogconcave} show one needs at least $\tilde\Omega(d^{1-\oo(1)})$ queries, so that
    the remaining gap is due to imperfect samples.


    \subsection{Implied complexity for convex body volume estimation}
    We let $\mathcal{K}_T$ be a convex body (non-empty, convex and compact set) in $\bbR^d$,
    and we want to estimate its volume via SMC.
    This amounts to estimating the normalising constant of the uniform distribution over $\mathcal{K}_T$.
    For doing that with SMC, we will (approximately) sample from the uniform distributions over a decreasing sequence of convex bodies
    $\pi_t\propto \mathds{1}_{\mathcal{K}_t}$.
    We propose to use the rejection-based Proximal Sampler of~\citet{3737916.3741356}, which we state below for a generic
    uniform distribution $\pi$ over a convex body $\mathcal{K}$.

    \begin{algorithm}[!ht]
        \caption{One step of the proximal sampler with rejection budget $N'$}
        \label{alg:ps_capped}
        \KwIn{current state $x\in\mathcal{K}$, step size $h>0$, budget $N'\in\mathbb{N}$}
        \KwOut{next state $x'\in\mathcal{K}$}
        Draw $y\sim \mathcal{N}(x,hI_d)$\tcp*{forward step}
        \For{$i=1,\ldots,N'$}{
            Draw $z_i\sim \mathcal{N}(y,hI_d)$\;
            \If{$z_i\in\mathcal{K}$}{
                \Return $x'\gets z_i$\tcp*{backward step: first accepted proposal}
            }
        }
        \Return $x'\gets x$\tcp*{budget exhausted: leave the state unchanged}
    \end{algorithm}
    We let $\hat{K}$ denote the Markov kernel defined by Algorithm~\ref{alg:ps_capped} and $K$
    for the idealised version ($N'=\infty$ in the procedure).

    \citet[Th. 2]{3737916.3741356} proved that $K$ contracts toward $\pi$ in $\chi^2$ starting from $\mu\ll \pi$, provided
    $\pi$ satisfies a Poincaré
    inequality.
    We can show that a one-step application of $K$
    mollifies the point mass $\delta_x$ to a $\pi$-absolutely continuous distribution $\mu\coloneq \delta_x K$, whose warmness can be bounded independently
    of the starting point $x$ thanks to the compactness setting, so that this contraction result applies uniformly
    over $\mathcal{K}$.
    This yields an upper bound on the mixing
    time of $K$ starting from any deterministic point in $\mathcal{K}$.
    Finally, the comparison between the contractions of $K$ and $\hat{K}$ follows from a coupling argument, requiring
    the rejection budget $N'$ to scale polynomially with $1/\xi$. The $\tau_\beta$ bound is given by Proposition~\ref{prop:Khatps}
    in the appendix.

    \begin{corollary}
        \label{cor:convex_volume}
        Let $(\mathcal{K}_t)_{t=-1}^{T}$ be a decreasing sequence of convex bodies with
        $\mathcal{B}_1\subseteq\mathcal{K}_T$ and $\mathcal{K}_{-1}\subseteq\mathcal{B}_D$
        ($D\ge1$), let $\pi_t=\textup{Unif}(\mathcal{K}_t)$, and let $\bar C$ be the
        largest Poincaré constant of the $(\pi_t)$. Assume
        $1/2\le\textup{Vol}(\mathcal{K}_t)/\textup{Vol}(\mathcal{K}_{t-1})<c$ for some
        $c\in(1/2,1)$, so that $T=\tilde\Theta(d \log D)$. Run the standard
        SMC of Theorem~\ref{th:normalising_smc} with the capped proximal sampler
        (Algorithm~\ref{alg:ps_capped}) as $K_t$, taking
        \begin{equation*}
            M=\tilde\Omega(d\varepsilon^{-2}\log D),\qquad
            P=\tilde\Omega(d^3\bar C\log D),\qquad
            N'=\tilde\Omega(d^5 \varepsilon^{-2} \bar C(\log D)^3 ).
        \end{equation*}
        Then with probability at least $3/4$, $\widehat Z_T$ is an estimator of $\textup{Vol}(\mathcal{K}_T)$ within
        relative precision $1\pm\varepsilon$ and the total number of membership evaluations is at most
        $\tilde{\OO}(d^{5} \bar C (\log D)^3 \varepsilon^{-2})$.
    \end{corollary}
    For comparison, the Gaussian cooling algorithm for well-rounded convex bodies of~\citet{doi:10.1137/15M1054250} achieves $\tilde{\OO}(d^3\varepsilon^{-2})$ query complexity.
    Corollary~\ref{cor:convex_volume} is rather meant to illustrate the applicability of Theorem~\ref{th:normalising_smc} on a classical problem, with kernels that have no spectral gap.

    \subsection{Failure of the spectral gap assumption}

    Remarkably, $\hat{K}$ constitutes an example of a kernel with a spectral gap exponentially small in $d$
    while having polynomial mixing times (Proposition~\ref{prop:Khatps} in the appendix).
    \begin{proposition}
        \label{prop:stayput}
        Let $\mathcal{K}=[0,2]^d$, $0<h\le1$.
        Then, the spectral gap $\gamma$ of the Proximal Sampler $\hat{K}$ (Algorithm~\ref{alg:ps_capped}) with rejection budget $N'\geq 1$ satisfies
        \begin{equation*}
            \gamma \ \leq\frac{1.76\,N'\,2^{-d}}{1-\omega_d (\sqrt{h}/d)^d 4^{-d}} = \OO(N'2^{-d}),
        \end{equation*}
        where $\omega_d$ is the volume of a unit ball in $\bbR^d$.
        Furthermore, let $S:=\mathcal{K}\cap \mathcal{B}_{\sqrt h/d}$, $
        g:=\frac{\mathds{1}_S-\pi(S)}{\sqrt{\pi(S)\big(1-\pi(S)\big)}}$, and let $(\bar{X}^p)$ be a stationary Markov
        chain with kernel $\hat{K}$.
        For any $P\geq 1$, provided $2^d\geq 2\times 1.76N'$ and $d\geq 3$,
        \begin{equation*}
            \frac{\Var[P^{-{1/2}}\sum_{p=1}^P g(\bar{X}^p)]}{\Var_{\pi}(g)} \geq 1+\frac{1}{16}\,\min\left(P-1, \frac{2^{d}}{1.76 N'}\right) .
        \end{equation*}
    \end{proposition}

    The striking consequence of the previous proposition is that in the regime $P\leq 1+2^d/(1.76N')$,
    if $2\times 1.76N'\leq 2^d$, the waste-free SMC average $N^{-1}\sum_{n=1}^N g(X_T^{n})$ over $N=M\times P$ samples converges
    in $\mathcal{L}^2$ at rate $\OO(1/M)$ at most, with each of the $M$
    chains carrying, for this task, at most $\OO(1)$ effective samples.
    Specifically, Proposition 3 of~\citet{MR4400392} with Proposition~\ref{prop:stayput} show that as $M\to \infty$, the CLT-like
    convergence in law takes place
    \begin{equation*}
        \frac{1}{\sqrt{M\times P}}\left(\sum_{n=1}^{M\times P} \left\{g(X_T^n)-\pi(g)\right\}\right) \xrightarrow{\mathcal{L}} \mathcal{N}(0, \sigma^2),
    \end{equation*}
    where $ \sigma^2 \geq \Var[P^{-1/2}\sum_{p=1}^P g(\bar{X}^p)] = \Omega(P)$.

    The apparent contradiction with the polynomial rates derived using mixing time can be understood once we bear
    in mind that the spectral gap governs the variance of ergodic averages over all $\mathcal{L}^2(\pi)$-functions,
    and so our requirement that $P$ scales with $\gamma^{-1}$ buys us guarantees over all functions, including
    the constructed pathological function $g$ which has support of mass $\OO(d^{-d}4^{-d})$, i.e. which is essentially
    invisible under $\pi$.

    \section{Conclusion}
    \label{sec:discussion}

    \subsection{Practical recommendations}
    \label{sub:practical}

    For end users, the most important question is which quantities they are
    trying to estimate. If they are only interested in moments with respect to
    $\pi_{T-1}$, we recommend the greedy variant of waste-free SMC,
    \cref{alg:smc_size}. At iterations $t<T$, it is sufficient to take $P_t$
    large enough relative to the mixing time of kernel $K_t$; this may be
    assessed through e.g. (estimated) autocorrelation times. On the other hand,
    one must take $P_T$ much larger, and commensurate with $\varepsilon^{-2}$
    to reach an $\varepsilon$ error (Theorem~\ref{th:wf_smc_greedy}).

    To illustrate this point, consider running greedy waste-free SMC with
    a fixed computational budget $P_T + (T-1)P_t = TP = \mathrm{cst}$, where the final
    iteration is allocated $C$ times the per-iteration budget of waste-free SMC, i.e.
    $P_T = C P$.
    Waste-free SMC is recovered by taking $C=1$.
    Figure~\ref{fig:wf_smc} shows how $C$ affects the result in a particular
    example: taking $C>1$ is beneficial, unless $C$ is too large, in which
    case the bias blows up. This is because the first iterations act as a
    warm-start procedure for the $M$ chains of the last iteration: as $C$
    increases, $P_t$, $t<T$, decreases and eventually falls below the mixing
    time of the kernels $K_t$ (here, random-walk Metropolis kernels with a small
    step size), so that the last iteration starts far from its target, and the
    resulting burn-in is not absorbed by the larger $P_T$.
    Hence, allocating more CPU budget to the last iteration is beneficial, but
    one must make sure that the $P_t$, $t<T$, remain large relative to the
    mixing time of the kernels $K_t$.
    (We may be able to improve on these results by allowing $P_t$ to vary over time,
    in some adaptive manner, but we leave this to future work.)

    \begin{figure}[ht]
        \begin{center}
            \centerline{
                \includegraphics[height=5.5cm]{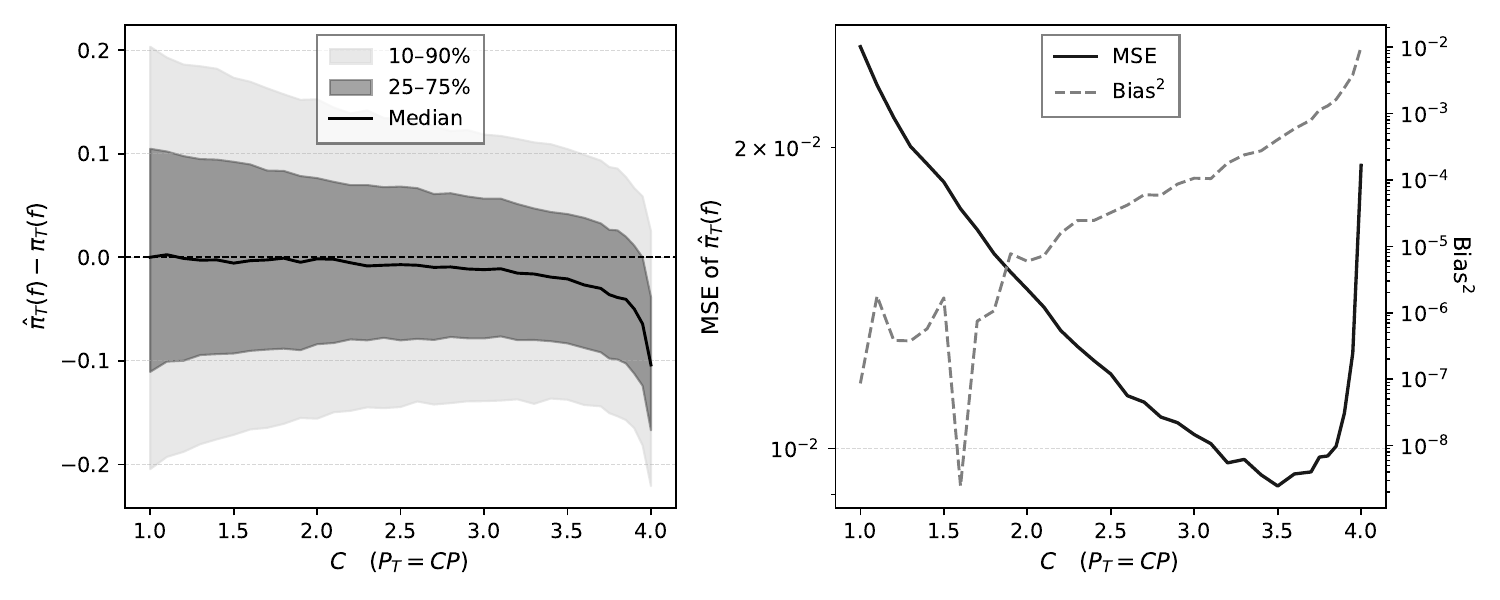}
            }
        \end{center}

        \caption{
            Quantiles (left) and MSE (mean squared error, right, log-scale) for
            the moment estimate errors $\hat{\pi}_T(f)-\pi_{T}(f)$
            ($f(x)=x_1$), for $\pi_T = \mathcal{N}((3, 0)^{\top}, I_2)$, starting from $\nu =
            \mathcal{N}(0_2, I_2)$,
            using a geometric annealing (tempering) path
            and random-walk Metropolis kernels with step size $0.2$.
            Each point is computed using $4 \times 10^3$ independent greedy
            waste-free SMC runs under a fixed budget and a final iteration
            allocation parameterised by $C$, as explained in the text.
        }
        \label{fig:wf_smc}
    \end{figure}

    In the tempering scenario, the usual practice of setting the tempering schedule
    so that the relative ESS (effective sample size) is greater than a certain constant is sound as
    this quantity is (up to a simple transformation) an estimate of the $\chi^2$ divergence
    in Assumptions~\ref{ass:chi} or~\ref{ass:chisquare_is_small} (the latter one corresponds to an ESS $\geq N/2$).
    This should lead to a sequence of length $T=\Theta(d^{1/2})$ (proven in the log-concave case, Assumption~\ref{ass:log_concave}, Theorem~\ref{th:chi_log_concave}).
    ESS-based adaptive tempering was initially proposed by
    \cite{MR2760137} and its asymptotic properties have been studied
    by \cite{Beskos2016} and \cite{CerouGuyader2016}.

    In case the end user wishes to estimate the normalising constants, we recommend
    keeping $P$ fixed across iterations. For tempering, the best complexity
    we managed to obtain, i.e. $\tilde\OO(d^{3/2}\varepsilon^{-2})$, for estimating $Z_T$,
    was for standard SMC, using the composition of ULMC and MALA kernels (Theorem~\ref{th:normalising_smc} and Corollary~\ref{cor:zhat_ulmc_mala}).
    Experiments suggest it may be the case that waste-free
    SMC is competitive with standard SMC in the same scenario, and establishing finite sample bounds for the former without
    relying on the spectral gap would be a first step toward closing the mismatch, but it remains an open problem.

    Regarding which estimators
    to use, $\hat{Z}^{\textup{med}}_T$ or $\hat{Z}_T$, no definite answer has been reached.
    Contrary to the product-of-medians estimator $\hat{Z}^{\textup{med}}_T$,
    the standard estimator $\hat{Z}_T$ is unbiased. This allows one to efficiently
    estimate the normalising constant by averaging over repeated runs.
    However, when the reweighting functions $G_t$ are heavy-tailed, meaning
    a few particles have disproportionately large weights, $\hat{Z}^{\textup{med}}_T$
    is more robust since taking the median estimated ratio over independent SMC runs
    prevents one weight from dominating. See~\cref{fig:comparison_zs}
    for a comparison of the two estimators at a fixed computational budget.

    \begin{figure}[ht]
        \begin{center}
            \centerline{
                \includegraphics[height=10.5cm]{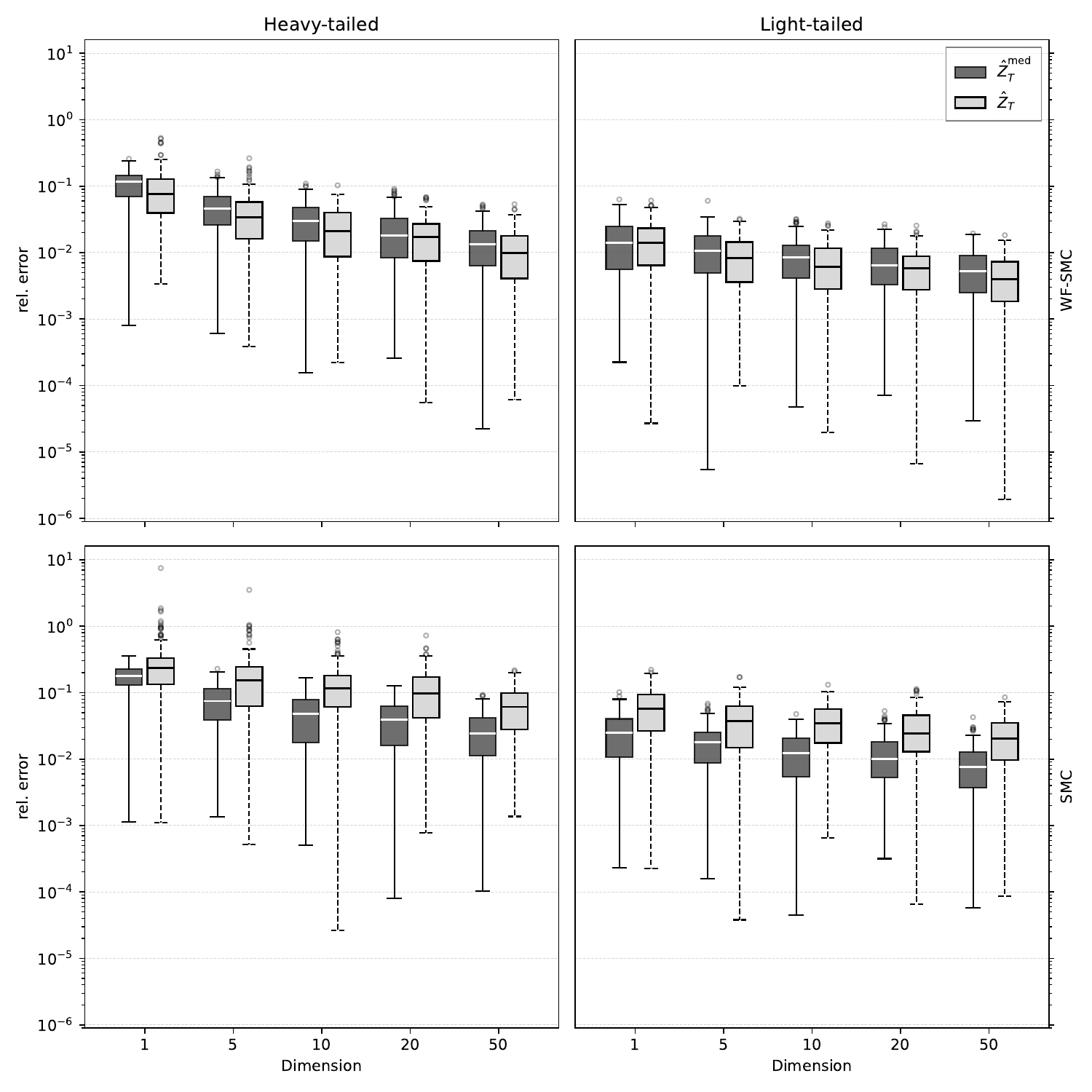}
            }
        \end{center}

        \caption{
            Relative error distributions (log-scale) for the estimators $\hat{Z}^{\textup{med}}_T$ and $\hat{Z}_T$, for $\pi_T = \mathcal{N}(1_d/2, \Sigma')$,
            starting from $\nu = \mathcal{N}(0_d, \Sigma)$, using a geometric annealing path (equidistant temperatures with $T\propto \sqrt{d}$)
            and random-walk Metropolis kernels (with proposal covariance $2.38^2/d I_d$, $\gamma=\Theta(d^{-1})$).
            At the top, waste-free SMC ($M=20$, $P\propto T^2/\gamma \propto d^2$, $J=10$), at the bottom, SMC ($M\propto T^2 \propto d$, $P\propto 1/\gamma$, $J=10$).
            At dimension $d$, the total budget allocated to $\hat{Z}^{\textup{med}}_T$ coincides with the total budget allocated to $\hat{Z}_T$.
            On the left, $\Sigma'\succ \Sigma$ (heavy-tailed reweighting functions), on the right $\Sigma'\prec \Sigma$. Each box is made out of $200$ independent draws for $\hat{Z}_T$ and $\hat{Z}^{\textup{med}}_T$.
        }
        \label{fig:comparison_zs}
    \end{figure}

    For waste-free SMC, our complexity bounds do not improve when $M$ grows.
    By default, we therefore recommend setting $M$ to a fixed value; for instance, to the
    number of nodes (independent processing units) in a parallel computing
    environment, as already advocated by~\cite{MR4400392}, as the $M$ chains may
    be simulated independently at each iteration $t$.
    An exception may be when the spectral gap $\gamma$ is exponentially small, for which we expect allocating the budget to $M$ to be more beneficial
    than allocating it to $P$, for both standard and waste-free SMC.
    For example, in case of multimodality with locally-mixing kernels, mass is moved between modes by the resampling
    steps rather than by the kernels~\citep{han2026timecomplexitysamplingmultimodaldistribution}, and one should take $M$ not too small to avoid
    mode collapse.

    The implementation of the numerical experiments is available
    at:~\url{https://github.com/ylefay/jaxSMC}; SMC and waste-free SMC are also
    implemented in~\url{https://github.com/nchopin/particles}.

    \subsection{Future work}\label{sub:future}

    A direction for future work is to determine how concentration scales jointly with the number of parallel chains $M$ and the chain length $P$.
    Our current finite-sample analysis of waste-free SMC essentially treats $M$ as a constant and does not yield sharp improvements as $M$ increases, despite the fact that each iteration produces $N=MP$ samples.
    The limitation stems from dependence induced by resampling, and it remains unclear whether one can recover an effective $1/(MP)$ variance rate rather than a $1/P$ rate.

\begin{acks}[Acknowledgments] 
  YLF acknowledges a CREST PhD scholarship. NC is supported partially by a
  hi! Paris chair from Institut Polytechnique de Paris. MV is supported by
  the Research Council of Finland (Finnish Centre of Excellence in Randomness
  and Structures, grants 346311 and 364216). We are grateful to Daniel
  Paulin and Pierre Jacob for insightful discussions at an early stage
  of this project, and to  the editor, the associate editor and the three
  referees for their careful reading and insightful remarks.
\end{acks}

    \appendix

    \section{Postponed proofs and technical results}
    \label{sec:proof_technical}
    \input{proofs_complexity}
    \bibliographystyle{imsart-nameyear}
    \bibliography{draft}

\end{document}

%% file: proofs_complexity.tex
\subsection{Finite-sample complexity for waste-free SMC (Theorems~\ref{th:main_wf},~\ref{th:wf_smc_greedy})}
\label{appendix:proof_thm_main_wf}
For a general overview of the proof of the theorem, we refer the reader to the dedicated Section~\ref{sec:proof}.

\subsubsection{Properties of the exact maximal distributional coupling (Lemma~\ref{lemma:coupling})}
\begin{proof}
    The existence of the distributional exact maximal coupling $(Y, \bar{Y})$ follows from~\citep[][Theorem 19.3.9]{Douc2018}.
    Then,
    \begin{equation}
        \label{eq:prp2}
        \Pr(R>r) = \textup{TV}(\tilde{\pi}K^{r-1}, \pi K^{r-1})
        = \textup{TV}(\tilde{\pi}K^{r-1} , \pi),
    \end{equation}
    where the first equality follows from the maximal coupling equality~\citep[][Lemma 19.3.6]{Douc2018}, the second equality follows from $K$ leaving $\pi$ invariant.
    By definition of the mixing time, if $r > \tau(\xi, \tilde{\pi}, K)$ then
    $\textup{TV}(\tilde{\pi}K^{r-1} , \pi)\leq \xi$.
\end{proof}

\subsubsection{Control for the meeting event of all chains (Lemma~\ref{lemma:prasrs})}

\begin{proof}
    By Lemma~\ref{lemma:marginal_is_warm}, conditional on $\bfC_{t-1}$, the
    resampled particles $\tilde{X}_{t}^{m}$ ($=Y_{t}^{m}[1]$) have marginal
    distribution $\tilde{\pi}_{t-1}(\dd x\mid \bfC_{t-1})$ which is
    $\Omega_{t-1}$-warm with respect to $\pi_{t-1}$.
    We take $r_{t}> \tau(\delta/M, \Omega_{t-1}, K_t)$, i.e., greater than the mixing time of $K_t$ starting from any $\Omega_{t-1}$-warm distribution.
    Then by definition of the mixing time,
    \begin{equation}
        \textup{TV}(\tilde{\pi}_{t-1}(\dd x\mid \bfC_{t-1}) K_t^{r_{t}-1}, \pi_{t-1}(\dd x))\leq \delta/M.
    \end{equation}
    Since $(Y_t^m,\bar{Y}_t^m)$ is maximal relative to $\tilde{\pi}_{t-1}(\cdot\mid \bfC_{t-1})$ (Section~\ref{sec:coupling}), the maximal coupling property gives $\Pr(Y_t^m[r_t]=\bar{Y}_t^m[r_t]\mid \bfC_{t-1})\geq 1-\delta/M$ for any $m=1,\ldots, M$.
    By the union bound, conditional on $\bfC_{t-1}$, all chains $Y_t^{1:M}$ are coupled with the stationary Markov chains $(\bar{Y}_t^{1:M})$ at time $r_{t}$ with probability at least $1-\delta$.
\end{proof}

\subsubsection{The conditional distribution of a chain is warm (Lemma~\ref{lemma:warmness_path_density})}
We first prove the warmness property for the conditional marginal distribution of the resampled particles,
and the warmness of the full path will follow from the Markov property.
\begin{lemma}
    \label{lemma:marginal_is_warm}
    Let $t\geq 1$.
    Assume $\Pr(\bfC_{t'})\geq 4/\omega_{t'}$ for some $\omega_{t'}>4$, and all $t'=0,\ldots, t-1$.
    Then, the conditional marginal distribution of the resampled particles $\tilde{X}_{t}^{1:M}$, defined for any $B\in \ambientalgebra$ by
    \begin{equation}
        \tilde{\pi}_{t-1}(B \mid \bfC_{t-1}) = \Pr(\tilde{X}_t^{m}\in B\mid \bfC_{t-1}),
    \end{equation}
    is warm with respect to $\pi_{t-1}$ with warmness $\Omega_{t-1}$ satisfying the recurrence relation
    \begin{equation}
        \label{eq:rec_omega1}
        \Omega_{t-1} =\frac{4(r_{t-1}-1)}{P} \frac{\Pr(\bfC_{t-2})}{\Pr(\bfC_{t-1})}\Omega_{t-2}+\omega_{t-1},\indent \Omega_{-1} = 1.
    \end{equation}
    \begin{proof}
        Let $B\subset \mathcal{X}$ be a measurable set, $t\geq 1$, and fix an integer sequence $\bfr = (r_0,\ldots, r_{t-1})$.
        Let $\tilde{\pi}_{t-1}(\dd x \mid \bfC_{t-1})$  be the conditional distribution of the resampled particles at iteration $t$, then
        \begin{equation}
            \label{eq:pism1B}
            \begin{split}
                \MoveEqLeft \tilde{\pi}_{t-1}(B\mid \bfC_{t-1})  \\
                &= \Pr(\tilde{X}_{t}^m\in B\mid \bfC_{t-1})\\
                &=\sum_{n=1}^N \bbE\left[\frac{G_{t-1}(X_{t-1}^n)}{\sum_{k=1}^N
                G_{t-1}(X_{t-1}^k)}\mathds{1}_B(X_{t-1}^n)\middle| \bfC_{t-1}\right]\\
                &=\sum_{m=1}^{M}\sum_{p=1}^{P}\bbE\left[\frac{G_{t-1}(Y_{t-1}^m[p])}{\sum_{m'=1}^{M}\sum_{j=1}^{P}
                    G_{t-1}(Y_{t-1}^{m'}[j])}\mathds{1}_B(Y_{t-1}^m[p])\middle| \bfC_{t-1}\right]\\
                &\leq \frac{4}{N \pi_{t-2}(G_{t-1})}\sum_{m=1}^M \sum_{p=1}^P
                \bbE\left[G_{t-1}(Y_{t-1}^m[p])\mathds{1}_B(Y_{t-1}^m[p])\middle| \bfC_{t-1}\right],
            \end{split}
        \end{equation}
        where to go from the third to the fourth line, we use that conditional on $\bfC_{t-1}$,
        \begin{equation}
            \begin{split}
                \sum_{k=1}^M \sum_{j=1}^P G_{t-1}(Y^k_{t-1}[j])=N\widehat{\pi}_{t-2}(G_{t-1})\geq \frac{1}{4}N\pi_{t-2}(G_{t-1}).
            \end{split}
        \end{equation}
        For any $1\leq m\leq M$, let $R_{t-1}^m=\inf\{p\geq 1 : Y_{t-1}^m[p] =\bar{Y}_{t-1}^m[p]\}$ be the meeting time for chain $m$, we have
        \begin{equation}
            \label{eq:tildepis}
            \begin{split}
                \bbE\left[ \sum_{p=1}^{P}G_{t-1}(Y_{t-1}^m[p])
                        \mathds{1}_B(Y_{t-1}^m[p])\middle| \bfC_{t-1}\right]&=
                \bbE\left[\sum_{p=1}^{R_{t-1}^m-1}G_{t-1}(Y_{t-1}^m[p])\mathds{1}_B(Y_{t-1}^m[p])\middle| \bfC_{t-1}\right]\\
                &\indent + \bbE\left[\sum_{p=R_{t-1}^m}^P
                G_{t-1}(\bar{Y}_{t-1}^m[p])\mathds{1}_B(\bar{Y}_{t-1}^m[p])\middle| \bfC_{t-1}\right].
            \end{split}
        \end{equation}
        Assume that $\bfr$ is taken such that $\Pr(\bfC_{t-1})\geq
        4/\omega_{t-1}$, then conditional on $\bfC_{t-1}$, the contribution of the stationary terms (after $R_{t-1}^m$) is bounded by
        \begin{equation}
            \begin{split}
                \label{eq:tildepis1}
                \MoveEqLeft\bbE\left[\sum_{p=R_{t-1}^m }^P
                G_{t-1}(\bar{Y}_{t-1}^m[p])\mathds{1}_B(\bar{Y}_{t-1}^m[p])
                                   \bigg| \bfC_{t-1}\right] \\
                &\leq  \frac{1}{\Pr(\bfC_{t-1})} \sum_{p=1}^P \bbE\left[G_{t-1}(\bar{Y}_{t-1}^m[p])\mathds{1}_B(\bar{Y}_{t-1}^m[p])\right]\\
                & = \frac{P}{\Pr(\bfC_{t-1})} \int G_{t-1}(\bar{y})\mathds{1}_B(\bar{y}) \pi_{t-2}(\dd \bar{y})\\
                &\leq \frac{\omega_{t-1} P}{4}\int G_{t-1}(\bar{y})\mathds{1}_B(\bar{y}) \pi_{t-2}(\dd \bar{y}),
            \end{split}
        \end{equation}
        where to go from the second to the third line we use
        $\bar{Y}_{t-1}^m[p]\sim \pi_{t-2}$. Remark that,
        \begin{equation}
            \begin{split}
                \label{eq:computationgpi}
                \int \frac{G_{t-1}(\bar{y})}{\pi_{t-2}(G_{t-1})}\mathds{1}_B(\bar y) \pi_{t-2}(\dd \bar{y})&=\frac{1}{\pi_{t-2}(G_{t-1})} \int \frac{g_{t-1}(\bar y)}{g_{t-2}(\bar{y})}\mathds{1}_B(\bar{y})\frac{g_{t-2}(\bar{y})\nu(\dd \bar{y})}{\nu(g_{t-2})}\\
                &= \left(\frac{\nu(g_{t-1})}{\nu(g_{t-2})}\right)^{-1}\frac{1}{\nu(g_{t-2})}\int g_{t-1}(\bar{y})\mathds{1}_B(\bar{y})\nu(\dd \bar{y})\\
                &=\frac{1}{\nu(g_{t-1})}\int\mathds{1}_B(\bar{y})g_{t-1}(\bar{y})\nu(\dd \bar{y})\\
                &=\pi_{t-1}(B).
            \end{split}
        \end{equation}

        At this stage of the proof, we have successfully tackled all the terms
        in~\eqref{eq:pism1B} after the coupling events; that is, for all $t\geq
        1$, we have shown that
        \begin{equation}
            \begin{split}
                \label{eq:almost_finished}
                \indent &\tilde{\pi}_{t-1}(B\mid \bfC_{t-1}) \\
                &\leq \frac{4}{N\pi_{t-2}(G_{t-1})}\sum_{m=1}^M \bbE\left[\sum_{p=1}^{P}\mathds{1}[p< R_{t-1}^m] G_{t-1}(Y_{t-1}^m[p])\mathds{1}_{B}(Y_{t-1}^m[p])\bigg| \bfC_{t-1}\right] \\
                &\indent + \omega_{t-1}\pi_{t-1}(B).
            \end{split}
        \end{equation}
        Using $\bfC_{t-2}\supset \bfC_{t-1}$,
        \begin{equation}
            \label{eq:get_rid_of_condition_fixed}
            \begin{split}
                &\indent\bbE\left[\mathds{1}[p< R_{t-1}^m] \frac{G_{t-1}(Y_{t-1}^m[p])}{\pi_{t-2}(G_{t-1})}\mathds{1}_{B}(Y_{t-1}^m[p])\mathds{1}_{\bfC_{t-1}}\right]\\
                &\leq \bbE\left[\frac{G_{t-1}(Y_{t-1}^{m}[p])}{\pi_{t-2}(G_{t-1})} \mathds{1}_B(Y_{t-1}^{m}[p])\bigg| \bfC_{t-2}\right] \Pr(\bfC_{t-2})\\
                &            =\int \frac{G_{t-1}(y)}{\pi_{t-2}(G_{t-1})}\mathds{1}_B(y)\tilde{\pi}_{t-2}K_{t-1}^{p-1}(\dd y\mid \bfC_{t-2}) \Pr(\bfC_{t-2}),
            \end{split}
        \end{equation}
        and thus,
        \begin{equation}
            \begin{split}
                &\indent \bbE\left[\mathds{1}[p< R_{t-1}^m] \frac{G_{t-1}(Y_{t-1}^m[p])}{\pi_{t-2}(G_{t-1})}\mathds{1}_{B}(Y_{t-1}^m[p])\bigg| \bfC_{t-1}\right]\\
                &\leq \frac{\Pr(\bfC_{t-2})}{\Pr(\bfC_{t-1})} \int \frac{G_{t-1}(y)}{\pi_{t-2}(G_{t-1})}\mathds{1}_B(y)\tilde{\pi}_{t-2}K_{t-1}^{p-1}(\dd y\mid \bfC_{t-2}).
            \end{split}
        \end{equation}
        This expectation is $0$ for $p\geq r_{t-1}$ since then $R_{t-1}^m \leq R_{t-1} \leq r_{t-1}\leq p$ conditioned on $\bfC_{t-1}$.

        We now proceed by recurrence.
        For the initialisation, recall that $Y_0^{1:M}[1]\sim  \nu^{\otimes M}(\dd y)$, and thus both $(Y_0^{1:M})$ and $(\bar{Y}_0^{1:M})$ are stationary.
        Since the coupling is maximal, they all meet at $R_0=1$ almost surely.
        Combining this with~\eqref{eq:tildepis} and~\eqref{eq:tildepis1} shows that for any $B\subset \ambientspace$, provided $r_0\geq R_0=1$ is such that $\Pr(\bfC_0)\geq 4/\omega_0$, we have
        \begin{equation}
            \tilde{\pi}_0(B\mid \bfC_0)\leq \omega_0 \pi_0(B) \leq \Omega_0\pi_0(B).
        \end{equation}
        Let $t\geq 1$, and assume that $\tilde{\pi}_{t-1}(\dd x\mid \bfC_{t-1})$ is $\Omega_{t-1}$-warm. %
        For any $p\geq 1$,
        \begin{equation}
            \begin{split}
                \label{eq:fg1pi0g1p}
                &\indent \int \frac{G_{t}(y)}{\pi_{t-1}(G_t)} \mathds{1}_B(y) \tilde{\pi}_{t-1}K_{t}^{p-1}(\dd y\mid \bfC_{t-1}) \\
                &\leq \Omega_{t-1}\int \frac{G_{t}(y)}{\pi_{t-1}(G_t)} \mathds{1}_B(y) \pi_{t-1}K_{t}^{p-1}(\dd y)\\
                &=\Omega_{t-1}\pi_t(B),
            \end{split}
        \end{equation}
        where the first inequality stems from the warmness property on $\tilde{\pi}_{t-1}(\dd y \mid \bfC_{t-1})$, the second from $K_t$ leaving $\pi_{t-1}$ invariant, and computations done in~\eqref{eq:computationgpi}.
        Combining the previous equality~\eqref{eq:fg1pi0g1p} with~\eqref{eq:get_rid_of_condition_fixed},~\eqref{eq:almost_finished}, we obtain
        \begin{equation}
            \tilde{\pi}_t(B\mid \bfC_t)\leq \underbrace{\left(\frac{4(r_t-1)}{P} \frac{\Pr(\bfC_{t-1})}{\Pr(\bfC_t)}\Omega_{t-1}+\omega_t\right)}_{=\Omega_t}\pi_t(B).
        \end{equation}
        We have shown that $\tilde{\pi}_{t}(\dd x\mid \bfC_t)$ is $\Omega_t$-warm with respect to $\pi_t$ given $\Pr(\bfC_t)\geq 4/\omega_t$ and $\tilde{\pi}_{t-1}(\dd x\mid \bfC_{t-1})$ is $\Omega_{t-1}$-warm.
        This proves the lemma.
    \end{proof}
\end{lemma}
\begin{proof}[Proof of Lemma~\ref{lemma:warmness_path_density}]
    Remark that $Y_t^1 \mid \bfC_{t-1}$ is a Markov chain with initial state $Y_t^1[1] \mid \bfC_{t-1} \sim \tilde{\pi}_{t-1}(\dd y \mid \bfC_{t-1})$, and Markov kernel $K_t$.
    Let $\tilde{\Pi}^{(t-1)}_{\mid \bfC_{t-1}, K_t}$ be its path density truncated to the $P$-th term:
    \begin{equation*}
        \tilde{\Pi}^{(t-1)}_{\mid \bfC_{t-1}, K_t}(\dd Y) = \tilde{\pi}_{t-1}(\dd Y[1]\mid \bfC_{t-1})\prod_{p=2}^P K_t(Y[p-1], \dd Y[p]).
    \end{equation*}
    Let similarly $\Pi^{(t-1)}_{K_t}$ be the path density of the (truncated) stationary Markov chain, then
    \begin{equation*}
        \frac{\tilde{\Pi}^{(t-1)}_{\mid \bfC_{t-1}, K_t}(\dd Y)}{\Pi^{(t-1)}_{K_t}(\dd Y)} = \frac{\tilde{\pi}_{t-1}(\dd Y[1]\mid \bfC_{t-1})}{\pi_{t-1}(\dd Y[1])}\leq \Omega_{t-1},
    \end{equation*}
    where inequality is due to the warmness Lemma~\ref{lemma:marginal_is_warm}.
\end{proof}
A useful lemma immediately follows from Lemma~\ref{lemma:warmness_path_density}.
\begin{lemma}
    \label{lemma:prod}
    Let $\bar{Y}_t$ be a (truncated) Markov chain starting at stationarity.
    For any measurable function $h:\ambientspace^P\to \bbR^+$,
    \begin{equation}
        \bbE\left[\prod_{m=1}^M h(Y^m_t)\bigg| \bfC_{t-1}\right] \leq \Omega_{t-1}\bbE\left[h(\bar{Y}_t)^M\right].
    \end{equation}
    \begin{proof}
        The variables $Y_t^{1:M}$ are exchangeable, and conditioning on
        $\bfC_{t-1}$ preserves the exchangeability, thus
        $\bbE\left[h(Y^m_t)^M\mid \bfC_{t-1}\right]$ does not depend on $m$.
        Since $h\geq 0$, Hölder's inequality gives
        \begin{align*}
            \bbE\left[\prod_{m=1}^M h(Y^m_t)\mid \bfC_{t-1}\right]
            & \leq \left\{\prod_{m=1}^M\left(\bbE\left[h(Y^m_t)^M\mid \bfC_{t-1}\right]\right)\right\}^{1/M} \\
            & \leq \bbE\left[h(Y^1_t)^M\mid \bfC_{t-1}\right]\\
            &\leq \Omega_{t-1}\bbE\left[h(\bar{Y}_t)^M\right].
        \end{align*}
        where we used Lemma~\ref{lemma:warmness_path_density} in the last
        inequality.
    \end{proof}
\end{lemma}

\subsubsection{Gaussian concentration under spectral gap (Lemma~\ref{lemma:concentration})}
\label{subsubsec:proof_of_concentration}
We now prove the concentration result for ergodic averages of bounded functions given by Lemma~\ref{lemma:concentration}.
\begin{proof}[Proof of Lemma~\ref{lemma:concentration}]
    Let $\bar{f} = f-\pi_{t-1}(f)$.
    Let $\bar{Y}_t=\bar{Y}_t^1$, then $\bar{Y}_t$ is a stationary $\pi_{t-1}$-reversible Markov chain with kernel $K_t$.
    Since $K_t$ has spectral gap $\gamma_t>0$,~\citet[Th. 1]{JMLR:v22:19-479}
    state that, for any $u>0$,
    \begin{equation}
        \label{eq:mgf_one}
        \begin{split}
            \bbE\left[\exp\left(u \sum_{p=1}^P \bar{f}(\bar{Y}_t[p])\right)\right]&\leq \exp\left(\frac{2-\gamma_t}{\gamma_t}\frac{u^2}{2}P\norm{f}_{\infty}^2\right)\\
            &\leq \exp\left(\frac{\norm{f}_{\infty}^2}{\gamma_t}u^2P \right).
        \end{split}
    \end{equation}
    From Lemma~\ref{lemma:prod} applied to $h(Y) =
    \exp\left(\frac{u}{MP}\sum_{p=1}^P \bar{f}(Y[p])\right)$ and the above inequality,
    \begin{equation*}
        \begin{split}
            \bbE\left[\exp\left(\frac{u}{MP}\sum_{m=1}^M\sum_{p=1}^P \bar{f}(Y_t^m[p])\right)\bigg|\bfC_{t-1}\right] &\leq %
            \Omega_{t-1} \bbE[h(\bar{Y}_t)^M] \\
            &\leq \Omega_{t-1} \exp\left(\frac{\norm{f}_{\infty}^2}{\gamma_t}\frac{u^2 }{P}\right).
        \end{split}
    \end{equation*}
    Via Chernoff's argument, for any random variable $Z$ such that $\bbE[e^{u Z}] \leq e^{au^2}$ for any $u \in \bbR$, and some $a>0$, we have for any $\epsilon > 0$, $\Pr[\abs{Z} \geq \epsilon] \le 2e^{-\epsilon^2/(4a)}$.
    This yields,
    \begin{equation}
        \label{eq:final_inequality_concentration}
        \begin{split}
            \Pr\left[\left|\frac{1}{M}\sum_{m=1}^M
                         \frac{1}{P}\sum_{p=1}^{P}\bar{f}(Y_t^m[p])\right|\geq\varepsilon
                   \middle| \bfC_{t-1}\right]
            \leq  2\Omega_{t-1}\exp\left(-\frac{P \gamma_t\varepsilon^2}{4\norm{f}_{\infty}^2} \right).
        \end{split}
    \end{equation}
    Provided that $P\geq \frac{4\norm{f}_{\infty}^2}{\gamma_t \varepsilon^2}\log(\frac{2\Omega_{t-1}}{\delta'})$,
    \begin{equation*}
        \begin{split}
            \label{eq:concentrationconditionedbfa}
            \Pr(\abs{\widehat{\pi}_{t-1}(f)-\pi_{t-1}(f)}\geq \varepsilon \mid \bfC_{t-1})&\leq \delta'.
        \end{split}
    \end{equation*}
    Therefore,
    \begin{equation}
        \begin{split}
            \Pr\left(\abs{\widehat{\pi}_{t-1}(f)-\pi_{t-1}(f)}< \varepsilon\right)
            &\geq \Pr\left(\abs{\widehat{\pi}_{t-1}(f)-\pi_{t-1}(f)}< \varepsilon, \bfC_{t-1}\right)\\
            &= \Pr(\abs{\widehat{\pi}_{t-1}(f)-\pi_{t-1}(f)}< \varepsilon \mid \bfC_{t-1})\Pr(\bfC_{t-1})\\
            &\geq (1-\delta')\Pr(\bfC_{t-1}).
        \end{split}
    \end{equation}
\end{proof}
We now turn to a Bernstein-type concentration inequality to bound the probability of $\bfB_{t}$ conditionally on $\bfC_{t-1}$.%

\subsubsection{Gaussian concentration of signed functions under spectral gap}
Below, we prove a useful technical lemma to deal with concentration of signed functions.
This is required to prove Lemma~\ref{lemma:concentration_of_weights}, see Section~\ref{subsec:concentration_spgap}.
\begin{lemma}[Technical lemma]
    \label{lemma:technical}
    Let $(X_p)$ be a $\pi$-stationary Markov chain with spectral gap $\gamma >0$.
    Let $g :\ambientspace \to \bbR$ be a square-integrable non-positive function, and $\mu = \bbE_{\pi}[g(X)]$, $\sigma^2 = \Var_{\pi}[g(X)]$.
    Let $\bar{g}=g-\mu$, and $\bar{g}^+ = \max(0, \bar{g})$ the positive part of $\bar{g}$.
    Then for any $s \geq 0$,
    \begin{multline*}
        \inf_{\theta > 0} e^{-s P \theta} \bbE\left[\exp\left\{\theta
                                                            \left(\sum_{p=1}^P\bar{g}^+(X_p)-\bbE[\bar{g}^+(X_p)]\right)\right\}\right]
        \\
        \leq \exp\left\{-\frac{P\gamma}{2}
                     \min\left(\frac{s^2}{4\sigma^2},-\frac{s}{10\mu}\right)\right\}.
    \end{multline*}
    \begin{proof}
        Notice that $0\leq \bar{g}^+ \leq -\mu$ since $g\leq 0$, and $\abs{\bar{g}^+-\bbE[\bar{g}^+]}\leq -\mu$, and $\Var[\bar{g}^+]\leq \sigma^2$.
        From the moment-generating function bound~\citep[][Theorem
        1]{jiang2025bernsteinsinequalitiesgeneralmarkov} for Markov chain
        ergodic averages of bounded functions, we have, for any $s \geq 0$,
        \begin{equation}
            \begin{split}
                \label{eq:chernoff_tilded_2}
                \inf_{\theta> 0}e^{-s P \theta}\bbE\left[e^{\theta (\sum_{p=1}^P \bar{g}^+_p-\bbE[\bar{g}^+_p])}\right] &\leq \exp\left(-\frac{\gamma P s^2/2}{(2-\gamma)\sigma^2 - 5s \mu}\right)\\
                &\leq \exp\left(-\frac{P\gamma}{2} \min\left(\frac{s^2}{4\sigma^2}, -\frac{s}{10\mu}\right)\right),
            \end{split}
        \end{equation}
        where we used $1/(a+b)\geq \frac{1}{2} \min(1/a, 1/b)$ for any $a, b >0$ and $2-\gamma\leq 2$. %
    \end{proof}
\end{lemma}

\subsubsection{One-sided deviations of estimated ratio of normalising constants (Lemma~\ref{lemma:concentration_of_weights})}

We do not distinguish the case $t=0$ as it is handled by the next lemma with $\Omega_{-1} = 1$, $\gamma_0=1$, and $\bfC_{-1}$ any almost-sure event.

\begin{proof}[Proof of Lemma~\ref{lemma:concentration_of_weights}]
    Let $\bar{G}_t = G_t-\pi_{t-1}(G_t)$.
    For the sake of notation, we omit the subscript in $t$, and let $Y_t^m[p]=Y_p^m$.
    We let $\bar{Y}$ be a stationary Markov chain, we let $g = -G_t < 0$, and $\bar{g}=-\bar{G}_t = g - \bbE_{\pi_{t-1}}[g]$.
    We know from Markov's inequality that for any $s\geq 0$, $\theta > 0$
    \begin{equation}
        \begin{split}
            \Pr\left[ \frac{1}{M}\sum_{m=1}^M \frac{1}{P}\sum_{p=1}^P \bar{g}(Y_p^m) \geq s \mid \bfC_{t-1}\right] &\leq e^{-N s \theta } \bbE\left[e^{\theta (\sum_{m=1}^M \sum_{p=1}^P \bar{g}(Y_p^m))}\mid \bfC_{t-1}\right]\\
            &\leq \Omega_{t-1}e^{-Ps (M\theta)} \bbE\left[e^{(M \theta) \sum_{p=1}^P \bar{g}(\bar{Y}_p)}\right]\\
            &\leq \Omega_{t-1}e^{-P s (M\theta)} \bbE\left[e^{(M \theta ) \sum_{p=1}^P \bar{g}^+(\bar{Y}_p)}\right]\\
            &= \Omega_{t-1} e^{-P(s-\bbE[\bar{g}^+(\bar{Y}_p)])(M\theta)}\\
            &\indent \times \bbE\left[e^{(M\theta) \sum_{p=1}^P \bar{g}^+(\bar{Y}_p)-\bbE[\bar{g}^+(\bar{Y}_p)]}\right]
        \end{split}
    \end{equation}
    where we use Lemma~\ref{lemma:prod} to go from the first to the second line, and $e^{x}\leq e^{x^{+}}$ for any $x\in \bbR$ from the second to the third.
    The previous Lemma~\ref{lemma:technical} applied to $s' = s-\bbE[\bar{g}^+]\geq 0$ yields
    \begin{equation}
        \label{eq:prlfam}
        \begin{split}
            &\indent \Pr\left[ \frac{1}{M}\sum_{m=1}^M \frac{1}{P}\sum_{p=1}^P \bar{g}(Y_p^{m}) \geq s \mid \bfC_{t-1} \right] \\
            &\leq\Omega_{t-1} \left\{\inf_{\theta> 0} \left(e^{-P(s-\bbE[\bar{g}^+])\theta}\bbE\left[e^{\theta \sum_{p=1}^P \bar{g}^+(\bar{Y}_p)-\bbE[\bar{g}^+(\bar{Y}_p)]}\right]\right)\right\}\\
            &\leq \Omega_{t-1}  \exp\left(-\frac{P \gamma_t}{2} \min\left(\frac{(s-\bbE[\bar{g}^{+}])^2}{4\Var_{\pi_{t-1}}(G_t)},\frac{s-\bbE[\bar{g}^+]}{10\pi_{t-1}(G_t)}\right)\right).
        \end{split}
    \end{equation}
    Remark that $\bar{g}=\bar{g}^{+}-\bar{g}^{-}$, and $\bbE[\bar{g}]=0$, thus $\bbE[\bar{g}^{+}] = (\bbE[\bar{g}^{+}]+\bbE[\bar{g}^{-}])/2$, but $\bar{g}^{+} +\bar{g}^{-} = \abs{\bar{g}}$, and $\bbE[\abs{\bar{g}}]\leq \sigma$ by Jensen's inequality, thus $\bbE[\bar{g}^{+}]\leq  \sigma/2$.
    We can replace any $\bbE[\bar{g}^{+}]$ by $\sigma/2$ in~\eqref{eq:prlfam}, for any $s\geq \sigma/2$,
    \begin{equation}
        \label{eq:prlfam2}
        \begin{split}
            &\indent \Pr\left[ \frac{1}{M}\sum_{m=1}^M \frac{1}{P}\sum_{p=1}^P \bar{g}(Y_p^{m}) \geq s \mid \bfC_{t-1} \right]\\
            &\leq \Omega_{t-1} \exp\left(-\frac{P\gamma_t}{2} \min\left(\frac{(s-\sigma/2)^2}{4\sigma^2}, \frac{s-\sigma/2}{10\pi_{t-1}[G_t]}\right)\right).
        \end{split}
    \end{equation}
    Let $\beta \in (0, 1)$, and let $ s= \beta \pi_{t-1}[G_t]>0$, then $(s-\sigma/2)^2 /(4\sigma^2) = (\beta/(2\chi)-1/4)^2 $, with $\chi = \sqrt{\chi^2(\pi_t\mid \pi_{t-1})}$.
    Assume that $s\geq \sigma/2$, or equivalently $\beta \geq \chi/2$,~\eqref{eq:prlfam2} specialises to
    \begin{equation}
        \begin{split}
            \label{eq:concentrationsumgt}
            &\indent \Pr\left[ \frac{1}{M}\sum_{m=1}^M \frac{1}{P}\sum_{p=1}^P G_{t}(Y_p^{m}) \leq (1-\beta)\pi_{t-1}[G_t] \mid \bfC_{t-1}\right] \\
            &\leq \Omega_{t-1}\exp\left(-\frac{P\gamma_t}{2} \min\left(\left(\frac{\beta}{2\chi}-\frac{1}{4}\right)^2, \frac{\beta-\chi/2}{10}\right)\right).
        \end{split}
    \end{equation}
    Assume $\chi^2 \leq 1$ (Assumption~\ref{ass:chisquare_is_small}), take $\beta = 3/4$ in~\eqref{eq:concentrationsumgt}, then numerical computations yield that the RHS of~\eqref{eq:concentrationsumgt} specialises to $\Omega_{t-1}\exp(-P\gamma_t/128)$.

    Take $P\geq \frac{128}{\gamma_t}\log(\frac{\Omega_{t-1}}{\delta'})$,  then with probability at least $1-\delta'$, $\widehat{\pi}_{t-1}(G_t)\geq \pi_{t-1}(G_t)/4$.
\end{proof}

\subsubsection{Union bound and finishing the induction (Lemma~\ref{lemma:main_thm_technical})}
Before proving Lemma~\ref{lemma:main_thm_technical}, we prove an intermediary result.
\begin{lemma}
    \label{lemma:csr_recurrence}
    Take $r_t> \tau(\delta/M, \Omega_{t-1}, K_t)$, and $P\geq \frac{128}{\gamma_t}\log(\frac{\Omega_{t-1}}{\delta'})$.
    Then $\Pr(\bfC_{t})\geq (1-\delta-\delta')\Pr(\bfC_{t-1})$.
    \begin{proof}
        For $t\geq 1$,
        Lemma~\ref{lemma:concentration_of_weights} yields
        $            \Pr(\bfB_{t} \mid \bfC_{t-1}) \geq 1-\delta'$,
        provided we take $P$ as specified in the lemma.
        Lemma~\ref{lemma:prasrs} shows
        $          \Pr(\bfA_{t}\mid \bfC_{t-1})\geq 1-\delta$,
        provided $r_t>\tau(\delta/M, \Omega_{t-1}, K_t)$.
        Therefore by a union bound, we have
        \begin{equation}
            \begin{split}
                \Pr(\bfC_{t}) &=\Pr(\bfB_{t} \cap \bfA_{t}\cap \bfC_{t-1})\\
                &=\Pr(\bfB_{t} \cap \bfA_{t}\mid \bfC_{t-1})  \Pr(\bfC_{t-1}) \\
                &\geq (1-\delta-\delta')\Pr(\bfC_{t-1}).
            \end{split}
        \end{equation}
        For $t=0$, since $Y_0^{1:M}\sim (\nu^{\otimes P})^{\otimes M}$, all chains are coupled at time $R_0=r_0=1$ almost surely, so that $\Pr(\bfA_0)=1$, and therefore $\Pr(\bfC_0)=\Pr(\bfB_0)\geq 1 -\delta' \geq 1 -\delta'-\delta$ by the
        initialisation case of Lemma~\ref{lemma:concentration_of_weights}.
    \end{proof}
\end{lemma}
The proof of the technical version of Theorem~\ref{th:main_wf} (Lemma~\ref{lemma:main_thm_technical}) follows from recursively applying Lemma~\ref{lemma:csr_recurrence} and Lemma~\ref{lemma:marginal_is_warm}, and then Lemma~\ref{lemma:concentration} at final iteration.
\begin{proof}[Proof of Lemma~\ref{lemma:main_thm_technical}]
    We proceed by recurrence to show first that for choices of $P$ and $r_t$ prescribed by Lemma~\ref{lemma:csr_recurrence}, $\Pr(\bfC_{T-1})\geq (1-\delta-\delta')^{T}$.
    For $t=0$, let $r_0 = 1$, and $P$ as given by Lemma~\ref{lemma:csr_recurrence}, then $\Pr(\bfC_0)\geq 1 -\delta'-\delta$, let $\omega_0 = 4/(1-\delta'-\delta)$.
    The general case $t\geq 1$ follows by recursively applying Lemma~\ref{lemma:csr_recurrence}, which ensures $
    \Pr(\bfC_{t}\mid \bfC_{t-1})\geq (1-\delta'-\delta)$ provided the choices for $P$ and $r_t$ are appropriate, and Lemma~\ref{lemma:marginal_is_warm} with $\omega_t = 4/(1-\delta-\delta')^{t+1}$.

    Let $f:\ambientspace\to \bbR$ with $\norm{f}_{\infty}\leq 1$, fix $\varepsilon >0$.
    By Lemma~\ref{lemma:concentration},
    \begin{equation}
        \begin{split}
            \Pr(\abs{\widehat{\pi}_{T-1}(f)-\pi_{T-1}(f)}<\varepsilon)&\geq (1-\delta')\Pr(\bfC_{T-1})\\
            &\geq (1-\delta-\delta')^{T}(1-\delta')\\
            &\geq 1-T\delta - (T+1)\delta'.
        \end{split}
    \end{equation}
    We take $\delta=\eta/(2T)$ and $\delta' = \delta/2 = \eta/(4T)$, then $\Pr(\abs{\widehat{\pi}_{T-1}(f)-\pi_{T-1}(f)}<\varepsilon)\geq 1-\eta$.
\end{proof}

\subsubsection{Upper bounds on the warmness parameters of the resampled distributions}
For a specific choice of sequence $(r_t)$, the warmness parameters can be uniformly bounded by a constant.
\begin{lemma}
    \label{lemma:bound_warmness}
    Let $\delta = \eta/(2T)$.
    Let $r_t = \tau(\delta/M, \Omega_{t-1}, K_t)  + 1$ for $t=0,\ldots, T$.
    Let $P\geq 32\tau(\delta/M, 32)$ and $P$ satisfy the requirement of Lemma~\ref{lemma:csr_recurrence}.
    Then $\Omega_{t-1}\leq 32$ for all $t=1,\ldots, T$.

    \begin{proof}
        For any $t\geq 0$, let $r_t= \tau(\delta/M, \Omega_{t-1}, K_t)+1$,
        where $\Omega_t$ satisfies~\eqref{eq:rec_omega}. %

        We want to construct a sequence $(\bar{\Omega}_t)$ such that
        \begin{equation}
            \label{eq:brstau2dm}
            \tau(\delta/M, \bar{\Omega}_{t-1}, K_t)\geq \tau(\delta/M, \Omega_{t-1}, K_t) = r_t-1.
        \end{equation}
        Since the mixing time is an increasing function of the warmness, if we can construct a sequence $\bar{\Omega}$ satisfying $\bar{\Omega}_t\geq \Omega_t$, then the right-hand side of the previous inequality is automatically satisfied.

        We proceed by recurrence to construct such a sequence.
        Since $r_0 = 1$, we have $\Omega_0=\omega_0=\bar{\Omega}_0$, therefore from~\eqref{eq:rec_omega} and $r_1 = \tau(\delta/M, \Omega_{0}, K_1) +1$, we have
        \begin{equation}
            \Omega_1 = \frac{4\tau(\delta/M, \Omega_0, K_1)}{P}\frac{\Pr(\bfC_0)}{\Pr(\bfC_1)}\bar{\Omega}_0 + \omega_1.
        \end{equation}
        Provided $P\geq \frac{128}{\gamma_1}\log(\frac{\Omega_0}{\delta'})$, $\Pr(\bfC_0)/\Pr(\bfC_1)\leq 1/(1-\delta-\delta')$ by Lemma~\ref{lemma:csr_recurrence}.
        Take $P\geq 8\tau(\delta/M, \bar{\Omega}_0, K_1)/(1-\delta-\delta')$, and let $\bar{\Omega}_1 = \bar{\Omega}_0/2 + \omega_1 = \omega_0/2+\omega_1$, then $\Omega_1\leq \bar{\Omega}_1$, we are done for the initialisation.
        Let $t\geq 1$, and assume that $\Omega_t\leq \bar{\Omega}_t$, and $\bar{\Omega}_t = \bar{\Omega}_{t-1}/2 + \omega_t$.
        From~\eqref{eq:rec_omega}, we have
        \begin{equation}
            \begin{split}
                \Omega_{t+1}&= \frac{4\tau(\delta/M, \Omega_t, K_{t+1})}{P} \frac{\Pr(\bfC_t)}{\Pr(\bfC_{t+1})}\Omega_{t} + \omega_{t+1} \\
                &\leq \frac{4\tau(\delta/M, \bar{\Omega}_t, K_{t+1})}{P} \frac{1}{1-\delta-\delta'}\Omega_t + \omega_{t+1},
            \end{split}
        \end{equation}
        where the second inequality follows from the induction hypothesis and $P\geq \frac{128}{\gamma_{t+1}}\log(\frac{\Omega_t}{\delta'})$.
        Take $P\geq 8\tau(\delta/M, \bar{\Omega}_{t}, K_{t+1})/(1-\delta-\delta')$, then from the previous line,
        \begin{equation}
            \Omega_{t+1}\leq \frac{1}{2} \Omega_{t}+\omega_{t+1}\leq \frac{1}{2}\bar{\Omega}_{t}+\omega_{t+1}\coloneq \bar{\Omega}_{t+1}.
        \end{equation}
        This concludes the recurrence.

        Let $\delta=\eta/(2T)$, $\delta'=\eta/(4T)$, then $\delta\leq 1/(2T)\leq 1/2$, $\delta'\leq 1/(4T)\leq 1/4$, therefore, for $0\leq t<T$
        \begin{equation}
            \omega_t = 4\left(1-\frac{3\eta}{4T}\right)^{-(t+1)}\leq 4\left(1-\frac{3}{4T}\right)^{-T} \leq 16,
        \end{equation}
        and from $\bar{\Omega}_{t+1}= \bar{\Omega}_t / 2 + \omega_{t+1}$ with $\bar{\Omega}_0=\omega_0\leq 16$
        \begin{equation}
            \bar{\Omega}_{t+1}\leq 32.
        \end{equation}
        In particular, for $P\geq 32\tau(\delta/M, 32)$ and $P$ satisfying the requirement of Lemma~\ref{lemma:csr_recurrence}, we have $\Omega_{t-1}\leq \bar{\Omega}_{t-1}\leq 32$ for all $t=1,\ldots, T$.
    \end{proof}
\end{lemma}
\begin{proof}[Proof of Theorem~\ref{th:main_wf}]
    Theorem~\ref{th:main_wf} follows from choosing the specific sequence $\bfr = (r_0,\ldots, r_t)$ of Lemma~\ref{lemma:bound_warmness},
    for which $\Omega_{t-1}\leq 32$, plugging back this upper bound in Lemma~\ref{lemma:main_thm_technical}, and replacing $\tau$ by the upper bound~\eqref{eq:upperboundtau}.
    For completeness, the proof of the known upper bound~\eqref{eq:upperboundtau} is provided, see Lemma~\ref{lemma:sg_to_mt}.
\end{proof}

\begin{proof}[Proof of Theorem~\ref{th:wf_smc_greedy}]
    The proof follows the same lines as the proof of
    Theorem~\ref{th:main_wf} except that the requirement on $P$ is replaced by
    an iteration-dependent requirement on $P_t$.
    Lemma~\ref{lemma:marginal_is_warm} is still valid if we replace $P$ by
    $P_t$ in the recurrence defining $\Omega_{t-1}$. This is also the case
    for Lemmas~\ref{lemma:prasrs},~\ref{lemma:concentration}
    and~\ref{lemma:concentration_of_weights}.
    If for all $t=0,\ldots, T-1$,
    \begin{equation}
        P_t\geq \frac{128}{\gamma}\log\left(\frac{\Omega_{t-1}}{\delta'}\right),
    \end{equation}
    then $\Pr(\bfC_{T-1})\geq (1-\delta-\delta')^T$,
    and if $P_T \geq \frac{4}{\gamma\varepsilon^2}\log\left(\frac{2\Omega_{T-1}}{\delta'}\right)$, then $\Pr(\abs{\widehat{\pi}_{T-1}(f)-\pi_{T-1}(f)}<\varepsilon)\geq (1-\delta')\Pr(\bfC_{T-1})$.
    Therefore, along the same lines as in the proof of Lemma~\ref{lemma:main_thm_technical}, we have $\Pr(\abs{\widehat{\pi}_{T-1}(f)-\pi_{T-1}(f)}<\varepsilon)\geq 1-\eta$ provided $\delta = \eta/(2T)$ and $\delta'=\eta/(4T)$.
    Lemma~\ref{lemma:bound_warmness}  also remains valid.
    Therefore, the following conditions on $P_{0:T}$ are sufficient for the high-probability bounds to hold:
    \begin{equation}
        P_t \geq \max\left(\frac{128}{\gamma}\log\left(\frac{128T}{\eta}\right), 32\tau(\eta/(2MT), 32)\right),\indent  P_T \geq \frac{4}{\gamma\varepsilon^2}\log\left(\frac{256T}{\eta}\right).
    \end{equation}
    The conditions above are implied by~\eqref{eq:wf_smc_greedy_requirement}.
\end{proof}

\subsection{Improved finite-sample complexity for the normalising constant (Theorems~\ref{th:strong_normalising_constant},~\ref{th:boosted})}
\label{appendix:improved_th_constant}
We will need the following technical lemma.
\begin{lemma}
    \label{lemma:compound_errors}
    Let $\varepsilon \in (0, 2)$, and $a_0,\ldots, a_T$ be a sequence of positive reals such that $\sum_{t=0}^T \abs{a_t-1}<\varepsilon/2$
    then $\abs{\prod_{t=0}^T a_t -1 }<\varepsilon$.
    \begin{proof}
        Since $\log( 1 + x)\leq x$, we have $\log a_t\leq \log( 1 + \abs{a_t-1}) \leq \abs{a_t-1}$, and $\prod_{t=0}^T a_t = \exp\left(\sum_{t=0}^T \log(a_t)\right)< e^{\varepsilon/2}\leq 1+\varepsilon$.
        Furthermore, $\abs{1-a_t}<\varepsilon/2<1$ and thus $\prod_{t=0}^T a_t \geq \prod_{t=0}^T (1-\abs{1-a_t})\geq 1 - \sum_{t=0}^T \abs{a_t-1}> 1 - \varepsilon/2$.
    \end{proof}
\end{lemma}
\begin{lemma}
    \label{lemma:union}
    For any $\varepsilon \in (0, 2)$,
    \begin{equation}
        \Pr(\abs{\widehat{Z}_T/Z_T-1}\geq \varepsilon) \leq \Pr(\bfC_{T-1}^{\textup{C}}) + \frac{4}{\varepsilon^2}\left(\sum_{t=0}^{T} \left\{ \frac{\bbE[(\widehat{\pi}_{t-1}(G_t)-\pi_{t-1}(G_t))^2\mid \bfC_{t-1}]}{\pi^2_{t-1}(G_t)}\right\}^{1/2}\right)^2.
    \end{equation}
    \begin{proof}
        By Lemma~\ref{lemma:compound_errors}, if $S=\sum_{t=0}^T \abs{\widehat{\pi}_{t-1}(G_t)/\pi_{t-1}(G_t)-1}< \varepsilon/2$ then $\abs{\widehat{Z}_T/Z_T-1}< \varepsilon$ for $\varepsilon \in (0, 2]$.
        By a union bound, the previous remark and Markov's inequality
        \begin{equation}
            \begin{split}
                \Pr(\abs{\widehat{Z}_T/Z_T-1}\geq \varepsilon) &\leq \Pr(\bfC_{T-1}^{\textup{C}}) + \Pr(\abs{\widehat{Z}_T/Z_T-1}\geq \varepsilon, \bfC_{T-1})\\
                &\leq \Pr(\bfC_{T-1}^{\textup{C}}) + \Pr(\mathds{1}_{\bfC_{T-1}}S \geq \varepsilon/2)\\
                &\leq \Pr(\bfC_{T-1}^{\textup{C}}) + \frac{4}{\varepsilon^2} \bbE[\mathds{1}_{\bfC_{T-1}}S^2 ].
            \end{split}
        \end{equation}
        By Minkowski's inequality and $\bfC_{t-1}\supseteq \bfC_{T-1}$,
        \begin{equation}
            \begin{split}
                \bbE[\mathds{1}_{\bfC_{T-1}}S^2 ]^{1/2}&\leq \sum_{t=0}^T \bbE[\mathds{1}_{\bfC_{T-1}} \abs{\widehat{\pi}_{t-1}(G_t)/\pi_{t-1}(G_t)-1}^2]^{1/2}\\
                &\leq \sum_{t=0}^T \{\Pr(\bfC_{t-1})\bbE[\abs{\widehat{\pi}_{t-1}(G_t)/\pi_{t-1}(G_t)-1}^2\mid \bfC_{t-1}]\}^{1/2}\\
                &\leq \sum_{t=0}^T \{\bbE[\abs{\widehat{\pi}_{t-1}(G_t)/\pi_{t-1}(G_t)-1}^2\mid \bfC_{t-1}]\}^{1/2}.
            \end{split}
        \end{equation}
    \end{proof}
\end{lemma}
\begin{lemma}
    \label{lemma:bound_mse}
    Under Assumption~\ref{ass:chisquare_is_small}, for any $1\leq t\leq T$,
    \begin{equation}
        \bbE[(\widehat{\pi}_{t-1}(G_t)-\pi_{t-1}(G_t))^2\mid \bfC_{t-1}]\leq \frac{2\Omega_{t-1}}{\gamma P}\pi_{t-1}(G_t)^2.
    \end{equation}
    For $t=0$, $\bbE[(\widehat{\nu}(G_0)-\nu(G_0))^2] \leq \nu(G_0)^2/N$.
    \begin{proof}
        Let $\bar{G}_t = G_t-\pi_{t-1}(G_t)$, then $\widehat{\pi}_{t-1}(G_t) -\pi_{t-1}(G_t) = \widehat{\pi}_{t-1}(\bar{G}_t)$.
        Let $S_t(Y) = P^{-1}\sum_{p=1}^P \bar{G}_t(Y[p])$, then $\widehat{\pi}_{t-1}(\bar{G}_t) = M^{-1}\sum_{m=1}^M S_t(Y_t^m)$.
        Remark that $Y_t^{1:M}$ are exchangeable, and this is preserved by conditioning on $\bfC_{t-1}$.
        By Cauchy-Schwarz, the exchangeability implies
        \begin{equation}
            \bbE[S_t(Y_t^m)S_t(Y_t^{m'})\mid \bfC_{t-1}]\leq \bbE[S_t(Y_t^m)^2\mid \bfC_{t-1}].
        \end{equation}
        When expanding $\widehat{\pi}_{t-1}(\bar{G}_t)^2$, there are $M(M-1)$ such terms, and the $M$ terms on the diagonal coincide with the previous RHS bound, therefore
        \begin{equation}
            \bbE[\widehat{\pi}^2_{t-1}(\bar{G}_t)\mid \bfC_{t-1}]\leq \bbE[S_t(Y_t^1)^2\mid \bfC_{t-1}].
        \end{equation}
        By the previous Lemma~\ref{lemma:warmness_path_density} and the upper bound for ergodic average under the spectral gap assumption given by~\citet[Theorem 3.1]{10.1214/EJP.v20-4039}
        \begin{equation}
            \begin{split}
                \bbE[S_t(Y_t^1)^2\mid \bfC_{t-1}] &\leq \Omega_{t-1}\bbE[S_t(\bar{Y}_t^1)^2]\\
                &=\Omega_{t-1}\Var(S_t(\bar{Y}_t^1))\\
                &\leq \frac{2\Omega_{t-1}}{\gamma P}\Var_{\pi_{t-1}}(G_t),
            \end{split}
        \end{equation}
        where, to go from the first to the second line, we used $\bbE[S_t(\bar{Y}_t^1)] = \pi_{t-1}(\bar{G}_t) = 0$.
        The result follows by using $\Var_{\pi_{t-1}}(G_t)\leq \pi_{t-1}(G_t)^2$ (Assumption~\ref{ass:chisquare_is_small}).
        The base case follows immediately from $Y_0^{1:M}[1:P]\sim \nu^{\otimes N}$, and $\Var_{\nu}(G_0)\leq \nu(G_0)^2$.
    \end{proof}
\end{lemma}
Let $\bar{\Omega}$ be an upper bound on the warmness parameters $(\Omega_{t})_{t=-1,\ldots, T-1}$.%
\begin{proof}[Proof of Theorem~\ref{th:strong_normalising_constant}]
    Let $\varepsilon \in (0, 2)$, then Lemma~\ref{lemma:union} combined with Lemma~\ref{lemma:bound_mse} yields
    \begin{equation}
        \begin{split}
            \Pr(\abs{\widehat{Z}_T/Z_T-1}\geq \varepsilon) &\leq \Pr(\bfC_{T-1}^{\textup{C}}) + \frac{8(T+1)^2\bar{\Omega}}{\varepsilon^2 \gamma P}\\
            &\leq \Pr(\bfC_{T-1}^{\textup{C}})  + \frac{1}{20},
        \end{split}
    \end{equation}
    provided $P\geq \frac{640T^2\bar{\Omega}}{\varepsilon^2 \gamma}$, using $(T+1)^2\leq 4T^2$.
    Let $\eta = 1/4$, $\delta=\eta/(2T)$, $\delta'=\eta/(4T)$, $r_t = \tau(\delta/M, \Omega_{t-1}, K_t) + 1$ for all $t=0,\ldots, T-1$.
    Let  $P\geq \frac{128}{\gamma}\log(\frac{\Omega_{t-1}}{\delta'})$, then by Lemma~\ref{lemma:csr_recurrence}, $\Pr(\bfC_{T-1}) \geq (1-\delta-\delta')^T\geq 1-T(\delta+\delta') \geq 4/5$.
    Finally, for this choice $\Pr(\abs{\widehat{Z}_T/Z_T-1}\geq \varepsilon)\leq 1/4$.
    Since $\tau(\xi, \omega) \leq \frac{1}{\gamma}\log(\frac{\omega}{\xi})$, provided $P \geq \frac{32}{\gamma}\log(\frac{32M}{\delta})$, we have $\bar{\Omega}\leq 32$ by Lemma~\ref{lemma:bound_warmness}.
    This leads to the following requirement $P\geq
    \max\left(\frac{20480T^2}{\varepsilon^2\gamma},
            \frac{128}{\gamma}\log(512T), \frac{32}{\gamma}\log(256 M T)\right)$,
    which holds if $P\geq (20480T^2)/(\gamma\varepsilon^2)$ and $P\geq 32\log(256 M T)/\gamma$.
    The overall number of Markov steps performed by the algorithm
    is $TMP=\OO\left(\frac{T^3}{\varepsilon^2\gamma}\right)$.
\end{proof}
We now state the guarantee for the product-of-medians estimator of Algorithm~\ref{alg:boosted}.
\begin{theorem}
    \label{th:boosted}
    Assume~\ref{ass:spg} and~\ref{ass:chisquare_is_small}.
    Let $\varepsilon\in (0, 2)$, $\eta\in(0,1]$, $J=12\lceil \log(2T/\eta)\rceil + 1$, $M=1$, and $P\geq \frac{20480T^2}{\gamma \varepsilon^2 }$.
    Then Algorithm~\ref{alg:boosted} returns $\widehat{Z}^{\textup{med}}_T$ such that $\abs{\widehat{Z}^{\textup{med}}_T/Z_T-1}< \varepsilon$ with probability at least $1-\eta$.
    It has cost
    \begin{equation*}
        \OO\left(\frac{T^3}{ \gamma \varepsilon^2} \log\left(\frac{T}{\eta}\right)\right).
    \end{equation*}
\end{theorem}
Its proof relies on the following lemma.
\begin{lemma}
    \label{lemma:boosted}
    Let $\varepsilon\in (0, 1/2]$, $J=12\lceil \log(2T/\eta)\rceil + 1$, and let $\{\widehat{\pi}_{t-1}(G_t)^{(j)}\}_{t=0,\ldots, T, j=1,\ldots, J}$ be $J$ independent copies of
    $\{\widehat{\pi}_{t-1}(G_t)\}_{t=0,\ldots, T}$ such that for all $t=0,\ldots, T$,
    \begin{equation*}
        \Pr(\abs{\widehat{\pi}_{t-1}(G_t)/\pi_{t-1}(G_t)-1}< \varepsilon/T)\geq 3/4.
    \end{equation*}
    Then $\abs{\widehat{Z}^{\textup{med}}_T/Z_T-1}< 4\varepsilon$ with probability at least $1-\eta$ where
    \begin{equation*}
        \widehat{Z}^{\textup{med}}_T = \prod_{t=0}^T \widehat{\frac{Z_t}{Z_{t-1}}}^{\textup{med}},
    \end{equation*}
    and $\widehat{\frac{Z_t}{Z_{t-1}}}^{\textup{med}}$ is the median of the $\{\widehat{\pi}_{t-1}(G_t)^{(j)}\}_{j=1,\ldots, J}$.
\end{lemma}
\begin{proof}[Proof of Lemma~\ref{lemma:boosted}]
    Let $\eta\leq 1$, fix $J = 12\lceil \log(2T/\eta)\rceil + 1$, and let $\{\widehat{\pi}_{t-1}^{(j)}(G_t)\}_{t=0,\ldots, T, j=1,\ldots, J}$ be a sequence
    of independent (across the different $j$'s) sequences of estimated ratio of normalising constants satisfying for each $t=0,\ldots, T$
    \begin{equation}
        \Pr(\abs{\widehat{\pi}_{t-1}(G_t)^{(j)}/\pi_{t-1}(G_t)-1}< \varepsilon/T)\geq 3/4.
    \end{equation}
    Let $\widehat{Z_t/Z_{t-1}}^{\textup{med}}$ be the median of the $J$ estimates $\{\widehat{\pi}_{t-1}(G_t)^{(j)}\}_{j}$ and let $a_t=\widehat{Z_t/Z_{t-1}}^{\textup{med}}/\pi_{t-1}(G_t)$, this for all $t=0,\ldots, T$ so that $\prod_{t=0}^T a_t = \widehat{Z}_T^{\textup{med}}/Z_T $.
    By~\citet[][Lemma 6.1]{JERRUM1986169},
    \begin{equation}
        \Pr(\abs{a_t-1}< \varepsilon/T)\geq 1 - \frac{\eta}{2T}.
    \end{equation}
    By a union bound, with probability at least $1-(T+1)\eta/(2T)\geq 1-\eta$, for each $t=0,\ldots, T$, $1-\varepsilon/T< a_t < 1+ \varepsilon/T$, and then
    \begin{equation}
        1-2\varepsilon\leq 1-\frac{(T+1)\varepsilon}{T}\leq (1-\varepsilon/T)^{T+1} < \widehat{Z}_T^{\textup{med}}/Z_T < (1+\varepsilon/T)^{T+1}\leq e^{2\varepsilon},
    \end{equation}
    and for $\varepsilon\in (0, 1/2]$, $e^{2\varepsilon}\leq 1+4\varepsilon$.
    Thus, the estimator defined by $\widehat{Z}_T^{\textup{med}} = \prod_{t=0}^{T} \widehat{Z_t/Z_{t-1}}^{\textup{med}}$ satisfies for $\varepsilon\in (0, 1/2]$,
    \begin{equation}
        \Pr(\abs{\widehat{Z}_T^{\textup{med}}/Z_T-1}< 4\varepsilon) \geq 1 -\eta,
    \end{equation}
    which in turn implies the lemma.
\end{proof}
\begin{proof}[Proof of Theorem~\ref{th:boosted}]
    Let $\eta = 1/4$, $\delta=\eta/(2T)$, $\delta'=\eta/(4T)$, $r_t = \tau(\delta/M, \Omega_{t-1}, K_t) + 1$ for all $t=0,\ldots, T-1$.
    Assume $P \geq \frac{128 }{\gamma}\log(\Omega_{t-1}/\delta')$, then by Lemma~\ref{lemma:csr_recurrence},
    \begin{equation}
        \begin{split}
            \Pr(\bfC_{T-1}) &\geq (1-\delta-\delta')^{T}\\
            &\geq 1 - T(\delta+\delta')\\
            &\geq 1-\eta/2-\eta/4\geq 4/5,
        \end{split}
    \end{equation}
    and a fortiori $\Pr(\bfC_{t-1})\geq 4/5$.
    By Markov's inequality, Lemma~\ref{lemma:bound_mse} and $\Omega_{t-1}\leq 32$ by Lemma~\ref{lemma:bound_warmness},
    \begin{equation}
        \begin{split}
            \MoveEqLeft\Pr(\abs{\widehat{\pi}_{t-1}(G_t)/\pi_{t-1}(G_t)-1}\geq \varepsilon/T, \bfC_{t-1}) \\
            &\leq \frac{T^2}{\varepsilon^2}\bbE[(\widehat{\pi}_{t-1}(G_t)-\pi_{t-1}(G_t))^2 \mid \bfC_{t-1}]/\pi_{t-1}(G_t)^2\\
            &\leq \frac{64T^2}{\varepsilon^2 \gamma P} \leq 1/20,
        \end{split}
    \end{equation}
    provided $P\geq  \max\left(1280\frac{T^2}{\varepsilon^2 \gamma}, \frac{128 }{\gamma}\log(\Omega_{t-1}/\delta'), \frac{32}{\gamma}\log(256 M T)\right)$, and then
    \begin{equation}
        \begin{split}
            \MoveEqLeft\Pr(\abs{\widehat{\pi}_{t-1}(G_t)/\pi_{t-1}(G_t)-1}\geq \varepsilon/T)\\
            &\leq \Pr\left(\bfC^{\textup{C}}_{t-1}\right)+\Pr\left(\abs{\widehat{\pi}_{t-1}(G_t)/\pi_{t-1}(G_t)-1}\geq \varepsilon/T, \bfC_{t-1}\right) \\
            &\leq 1/4.
        \end{split}
    \end{equation}
    Lemma~\ref{lemma:boosted} implies that the output of Algorithm~\ref{alg:boosted} defined as $\widehat{Z}_T^{\textup{med}} = \prod_{t=0}^{T} \widehat{Z_t/Z_{t-1}}^{\textup{med}}$ with
    $\widehat{Z_t/Z_{t-1}}^{\textup{med}}$ the median estimators over $J=12\lceil \log(2T/\bar{\eta})\rceil + 1$ for any $\bar{\eta}\in(0,1]$, satisfies for $\varepsilon\in (0, 1/2]$
    \begin{equation}
        \Pr(\abs{\widehat{Z}_T^{\textup{med}}/Z_T-1}< 4\varepsilon) \geq 1 -\bar{\eta}.
    \end{equation}
    The statement of Theorem~\ref{th:boosted} is obtained by replacing $\varepsilon$ by $\varepsilon/4$ for $\varepsilon\in (0, 2)$, and checking that $P\geq \frac{20480T^2}{\varepsilon^2\gamma}$  and $P\geq 32\log(256 MT)/\gamma$ imply all requirements.
    In particular, the number of performed Markov steps is $J T M P$, which yields the desired complexity.
\end{proof}

\subsection{Geometric tempering sequence for log-concave and smooth distributions (Theorems~\ref{th:chi_log_concave} and~\ref{th:chi_log_concave_extension})}

\begin{proof}[Proof of Theorem~\ref{th:chi_log_concave}]
    Without loss of generality, we can assume $V(x^\star) = 0$, and $\nabla V(x^{\star}) = 0$ under the common-mode assumption.
    Let $0< \lambda\leq \lambda'$, for short, let $\pi_\lambda\propto e^{-(Q+\lambda V)}$, and let $\alpha_\lambda = \alpha_Q+\lambda \alpha_V$, then $\alpha_\lambda> 0$ under Assumption~\ref{ass:log_concave}.
    Let
    \begin{equation}
        C(\lambda', \lambda) = \log \{\chi^2(\pi_{\lambda'}\mid \pi_{\lambda})+1\}, \indent A(\lambda)=\log \int e^{-(Q+\lambda V)}.
    \end{equation}
    Let $K(\theta) = \log\pi_{\lambda}[e^{\theta V}]$ for any $\theta$, and remark that since $e^{C(\lambda', \lambda)}=\pi_{\lambda}[e^{-2\delta V}]/(\pi_{\lambda}[e^{-\delta V}])^2$ with $\delta = \lambda'-\lambda$,
    \begin{equation}
        \label{eq:clmbd}
        C(\lambda', \lambda) = K(-2\delta)-2\left(A(\lambda')-A(\lambda)\right).
    \end{equation}
    Let $\kappa_n = K^{(n)}(0)$ be the $n$-th cumulant of $V$ under $\pi_{\lambda}$.
    Remark that $K(\theta)=A(\lambda-\theta)-A(\lambda)$, therefore $\kappa_n (-1)^n = A^{(n)}(\lambda)$.
    Therefore, expanding~\eqref{eq:clmbd} in series
    \begin{equation}
        \begin{split}
            \label{eq:series_expansion_C}
            C(\lambda+\delta, \lambda) = \sum_{n=1}^\infty  \kappa_n  (-1)^n \frac{(2\delta)^n}{n!} - 2 \sum_{n=1}^{\infty}\kappa_n (-1)^n \frac{\delta^n}{n!}=\sum_{n=2}^\infty \kappa_n \frac{(-1)^n}{n!} (2^n-2)\delta^n,
        \end{split}
    \end{equation}
    as the first term cancels.
    For any positive real $\theta$ with $\theta\beta_V<\alpha_{\lambda}$, since $V$ is $\beta_V$-smooth and $V(x^\star)=0$,
    \begin{equation}
        \label{eq:smoothness_V}
        \pi_\lambda[e^{\theta V}]
        \leq \pi_\lambda\big[e^{(\theta\beta_V/2)\norm{X-x^\star}^2}\big].
    \end{equation}
    Let $G\sim \mathcal{N}(x^{\star}, I_d)$, remark $\bbE_{G}[e^{\sqrt{2 t}\langle G-x^{\star}, x-x^{\star}\rangle}]=e^{t\norm{x-x^{\star}}^2}$, for any $0\leq t<\alpha_{\lambda}/2$:
    \begin{equation}
        \begin{split}
            \pi_\lambda\big[e^{t\norm{X-x^\star}^2}\big]&= \int \bbE_{G}\left[e^{\sqrt{2t}\langle G-x^{\star}, x-x^\star\rangle}\right] \pi_{\lambda}(\dd x) \\
            &= \bbE_{G}\left[\int e^{\sqrt{2t}\langle G-x^{\star}, x-x^{\star}\rangle} \pi_{\lambda}(\dd x)\right]\\
            &\leq \bbE_{G}\left[\exp\left(\sqrt{2t}\langle G - x^\star, \int (x-x^\star)\pi_{\lambda}(\dd x)\rangle + (\sqrt{2t}\norm{G-x^{\star}})^2/(2\alpha_\lambda)\right)\right]\\
        \end{split}
    \end{equation}
    where we use Fubini's theorem, and then upper bound the integral under $\pi_\lambda$ via the classic exponential integrability result for Lipschitz functions under log-concave measures~\citep[see, ][Proposition 5.4.1]{Bakry2014}.
    We let $m_\lambda=\int (x-x^\star)\pi_{\lambda}(\dd x)$.
    Integrating the bound under the measure of $G$ after careful completion of the square yields
    \begin{equation}
        \label{eq:boundpilmbdetnrm}
        \pi_\lambda\big[e^{t\norm{X-x^\star}^2}\big]\leq (1-2t/\alpha_{\lambda})^{-d/2}e^{\frac{2\norm{m_\lambda}^{2}t}{1-\frac{2t}{\alpha_{\lambda}}}}.
    \end{equation}
    Taking $t=(\rho\beta_V)/2$ with $0<\rho<\alpha_{\lambda}/\beta_V$, and observing that the RHS of~\eqref{eq:boundpilmbdetnrm} is an increasing function of $t$, gives the explicit bound
    \begin{equation}
        \begin{split}
            \sup_{0\leq \theta \leq \rho} \pi_\lambda[e^{\theta V}] \leq (1-\rho\beta_V/\alpha_{\lambda})^{-d/2}e^{\frac{\norm{m_\lambda}^2\rho\beta_V}{1-\rho\beta_V/\alpha_{\lambda}}}
            &\leq (1-\rho\beta_V/\alpha_{\lambda})^{-d/2}e^{\frac{d\rho\beta_V/\alpha_{\lambda}}{1-\rho\beta_V/\alpha_{\lambda}}},
        \end{split}
    \end{equation}
    where the second inequality follows from $\norm{m_{\lambda}}^2\leq \bbE_{\pi_{\lambda}}[\norm{X-x^\star}^2]\leq d/\alpha_{\lambda}$~\citep[][Lemma 4.0.1]{chewi2025logconcave}.
    By Cauchy's estimate, for any $\rho\in(0,\alpha_{\lambda}/\beta_V)$
    \begin{equation}
        |\kappa_n|=|K^{(n)}(0)|\leq \frac{n!}{\rho^n}\sup_{|\theta|\leq \rho}|K(\theta)|,
    \end{equation}
    and this supremum is achieved for $\theta \in \{\pm \rho\}$ (over $\theta\in \mathcal{C}$ with modulus at most $\rho$).
    In particular, take $\rho = \alpha_{\lambda}/(2\beta_V)$, and combine with previous bounds to obtain the explicit cumulant bound
    \begin{equation}
        \label{eq:kappa-cauchy-explicit}
        |\kappa_n|\leq \frac{n!}{\rho^n}\left(d+\frac{d}{2} \log(2) \right).%
    \end{equation}
    We bound the $n=2$ term separately. Using successively Poincaré's inequality with constant $1/\alpha_{\lambda}$, $\norm{\nabla V}\leq \beta_V\norm{x-x^\star}$ and the previously mentioned inequality $\bbE_{\pi_{\lambda}}[\norm{X-x^\star}^2]\leq d/\alpha_{\lambda}$,
    \begin{equation}
        \begin{split}
            \kappa_2 &= \Var_{\pi_\lambda}(V)\\
            &\leq \frac{1}{\alpha_{\lambda}}\pi_\lambda[\norm{\nabla V}^2]\\
            &\leq \frac{\beta^2_V}{\alpha_{\lambda}}\pi_\lambda[\norm{X-x^\star}^2]\\
            &\leq \frac{\beta^2_V d}{\alpha_{\lambda}^2}.
        \end{split}
    \end{equation}
    Let $\delta = c\frac{\alpha_{\lambda}}{\beta_V\sqrt{d}}$ with a fixed $c\in(0,1)$, then $ \frac{|\delta|}{\rho}=\frac{2c}{\sqrt{d}}$, and from~\eqref{eq:kappa-cauchy-explicit}
    the tail sum is bounded by a geometric series:
    \begin{equation}
        \begin{split}
            \sum_{n=3}^\infty \frac{|2\delta|^n}{n!}|\kappa_n|
            &\leq\left(d+\frac{d}{2} \log(2) \right) \sum_{n=3}^\infty \left(\frac{4c}{\sqrt{d}}\right)^n\\
            &= \left(1+\frac{\log(2)}{2}\right) \frac{1}{\sqrt{d}} \frac{64c^3}{1-4c/\sqrt{d}}.
        \end{split}
    \end{equation}
    Let $c\leq 1/8$, then $4c/\sqrt{d}\leq 1/2$, so the denominator is bounded below by $1/2$. %
    Combining the bound for $n=2$ with the tail bound yields
    \begin{equation}
        \label{eq:bound_clmbd_lmbdp}
        \begin{split}
            C(\lambda+\delta, \lambda)
            \leq c^2 + \frac{128 c^3}{\sqrt{d}}\left(1+\frac{\log(2)}{2}\right)
            \leq c^2\left(1+\frac{24}{\sqrt{d}}\right),
        \end{split}
    \end{equation}
    using $0\leq c\leq 1/8$ and $\log(2)/2\leq 0.5$.
    In particular, for $c=\frac{1}{8\sqrt{1+24/\sqrt{d}}}$, one finds that
    \begin{equation}
        \chi^2(\pi_{\lambda'}\mid \pi_{\lambda}) = e^{C(\lambda', \lambda)}-1 \leq  1.
    \end{equation}
    The bound~\eqref{eq:bound_clmbd_lmbdp} holds a fortiori if $\delta\leq c\alpha_{\lambda}/(\beta_V\sqrt d)$.
\end{proof}

The previous theorem assumes the modes of the base measure and the target match; we now get rid of this assumption.
The proof is essentially the same but accounts for the linear term in the quadratic bound of $V$ around
the mode of $\pi_\lambda$.
\begin{proof}[Proof of Theorem~\ref{th:chi_log_concave_extension}]
    Let $\tilde{V} = V - V(x^\star_\lambda)$ with $x^\star_\lambda = \arg \min (Q+\lambda V)$.
    Then the associated cumulant function satisfies $\tilde{K}(\theta) = -\theta V(x^\star_\lambda) + K(\theta)$,
    thus $\kappa_n = K^{(n)}(0) = \tilde{K}^{(n)}(0)$ for $n\geq 2$, so we shall work with $\tilde{V}$.
    Smoothness of $V$ implies for any real positive $\theta$ with $\theta \beta_V<\alpha_\lambda$,
    \begin{equation}
        \begin{split}
            \label{eq:expansion_tilde_V}
            \theta\tilde{V}(x) &\leq \theta \nabla \tilde{V}(x_\lambda^\star)^{\top} (x-x^\star_\lambda) + \frac{\theta \beta_V}{2}\norm{x-x_\lambda^\star}^2\\
            &= t \norm{x-y_{\lambda}}^2 - t \norm{v_\lambda}^2/\beta_V^2,
        \end{split}
    \end{equation}
    where the last line follows from completing the square with $t = \theta \beta_V/2$,
    $y_\lambda = x^\star_\lambda - v_\lambda/\beta_V$ and $v_\lambda = \nabla \tilde{V}(x_\lambda^\star)$.
    We apply the same Gaussian integration trick, but this time under $G\sim \mathcal{N}(y_\lambda, I)$,
    to obtain
    \begin{equation}
        \pi_{\lambda}[e^{\theta \tilde{V}}] \leq e^{-t\norm{v_\lambda}^2/\beta_V^2} (1-2t/\alpha_{\lambda})^{-d/2}e^{\frac{2\norm{\bar{m}_\lambda}^{2}t}{1-\frac{2t}{\alpha_{\lambda}}}},
    \end{equation}
    where $\bar{m}_\lambda = \pi_{\lambda}[X-y_\lambda] = \pi_{\lambda}[X-x_\lambda^\star] + v_\lambda/\beta_V$.
    We drop in the previous bound the prefactor $e^{-t\norm{v_\lambda}^2/\beta_V^2}\leq 1$, then the same Cauchy's estimate
    applied to the cumulant function $\tilde{K}$ (we set $t = (\rho \beta_V)/2$, $0<\rho < \alpha_\lambda/\beta_V$
    and putting afterward $\rho = \alpha_\lambda / (2\beta_V)$) yields
    \begin{equation}
        \abs{\kappa_n}\leq \frac{n!}{\rho^n}\left(\frac{d}{2}\log(2) + \norm{\bar{m}_\lambda}^2\alpha_\lambda\right).
    \end{equation}
    Expanding the square in $\norm{\bar{m}_\lambda}^2$,
    using the upper bound $\bbE_{\pi_\lambda}[\norm{X-x^\star_\lambda}^2]\leq d/\alpha_\lambda$, along with the Cauchy-Schwarz inequality for the cross-term,
    we obtain $\norm{\bar{m}_\lambda}^2 \leq (\sqrt{d/\alpha_\lambda}+\norm{v_\lambda}/\beta_V)^2$, which combined with
    previous bound yields
    \begin{equation}
        \abs{\kappa_n}\leq \frac{n!}{\rho^n}d\left(\frac{\log(2)}{2} + \left(1+\frac{\sqrt{\alpha_\lambda}}{\beta_V}\frac{\norm{v_\lambda}}{\sqrt{d}}\right)^2\right).
    \end{equation}
    The bound on the $n=2$ term follows from the Poincaré inequality and previous computations along with $\norm{\nabla \tilde{V}(x)}\leq \norm{v_\lambda}+\beta_V \norm{x-x^\star_\lambda}$,
    \begin{equation}
        \kappa_2 \leq \left(\frac{\beta_V \sqrt{d}}{\alpha_\lambda} + \frac{\norm{v_\lambda}}{\sqrt{\alpha_\lambda}}\right)^2.
    \end{equation}
    The series expansion~\eqref{eq:series_expansion_C} still holds; we let $\delta^2 = c^2/\kappa_2$ so that the leading term
    is $c^2$. The remainder is a geometric sum of reason
    \begin{equation}
        \begin{split}
            r = \frac{2\delta}{\rho} = \frac{4\beta_V}{\alpha_\lambda}\delta &= \frac{4  \beta_V c}{\beta_V\sqrt{d}+\norm{v_\lambda}\sqrt{\alpha_\lambda}}\\
            &= \frac{4c}{\sqrt{d}}\left\{\frac{1}{1+\norm{v_\lambda}\sqrt{\alpha_\lambda}/(\beta_V\sqrt{d})}\right\} \leq \frac{4c}{\sqrt{d}},
        \end{split}
    \end{equation}
    thus $r\leq 1/2$ provided $c\leq 1/8$, and we recover the same bound on the remainder.
\end{proof}

\subsection{Example~\ref{example:special_case_gaussian}}
The statement in Example~\ref{example:special_case_gaussian} follows from Theorem~\ref{th:chi_log_concave} and the Gaussian case of the following proposition.
\begin{proposition}
    \label{prop:inverse_log_concave}
    For any $0\leq t\leq T$,
    $1/\pi_{t-1}(G_t) \leq e^{\frac{(\lambda_t-\lambda_{t-1})  \beta_V d}{2(\alpha_{Q}+\alpha_V\lambda_{t-1})}}$.

    \begin{proof}
        For any $\lambda\in [0, 1]$, let $\pi_{\lambda} \propto e^{-(Q+\lambda V)}$.
        Let $0\leq\lambda\leq\lambda'\leq 1$, set $\delta = \lambda'-\lambda$.
        By Jensen's inequality, $\beta_V$-smoothness of $V$, and log-concavity: $ \pi_{\lambda}[\norm{X-x^\star}^2]\leq d/\alpha_{\lambda}$: %
        \begin{equation}
            \label{eq:jensen_potential}
            \begin{split}
                \pi_{\lambda}[e^{-\delta V}]&\geq e^{-\delta \pi_{\lambda}[V]}\\
                &\geq \exp\left\{-\delta \pi_{\lambda}\left[V(x^\star)+\frac{\beta_V}{2}\norm{X-x^\star}^2\right]\right\}\\
                &=\exp\left\{-\frac{\delta \beta_V}{2}\pi_{\lambda}[\norm{X-x^\star}^2]\right\}\\
                &\geq \exp\left\{-\frac{\delta \beta_V d}{2\alpha_{\lambda}}\right\}.
            \end{split}
        \end{equation}
        Plugging back $\alpha_\lambda = \alpha_Q + \lambda \alpha_V$, we obtain the result.

        For the Gaussian case, let $q=e^{-\norm{x}^2/2}/(2\pi)^{d/2}$, $U=\frac{\norm{x}^2}{2\sigma^2}$ for some $\sigma^2 >0$.
        Then $G_t = \exp\left(-\frac{\delta \norm{x}^2}{2}\left(\frac{1} {\sigma^2}-1\right)\right)$, and $\pi_{t-1}(\dd x) \propto \exp\left(-\frac{\norm{x}^2}{2}\left(1+\lambda_{t-1}\left(\frac{1}{\sigma^2}-1\right)\right)\right)$.
        Then, $G_t(X_t)$ is a random variable whose logarithm is (up to a positive rescaling) $\chi^2(d)$-distributed.
        This is Example~\ref{example:special_case_gaussian}.
    \end{proof}
\end{proposition}

\subsection{Lower bound (Proposition~\ref{prop:lower_bound})}
\begin{proof}[Proof of Proposition~\ref{prop:lower_bound}]
    We prove separately a lower bound on the variance arising in the CLT for $\log \hat{Z}_T/Z_T$ and a sample size
    requirement.

    Let $\nu=\pi_{-1},\ldots,\pi_T$ be a sequence of distributions on $(\ambientspace, \ambientalgebra)$ with $\dd \pi_t/\dd\nu \propto g_t$.
    For any $0\leq t\leq T$, let $K_t(x, \dd y) = (1-\gamma)\delta_{x}(\dd y) + \gamma \pi_{t-1}(\dd y)$, then $K_t$ is uniformly geometrically ergodic and has spectral gap $\gamma\in (0, 1]$.
    Direct computations yield that for any such $K = K_t$ leaving $\pi=\pi_{t-1}$ invariant, and any $h\in\mathcal{L}_{0,\pi}^2$, we have $ K^p h = (1-\gamma)^p h$.
    Thus, the asymptotic variance $\sigma^2_\infty(h) = \lim_{P\to\infty} P \Var[P^{-1}\sum_{p=1}^P h(\bar{X}^p)] $, where $(\bar{X})$ is stationary, is
    \begin{equation}
        \begin{split}
            \sigma^2_\infty(h) &= \Var_{\pi}[h] + 2\sum_{p=1}^{\infty}\langle h, K^p h\rangle\\
            &= \Var_{\pi}[h]\left(1+2\sum_{p=1}^\infty (1-\gamma)^p\right)\\
            &= \Var_{\pi}[h] \frac{2-\gamma}{\gamma} \geq \frac{\Var_{\pi}[h]}{\gamma}.
        \end{split}
    \end{equation}
    Using previous lower bound with $h=G_t$, the asymptotic variance $\sigma^2_\infty$ arising in the CLT for $\log(\hat{Z}_T/Z_T)$ (Theorem~\ref{th:clt}) satisfies
    \begin{equation}
        \sigma_{\infty}^2 \geq \frac{1}{\gamma} \sum_{t=0}^T \chi^2(\pi_t\mid \pi_{t-1})\geq \frac{T \munderbar{\chi^2}}{\gamma}.
    \end{equation}

    For the sample size requirement, remark that for $\varepsilon\in (0, 1/2]$, $\{\abs{\hat{Z}_T/Z_T-1}\leq \varepsilon\}\subset \{\abs{\log \hat{Z}_T/Z_T}\leq 2\varepsilon\}$.
    Let $\Psi=\Pr(\abs{\mathcal{N}(0, 1)}\leq \cdot)$.
    By the CLT, we have as $P\to\infty$ ($M\geq 1$ is fixed) that $\sqrt{N}\log(\hat{Z}_T/Z_T)/\sigma_{\infty}\to \mathcal{N}(0, 1)$
    where the $\to$ denotes convergence in distribution. By P\'olya's theorem, this convergence
    implies uniform convergence of the cumulative distribution function, that is, for $a>0$
    \begin{equation}
        \Pr(\abs{\sqrt{N}\log \hat{Z}_T/Z_T}/\sigma_\infty\leq a) = \Psi(a) + \oo(1)
    \end{equation}
    where the $\oo$ term hides no dependency in $a$.
    Let $a_\eta= \sqrt{\log(1/\eta)}$, then since $\Psi(a)\leq 1-(\sqrt{2/\pi})e^{-a^2/2}a/\{1+a^2\}$~\citep{10.1214/aoms/1177731721}, we have for any $\eta\in (0, 1)$
    \begin{equation}
        \Psi(a_\eta)\leq 1-\sqrt{2/\pi}\eta^{1/2}a_{\eta}/\{1+a^2_\eta\}.
    \end{equation}
    Let $N(\eta,\varepsilon) = \sigma_\infty^2 \varepsilon^{-2} \log(1/\eta)/4$, then for any $N\leq N(\eta,\varepsilon)$, $2\varepsilon\sqrt{N}/\sigma_\infty \leq a_\eta$.
    Since $\Psi$ is non-decreasing, we obtain
    \begin{equation}
        \begin{split}
            \Pr(\abs{\hat{Z}_T/Z_T-1}> \varepsilon)&\geq \Pr(\abs{\log \hat{Z}_T/Z_T}> 2\varepsilon)\\
            &\geq 1-\Psi(a_\eta)-\oo(1) \\
            &\geq \sqrt{2/\pi}\,\eta^{1/2}a_\eta/\{1+a_\eta^2\}-\oo(1),
        \end{split}
    \end{equation}
    where the $\oo(1)$ vanishes as $\varepsilon\to 0$ provided $N_\varepsilon\to \infty$.
    Thus, for any $N_\varepsilon \leq N(\eta,\varepsilon)$ with $N_\varepsilon\to\infty$, we have
    \begin{equation}
        \liminf_{\varepsilon\to 0} \Pr(\abs{\hat{Z}_T/Z_T-1}>\varepsilon)\geq \sqrt{\frac{2}{\pi}} \frac{\eta^{1/2}\sqrt{\log(1/\eta)}}{1+\log(1/\eta)}.
    \end{equation}
\end{proof}

\subsection{Lifting the spectral gap assumption (Section~\ref{sec:spectral_gap_limitation})}

The resampled particles $\tilde{X}_t^{1:M}$ and the particles $X_t^{1:M}$ produced by Algorithm~\ref{alg:smc} satisfy for $\mathcal{G}_{t-1}=\sigma(X_{0:t-1}^{1:M})$:
\begin{equation*}
    \begin{split}
        X_0^{1:M}&\sim \nu^{\otimes M}\\
        \tilde{X}_t^{1:M}\mid \mathcal{G}_{t-1} &\sim \left(\sum_{m'=1}^M \frac{G_{t-1}(X_{t-1}^{m'})}{\sum_{m''=1}^M G_{t-1}(X_{t-1}^{m''})} \delta_{X_{t-1}^{m'}}\right)^{\otimes M}\\
        X_t^{m}\mid \tilde{X}_t^{1:M}&\sim \delta_{\tilde{X}_t^{m}}K_t^{P-1}(\dd X_t^m).
    \end{split}
\end{equation*}
The coupling procedure used in~\citet{MR4630952} is a maximal coupling performed independently for each $m=1,\ldots, M$ between $X_t^m$ and $\bar{X}_t^m\sim \pi_{t-1}$ marginally.
We let $\mathcal{F}_t = \sigma(\mathcal{G}_{t-1}, X_t^{1:M}, \bar{X}_t^{1:M})$.
We let $\bfA_t = \{X_t^{1:M} = \bar{X}_t^{1:M}\}$, $\bfB_t = \{\widehat{\pi}_{t-1}(G_t)\geq 2/3\pi_{t-1}(G_t)\}$ and $\bfC_t = \bfA_t\cap \bfB_t$.
For any function $f:\ambientspace\to \bbR$, we let $\widehat{\pi}_{t-1}(f) = M^{-1}\sum_{m=1}^M f(X_t^m)$.%

\begin{proof}[Proof of Theorem~\ref{th:normalising_smc}]%
    On event $\cap_{t=0}^T\bfC_t$, all end particles satisfy $X_t^{1:M}=\bar{X}_t^{1:M}$, $t=0,\ldots,T$, so that $\widehat{Z}_T=\bar{Z}_T:=\prod_{t=0}^T M^{-1}\sum_{m=1}^M G_t(\bar{X}_t^m)$.
    By~\citet[Lemma 4.3]{MR4630952}, $\bar{X}_t^{1:M}\sim \pi_{t-1}^{\otimes M }$. Moreover, the maximal couplings at iteration $t$ use fresh randomness and $\bar{X}_t^m\sim\pi_{t-1}$ whatever $\tilde{X}_t^m$,
    so that $\bar{X}_t^{1:M}$ is independent of $\mathcal{F}_{t-1}$, and $(\bar{X}_t^{1:M})_{t=0,\ldots,T}$ are independent.
    Therefore, using Assumption~\ref{ass:chisquare_is_small}, for $M\geq T+1$,
    \begin{equation*}
        \bbE[(\bar{Z}_T/Z_T)^2] = \prod_{t=0}^T\left(1+\frac{\Var_{\pi_{t-1}}(G_t)}{M\pi_{t-1}(G_t)^2}\right)\leq \left(1+\frac{1}{M}\right)^{T+1}\leq 1+\frac{2(T+1)}{M},
    \end{equation*}
    and Chebyshev's inequality yields
    \begin{equation*}
        \begin{split}
            \Pr(\abs{\widehat{Z}_T/Z_T-1}\geq \varepsilon)
            &\leq \Pr( (\cap_{t=0}^T\bfC_t)^{\textup{C}}) + \Pr(\abs{\bar{Z}_T/Z_T-1}\geq \varepsilon)\\
            &\leq \Pr( (\cap_{t=0}^T\bfC_t)^{\textup{C}}) + \frac{2(T+1)}{M \varepsilon^2}\\
            &\leq  \Pr( (\cap_{t=0}^T\bfC_t)^{\textup{C}}) + \frac{1}{8},
        \end{split}
    \end{equation*}
    provided $M\geq \frac{16 (T+1)}{\varepsilon^2}$.

    Take $\delta=\delta'\asymp 1/T$, and let $\mathbf{D}_t = \cap_{s=0}^t \bfC_s$.
    From the proof of~\citet[][Lemma 4.3, Corollary 4.1]{MR4630952}, provided $M\gtrsim \log(1/\delta')$ and $P\gtrsim \tau_{\beta}(\delta/M, 2)$, $\Pr(\bfA_t^{\textup{C}}\cap \bfC_{t-1}) = \Pr(\bfA_t^{\textup{C}}\mid \bfC_{t-1})\Pr(\bfC_{t-1})\leq \delta$ and $\Pr(\bfB_t^{\textup{C}}\cap \bfA_t)\leq \delta'$.
    Since $\mathbf{D}_t^{\textup{C}}\subseteq \mathbf{D}_{t-1}^{\textup{C}}\cup(\bfA_t^{\textup{C}}\cap \bfC_{t-1})\cup(\bfB_t^{\textup{C}}\cap \bfA_t)$,
    \begin{equation*}
        \Pr(\mathbf{D}_t^{\textup{C}})\leq \Pr(\mathbf{D}_{t-1}^{\textup{C}}) + \delta + \delta',
    \end{equation*}
    and $\Pr(\mathbf{D}_0^{\textup{C}})\leq \delta'$, hence $\Pr((\cap_{t=0}^T \bfC_t)^{\textup{C}}) = \Pr(\mathbf{D}_T^{\textup{C}})\leq (T+1)(\delta+\delta')\leq 1/8$ for numerical constants in front of $\delta$ and $\delta'$ small enough,
    and therefore the failure probability is at most $1/4$.
    In particular, all requirements are satisfied with $M=\Omega(\frac{T}{\varepsilon^2})$ and $P=\Omega(\tau_{\beta}(\frac{1}{MT}))$.
\end{proof}

\subsection{ULMC + MALA for SMC}
\begin{proof}[Proof of Proposition~\ref{prop:smc_mala_warm_start}]
    Theorem 4.1 of~\citet{AltChe24HighAcc} along with their following remark state that, for any $x\in\bbR^d$,
    \begin{equation}
        \chi^2(\delta_x K_{\textup{ULMC}}^{n(x)}\mid \pi)\lesssim \varepsilon^2, \qquad\textup{provided } n(x) = \tilde{\OO}\left(\kappa^{3/2}d^{1/2}\varepsilon^{-1}\log\left(\sqrt{d/\alpha}+\norm{x-x^\star}\right)\right),
    \end{equation}
    (hiding no dependency in $x$).
    If $\chi^2(\delta_x K_{\textup{ULMC}}^{n(x)}\mid \pi)=\OO(1)$, then $\textup{TV}(\delta_x K_{\textup{ULMC}}^{n(x)}K_{\textup{MALA}}^{n'}, \pi)\lesssim \xi$
    provided $n' = \tilde{\OO}(\kappa d^{1/2}\mathop{\textup{polylog}}\xi^{-1})$ (hiding $\log$-dependency in $\kappa$ and $d$ but no dependency in $x$) by~\citep[Theorem 1; Theorem 7.3.5]{pmlr-v134-chewi21a, chewi2025logconcave}.
    Since $\pi$ is $\alpha$-strongly log-concave, $\norm{\cdot-x^\star}$ is $1/\alpha$-sub-Gaussian, thus for any $t\geq 0$,
    \begin{equation}
        \label{eq:warm_start_subgaussian}
        \pi(\norm{X-x^\star}\geq \sqrt{d/\alpha}+t)\leq e^{-\alpha t^2/2}.
    \end{equation}
    Let $t=\sqrt{2\log(1/\xi)/\alpha}$ in~\eqref{eq:warm_start_subgaussian} and
    \begin{equation}
        \label{eq:warm_start_radius}
        R = \sqrt{d/\alpha}+\sqrt{2\log(1/\xi)/\alpha},
    \end{equation}
    then $
    \pi(\norm{X-x^\star}\geq R)\leq \xi$.
    Let $L = \log(\sqrt{d/\alpha}+R)$ and take $n = \tilde{\OO}(\kappa^{3/2}d^{1/2}L)$, then remark that for any $x\in \bbR^d$,
    \begin{equation}
        n = \tilde{\OO}(\kappa^{3/2}d^{1/2}\log(\sqrt{d/\alpha}+\norm{x-x^\star})\varepsilon^{-1}_r)
    \end{equation}
    with $\varepsilon_r = \log(\sqrt{d/\alpha}+r)/L$, $r=\norm{x-x^\star}$.
    This implies that for any $x\in \bbR^d$,
    \begin{equation}
        \begin{split}
            \textup{TV}(\delta_x K_{\textup{ULMC}}^{n}K_{\textup{MALA}}^{n'},\pi)&\leq \textup{TV}(\delta_x K_{\textup{ULMC}}^{n} ,\pi) \\
            &\lesssim \sqrt{\chi^2(\delta_x K_{\textup{ULMC}}^{n}\mid \pi)}\lesssim \varepsilon_r.
        \end{split}
    \end{equation}
    On the bulk $r\leq R$, the definition of $L$ gives $\varepsilon_r\leq 1$, hence
    $\chi^2(\delta_x K_{\textup{ULMC}}^{n}\mid\pi)=\OO(1)$, and by the choice of $n$ and $n'$,
    \begin{equation}
        \label{eq:bound_tv_kulmc_mala}
        \textup{TV}(\delta_x K_{\textup{ULMC}}^{n}K_{\textup{MALA}}^{n'}, \pi)\lesssim \xi,
        \qquad \forall x: \norm{x-x^\star}\leq R.
    \end{equation}
    On the tail $r>R$ we bound the total variation by one, so that by~\eqref{eq:warm_start_subgaussian}
    \begin{equation}
        \int_{r> R} \textup{TV}(\delta_x K_{\textup{ULMC}}^n K_{\textup{MALA}}^{n'},\pi) \pi(\dd x)
        \leq \pi(\norm{X-x^\star}> R)\leq \xi.
    \end{equation}
    Combining with~\eqref{eq:bound_tv_kulmc_mala} on $\{r\leq R\}$ yields
    \begin{equation}
        \int \textup{TV}(\delta_x K_{\textup{ULMC}}^{n}K_{\textup{MALA}}^{n'},\pi)\pi(\dd x)\lesssim \xi.
    \end{equation}
\end{proof}

\begin{proof}[Proof of Corollary~\ref{cor:zhat_ulmc_mala}]
    By Proposition~\ref{prop:smc_mala_warm_start}, one application of the composite
    kernel at accuracy $\xi=1/(MT)$ achieves
    $\int\textup{TV}(\delta_x K_t,\pi_{t-1})\,\pi_{t-1}(\dd x)\le\xi$ at a cost of
    $n+n'=\tilde{\OO}(\kappa^{3/2}\sqrt d)$ base Markov steps per particle, uniformly
    along the path. The number of Markov steps is $TM(n+n')$ for $\widehat Z_T$, with
    $M=\Omega(T/\varepsilon^2)$ (Theorem~\ref{th:normalising_smc}). Substituting
    $T=\tilde\Theta(\sqrt d\,\kappa_V)$ and $\kappa_V\le\kappa$ gives the rate.
\end{proof}

\subsection{Mixing times and practical implementations of Proximal Sampler for sampling convex bodies}
Let $\mathcal{K}$ be a convex body (non-empty, compact, convex set of $\bbR^d$), we are interested in sampling
from $\pi = \mathds{1}_{\mathcal{K}}/\textup{Vol}(\mathcal{K})$. We propose to use
the Proximal Sampler (PS, Algorithm~\ref{alg:ps_capped}) and adapt the analyses of~\citet{11369058,3737916.3741356}
to provide a practical complexity of computing the volume of $\mathcal{K}$ using SMC.

Let $\gamma_h(x) = (2\pi h)^{-d/2}\exp(-\frac{\norm{x}^2}{2h})$.
Write $l(y):=\int_{\mathcal{K}}\gamma_h(z-y)\dd z=\Pr(\mathcal{N}(y,hI_d)\in\mathcal{K})$
for the acceptance probability of the backward step, and
$a(y):=1-(1-l(y))^{N'}$ for the probability that the budget suffices. The kernel
$\widehat{K}$ of \cref{alg:ps_capped} is
\begin{equation}
    \label{eq:phat}
    \widehat{K}(x,\cdot)=\int_{\mathbb{R}^d}\gamma_h(y-x)
    \Big[a(y)\,\mathcal{N}(y,hI_d)\big|_{\mathcal{K}}(\cdot)
    +\big(1-a(y)\big)\,\delta_x(\cdot)\Big]\dd y .
\end{equation}
We define the ideal kernel $K$ similarly by assuming feasibility of the backward step
(equivalently $N'=\infty$ in $\hat{K}$). Its transition density with respect to Lebesgue measure on $\mathcal{K}$ is
\begin{equation}
    \label{eq:pdens}
    p(x,x')=\int_{\mathbb{R}^d}\gamma_h(y-x)\,\frac{\gamma_h(x'-y)}{l(y)}\dd y ,
    \qquad x,x'\in\mathcal{K}.
\end{equation}

Our analysis is an adaptation of their mixing time analyses, which are adapted
to general warm (initial) distributions and the idealised kernel ($N'=\infty$, $K$) along with
a failure-probability analysis.
In particular, we aim at two lower bounds, on both $N'$, the maximal number of proposals
per move, and $n$, the number of moves, so as to achieve
$    \int_{\mathcal{K}}\textup{TV}(\delta_x \hat{K}^n,\pi)\pi(\dd x)\lesssim \xi$.


We assume throughout that $\mathcal{B}_1\subset \mathcal{K}\subset \mathcal{B}_D$ for some $D\geq 1$; this is a standard assumption.

\begin{lemma}
    \label{lemma:warmness_ps}
    Assume $\mathcal{B}_1\subseteq \mathcal{K}\subseteq \mathcal{B}_D$.
    For every $x\in \mathcal{K}$, the law $\mu=\delta_x K$ is absolutely continuous w.r.t.\ $\pi$ and
    \begin{equation*}
        \norm*{\frac{\dd \mu}{\dd\pi}}_{\mathcal{L}^\infty(\pi)}\leq e^{3/2}\left(\frac{4D^2}{\sqrt{h}}\right)^d.
    \end{equation*}
    \begin{proof}[Proof]
    Throughout the proof, $x\in\mathcal{K}$ is fixed. We let $\omega_d = \textup{Vol}(\mathcal{B}_1)$.
    Fix $y\in\bbR^d$, we start by finding a lower bound on $l(y)$.
    Fix $p\in \mathcal{K}$, $t\in (0, 1]$. Since $\mathcal{B}_1\subset \mathcal{K}$ and $\mathcal{K}$ is convex, we have $\mathcal{B}_{t}((1-t)p) = (1-t)p+t\mathcal{B}_1\subset \mathcal{K}$.
    For any $z\in \mathcal{B}_{t}((1-t)p)$, we have $\norm{z-p}\leq \norm{z-(1-t)p}+\norm{(1-t)p-p}\leq t + t \norm{p}$, and $\norm{p}\leq D$ since $\mathcal{K}\subset \mathcal{B}_D$, thus $\mathcal{B}_{t}((1-t)p)\subset \mathcal{K}\cap \mathcal{B}_{t(D+1)}(p)$.
    Take $0<r\leq 1$ and $t=r/(2D)$, then since $D\geq 1\geq r$, $t\leq 1/2$ and $t(D+1)\leq 2Dt\leq r$, thus $\mathcal{B}_{r/(2D)}((1-t)p)\subset \mathcal{K}\cap \mathcal{B}_r(p)$.
    Then $\textup{Vol}(\mathcal{K}\cap \mathcal{B}_r(p))\geq \omega_d (r/(2D))^d$, and this holds for any $p\in \mathcal{K}$, thus $\inf_{p\in\mathcal{K}}\textup{Vol}(\mathcal{K}\cap \mathcal{B}_r(p))\geq \omega_d (r/(2D))^d$,
    provided $0<r\leq 1$.
    Let $p$ be the orthogonal projection of $y$ onto $\mathcal{K}$ and $\delta =\textup{dist}(y, \mathcal{K})$ the realised distance to $\mathcal{K}$.
    Then for every $z\in \mathcal{K}\cap \mathcal{B}_r(p)$, $\norm{z-y}\leq \norm{z-p}+\norm{p-y}\leq r + \delta$.
    Since $\gamma_h$ is radially decreasing, $
    \inf_{z\in\mathcal{K}\cap \mathcal{B}_r(p)} \gamma_h(z-y) \geq (2\pi h)^{-d/2}\exp\{-\frac{(r+\delta)^2}{2h}\}
    $, thus
    \begin{equation}
        \begin{split}
            l(y)= \int_{\mathcal{K}} \gamma_h(z-y)\dd z &\geq \textup{Vol}(\mathcal{K}\cap \mathcal{B}_r(p)) \times (2\pi h)^{-d/2}\exp\left\{-\frac{(r+\delta)^2}{2h}\right\}\\
            &\geq \omega_d (r/(2D))^d\times (2\pi h)^{-d/2}\exp\left\{-\frac{(r+\delta)^2}{2h}\right\}.
        \end{split}
    \end{equation}
    Fix $x'\in \mathcal{K}$, going back to our goal of finding a bound on $p(x, x')$, we need to bound
    \begin{equation}
        \int \gamma_h(y-x)\gamma_h(x'-y) \exp\{+\frac{(r+\delta)^2}{2h}\} \dd y = \frac{1}{(2\pi h)^{d}} \int \exp\left\{-\frac{1}{2h}(A-(\delta+r)^2)\right\}\dd y,
    \end{equation}
    with $A=\norm{y-x}^2 + \norm{y-x'}^2$.
    By definition of the distance to $\mathcal{K}$, $\delta=\textup{dist}(y, \mathcal{K})\leq \min(\norm{y-x}, \norm{y-x'})$,
    thus $\delta^2 \leq A/2$ with $A=\norm{y-x}^2+\norm{y-x'}^2$. By the arithmetico-geometric inequality, $2\delta r \leq \frac{1}{2}\delta^2 + 2r^2$.
    Thus,
    \begin{equation}
        (\delta+r)
        ^2 \leq \frac{3}{4}A + 3 r^2,
    \end{equation}
    from which we deduce that $A-(\delta+r)^2 \geq A/4 -3r^2 \geq \norm{y-x}^2 / 4 - 3r^2$.
    Let $r = \sqrt{h}$, plugging all the pieces together, we obtain
    \begin{equation}
        \begin{split}
            p(x, x')&\leq \frac{1}{ \inf_{p\in\mathcal{K}} \textup{Vol}(\mathcal{K}\cap \mathcal{B}_r(p)) } (2\pi h)^{-d/2} \int \exp\left\{-\frac{1}{2h}(\norm{y-x}^2/4 -3h)\right\}\dd y\\
            &\leq \left(\frac{2D}{r}\right)^{d}\omega_d^{-1} 2^d e^{3/2}.
        \end{split}
    \end{equation}
    Since $\dd\pi = \mathds{1}_\mathcal{K} \dd x / \textup{Vol}(\mathcal{K})$ and since $\mathcal{K}\subset \mathcal{B}_D$, $\textup{Vol}(\mathcal{K})\leq \omega_d D^d$.
    Thus,
    \begin{equation}
        \frac{\dd\mu} { \dd \pi} (x') \leq \omega_d D^d \frac{\dd \mu} { \dd x} (x') =  \omega_d D^d p(x, x')  \leq e^{3/2}\left(\frac{4D^2}{r}\right)^d,
    \end{equation}
    and this holds for any $x\in\mathcal{K}$.

    \end{proof}
\end{lemma}

Let $T$ be the number of rejected steps, capped at $N'$, before eventually $\mathcal{K}\ni x'\sim \mathcal{N}(y, hI_d)$, then
$\bbE[T\mid y] = \sum_{k=1}^{N'} \Pr[T\geq k \mid y]$ with $\Pr[T\geq k \mid y] = (1-l(y))^{k-1}$, thus $\bbE[T\mid y]= \frac{1-(1-l(y))^{N'}}{l(y)}\leq N' \wedge\frac{1}{l(y)}$.
\begin{lemma}
    \label{lemma:failure_probability}
    Let $\delta = 2/d$, $N'\geq 8$ and $h\leq (d^2 \log N')^{-1}$, then $\bbE[\frac{1}{l}\wedge N']\leq 9$ and $\bbE[(1-l)^{N'}]\leq \frac{4}{N'}$.
    \begin{proof}
        We follow the proof of~\citep[Lemma 4]{3737916.3741356}.
        Since $x\sim \pi$ and $y\mid x \sim \mathcal{N}(x, hI_d)$, and since $\gamma_h$ is even, the density of $y$ is
        \begin{equation}
            \pi_h(y) = \int \frac{\mathds{1}_{\mathcal{K}}(z)}{\textup{Vol}(\mathcal{K})}\gamma_h(y-z)\dd z = \frac{l(y)}{\textup{Vol}(\mathcal{K})}.
        \end{equation}
        Let $\mathcal{K}_{\delta} = \mathcal{K}+\mathcal{B}_{\delta}$ for any $\delta > 0$.
        Then $\bbE_{\pi_h}[l^{-1}\wedge N'] \leq N'\pi_h(\mathcal{K}_{\delta}^{\textup{C}}) + \textup{Vol}(\mathcal{K}_{\delta})/\textup{Vol}(\mathcal{K}) \leq (1+\delta)^d + N'\exp\left(-\frac{\delta^2}{2h}+\delta d\right)$.
        The choice of $N'$, $\delta$ and $h$ yields
        \begin{equation}
            \begin{split}
                \bbE_{\pi_h}\left[\frac{1}{l}\wedge N'\right]&\leq (1+\delta)^d + N'\exp\left(-\frac{\delta^2}{2h}+\delta d\right)\\
                &\leq e^2 + N' e^2 \exp(-2 \log N') = e^2(1+1/{N'})\leq 9.
            \end{split}
        \end{equation}

        Similarly, using $(1-t)^n\leq e^{-tn}$ for $t\in [0, 1]$ and $\sup_{t\in (0, 1]} te^{-N't}\leq 1/(eN')$ attained at $t=1/N'$, we can
        split the integrals into the $\mathcal{K}_{\delta}$ and $\mathcal{K}_{\delta}^{\textup{C}}$ and obtain
        \begin{equation}
            \begin{split}
                \bbE_{\pi_h}[(1-l)^{N'}] \leq \frac{(1+\delta)^{d}}{eN'} +\exp\left(-\frac{\delta^2}{2h}+\delta d\right) \leq \frac{e}{N'}+\frac{e^2}{(N')^2} \leq \frac{4}{N'}.
            \end{split}
        \end{equation}
    \end{proof}
\end{lemma}

\begin{proposition}
    \label{prop:Khatps}
    Assume $\mathcal{B}_1\subseteq \mathcal{K}\subseteq \mathcal{B}_D$,
    and further assume that $\pi$ satisfies a Poincaré inequality with $C_\pi < \infty$.
    Let $\xi\in(0,1)$ and let $\Lambda=\tfrac32+d\log\frac{4D^{2}}{\sqrt h}$.
    Take
    \begin{equation*}
        \begin{split}
            &0<h\le\min\{1,C_{\pi}\},\qquad N'\ge8,\qquad h\le\big(d^{2}\log N'\big)^{-1},\\
            &n\ \ge\ 1+\frac{C_{\pi}}{h}\Big(\Lambda+2\log\frac1\xi\Big),\\
            &N'\ \ge\ \frac{8n}{\xi}.
        \end{split}
    \end{equation*}
    Then
    \begin{equation*}
        \int\TV\big(\delta_x\widehat K^{\,n},\pi\big)\,\pi(\dd x)\ \le\ \xi,
        \qquad\text{i.e.}\qquad \tau_\beta(\xi,\widehat K)\ \le\ n .
    \end{equation*}
    \begin{proof}[Proof of Proposition~\ref{prop:Khatps}]
    Fix $x\in\mathcal{K}$.
    By Theorem 2 of~\citet{3737916.3741356} applied $n-1$ times and the warmness Lemma~\ref{lemma:warmness_ps},
    \begin{equation}
        \begin{split}
            \chi^2(\delta_x K^n\mid \pi)&\leq \frac{\chi^2(\delta_x K\mid \pi)}{(1+h/C_\pi)^{2(n-1)}}\\
            &\leq \frac{e^{\Lambda}}{(1+h/C_\pi)^{2(n-1)}}.
        \end{split}
    \end{equation}
    Since $h\leq C_\pi$ and $\log(1+u)\geq u/2$ for $u\in [0, 1]$, $\chi^2(\delta_x K^n\mid \pi)\leq e^{\Lambda - (n-1)h/C_{\pi}}\leq \xi^2$ by the choice of $n$.
    By Cauchy-Schwarz, $\textup{TV}(\delta_x K^n , \pi)\leq \xi/2$.

    We now construct a coupling between the two chains arising from $\hat{K}$ and $K$.
    On the same probability space, at step $j\geq 1$, given $X_{j-1}$, let $Y_{j}\sim \mathcal{N}(X_{j-1}, hI_d)$ and let $Z_{j,1},Z_{j, 2},\ldots$ be an infinite sequence
    of i.i.d.\ variables with $Z_{j,i}\sim \mathcal{N}(Y_{j}, hI_d)$, let $I_j = \inf\{i\geq 1 : Z_{j,i}\in \mathcal{K}\}$,
    and let $\hat{X}_j$ be defined as
    \begin{equation}
        \begin{split}
            \hat{X}_j = \left\{
                            \begin{array}{rl}
                                &Z_{j, I_j}\textup{ if } I_j \leq N' \\
                                &\hat{X}_{j-1} \textup{ else.}
                            \end{array}
            \right.
        \end{split}
    \end{equation}
    Let $X_0 = \hat{X}_0 = x$, then $(\hat{X}, X)$ is a coupling between the two Markov chains starting at $\delta_x$ with kernels $\hat{K}$ and $K$.
    Clearly, both agree until there exists some $j\geq 1$ such that $I_j> N'$, let $A_{j}$ denote this event, then $\Pr[A_j\mid Y_j]=(1-l(Y_j))^{N'}$.
    Start the chain at $X_0\sim \pi$, then $Y_j \sim \pi_h$ for any $j\geq 1$.
    Thus,
    \begin{equation}
        \begin{split}
            \bbE_{X_0 \sim \pi}\textup{TV}(\delta_{X_0} K^n, \delta_{X_0} \hat{K}^n) &\leq \sum_{j=1}^n \bbE_{X_0\sim \pi} [\Pr[A_j\mid Y_j]] \\
            &= \sum_{j=1}^n \bbE_{Y_j\sim \pi_h}[\Pr[A_j \mid Y_j ]]\\
            &= \sum_{j=1}^n \bbE_{Y_j\sim \pi_h}((1-l(Y_j))^{N'})\leq \frac{4 n}{N'} \leq \frac{\xi}{2},
        \end{split}
    \end{equation}
    where the last two inequalities follow from Lemma~\ref{lemma:failure_probability} and the choice of $N'$.
    By the triangle inequality on $\textup{TV}$, we have $\bbE_{X_0\sim \pi}\textup{TV}(\delta_{X_0} \hat{K}^n, \pi)\leq \xi$.
    \end{proof}
\end{proposition}

\begin{proof}[Proof of Corollary~\ref{cor:convex_volume}]
    Proposition~\ref{prop:Khatps} bounds the $\beta$-mixing time of the capped
    proximal sampler at accuracy $\xi=1/(MT)$ by $P=\tilde{\OO}(d^3\bar C\log D)$,
    provided the rejection budget satisfies $N'\ge 8P/\xi=\tilde\Omega(P\,MT)$.
    With the requirements of Theorem~\ref{th:normalising_smc}
    ($M=\Omega(T/\varepsilon^2)=\tilde{\Omega}(d\varepsilon^{-2} \log D)$
    and $T=\tilde\Theta(d \log D)$), this is
    $N'=\tilde\Omega(d^5\bar C(\log D)^3\,\varepsilon^{-2})$.
    Then $\hat{Z}_T$ is
    an estimator of $\mathrm{Vol}(\mathcal{K}_T)$ within relative precision $1\pm \varepsilon$ with probability at least $7/8$ (provided
    we take the numerical constants large enough). The number of Markov steps is $MPT=\tilde{\OO}(d^5\bar C(\log D)^3\,\varepsilon^{-2})$.

    We now tackle the total number of membership queries in instances of Algorithm~\ref{alg:ps_capped} throughout the execution
    of the SMC sampler. We let $c_{P-1}(\tilde{X}^m_{t+1})$ be the number of such queries for the $m$-th chain at iteration $t+1\geq 1$,
    starting from the resampled particle $\tilde{X}^m_{t+1}$, and let $\Gamma=\sum_{t=0}^{T-1}\sum_{m=1}^M c_{P-1}(\tilde{X}^m_{t+1})$
    be the total.
    Within the coupling event $\mathbf{C}_t$ of the proof of Theorem~\ref{th:normalising_smc} ($\mathbf{D}_T =\cap_{t=0}^T \bfC_{t}$),
    the resampled particles $\tilde{X}_{t+1}^m$ have marginal law that is $\OO(1)$-warm with respect to $\pi_{t}$ (Lemma 4.2 of~\citet{MR4630952}), so that
    \begin{equation}
        \bbE[\Gamma \mathds{1}_{\mathbf{D}_T}]\leq M\sum_{t=0}^{T-1} \bbE[c_{P-1}(\tilde{X}_{t+1}^{1}) \mathds{1}_{\bfC_t}]\lesssim M\sum_{t=0}^{T-1} \bbE[c_{P-1}(\bar{X}_{t+1})],
    \end{equation}
    where $\bar{X}_{t+1}\sim \pi_t$.
    By Lemma~\ref{lemma:failure_probability} (which applies because $\mathcal{B}_1\subseteq \mathcal{K}_t$),
    under stationarity and on average, there are at most $9$ proposals
    performed by a run of Algorithm~\ref{alg:ps_capped}, so that $\bbE[\Gamma \mathds{1}_{\mathbf{D}_T}]\leq CTMP$ for some constant
    $C>0$. By Markov's inequality, this implies that $\Pr(\Gamma \mathds{1}_{\mathbf{D}_T}\geq 16CTMP)\leq 1/16$.
    We can always take the constants so that $\Pr(\mathbf{D}_T^{\textup{C}})\leq 1/16$.
    A union bound yields the result.

\end{proof}
\begin{proof}[Proof of Proposition~\ref{prop:stayput}]
    Let us start with the spectral gap upper bound.
    The kernel $\hat{K}$ reads, for any $(x, x')\in\mathcal{K}^2$,
    \begin{equation}
        \hat{K}(x, \dd x') = \int_{\bbR^d}\gamma_h(y-x) \left\{ a(y) \gamma_h(x'-y) / l(y)  + (1-a(y))\delta_x(\dd x')\right\}\dd y,
    \end{equation}
    where $a(y) = 1-(1-l(y))^{N'}$ is the probability that at least one out of the $N'$ proposals is accepted, with $l(y)=
    \int_{\mathcal{K}}\gamma_h(x'-y)\dd x'$.
    We write $Q(x, \dd x')$ for the continuous part of the decomposition of $\hat{K}$.
    Since for any $y\in \bbR^d$, $a(y)\leq N'l(y)$,
    we have $Q(x, \dd x')\leq N'\int_{\bbR^d} \gamma_h(y-x) \gamma_h(x'-y)\dd y = N'\gamma_{2h}(x'-x)$.
    Now, for any set $S\subset \mathcal{K}$, for any $(x, x')\in S\times S^{\textup{C}}$ (with $S^{\textup{C}}=\mathcal{K}\setminus S$),
    $            \hat{K}(x, \dd x') = Q(x, \dd x')\leq N'\gamma_{2h}(x'-x)\dd x'$, thus $\hat{K}(x, S^{\textup{C}}) \leq N' \int_{\mathcal{K}}\gamma_{2h}(x'-x)\dd x'$, that is
    for $G\sim\mathcal{N}(0, I_d)$,
    \begin{equation}
        \begin{split}
            \hat{K}(x, S^{\textup{C}}) \leq N' \Pr(x + \sqrt{2h}G \in \mathcal{K})
            &\leq N' \Pr(0\leq x + \sqrt{2h}G)\\
            &= N'\prod_{i=1}^d\Phi\left(\frac{x_i}{\sqrt{2h}}\right)\\
            &\leq N'\Phi\left(\frac{1}{d\sqrt{2}}\right)^d \leq N'\left(\frac{1}{2} \left(1+\frac{1}{d\sqrt{\pi}}\right)\right)^d\\
            &\leq 1.76 N' \times 2^{-d},
        \end{split}
    \end{equation}
    where on the last lines, we further assume that $x\in S$ so that $\norm{x}\leq \sqrt{h}/d$, and thus component-wise, $\abs{x_i}\leq \sqrt{h}/d$,
    and since $\Phi$ is concave on $[0, \infty)$, we have $\Phi(t)\leq 1/2 + t/\sqrt{2\pi}$ for any $t\geq 0$.

    Now, let $f=\mathds{1}_{S}$, then $\int \hat{K}(x, S^{\textup{C}})f(x)\pi(\dd x)\leq \pi(S) N' \times 1.76 \times 2^{-d}$.
    From the variational characterisation of the spectral gap,
    \begin{equation}
        \gamma \leq \int \hat{K}(x, S^{\textup{C}})f(x)\pi(\dd x)/\Var[f] \leq 1.76 N' \times 2^{-d} / (1-\pi(S)),
    \end{equation}
    with $\pi(S)\leq \omega_d (\sqrt{h}/d)^d 4^{-d}$.

    For the second part, let $g$ be as defined, let $Z = \sum_{p=1}^P\{ \mathds{1}_{S}(\bar{X}^{p})-\pi(S)\}$.
    We have, by the Markov property and stationarity of the chain,
    \begin{equation}
        \begin{split}
            \bbE (Z^2) &= \sum_{q,p=1}^P \{\Pr(\bar{X}^{p}\in S, \bar{X}^q\in S)-\pi(S)^2\}\\
            &= \pi(S)\sum_{q,p=1}^P\{\Pr(\bar{X}^{\abs{p-q}}\in S\mid \bar{X}^0\in S)-\pi(S)\}\\
            &\geq \pi(S)\left(P(1-\pi(S))+2\sum_{r=1}^{m} \{p_r-\pi(S)\} (P-r)\right),
        \end{split}
    \end{equation}
    with $p_r = \Pr(\bar{X}^{r}\in S\mid \bar{X}^{0}\in S)$, and $m=\lfloor P/2\rfloor$, and last line follows by counting $(p, q)\in [1, P]^2$ by their lag $\abs{p-q}=r$.

    Let $\sigma_g^2 = \Var[P^{-1/2}\sum_{p=1}^P g(\bar{X}^{p})]$, we have
    \begin{equation}
        \begin{split}
            \sigma_g^2&= \frac{\Var(Z)}{P\pi(S)(1-\pi(S))} \\
            &\geq 1+\frac{2}{1-\pi(S)}\sum_{r=1}^{m}\{p_r-\pi(S)\} (1-r/P).
        \end{split}
    \end{equation}
    Remark that $\{\bar{X}^r\in S\}\supset \{\bar{X}^0,\ldots, \bar{X}^r \in S\}$, and for any $x\in S$, $\hat{K}(x, S)\geq 1-p$ with $p= 1.76 N'\times 2^{-d}$,
    thus, $p_r\geq (1-p)^r$. Let $q = \min(m, \lfloor 1/p\rfloor)$, and observe that
    $p_r\geq \pi(S)$ because $\hat{K}$ is a positive kernel, and $1-r/P \geq 1/2$ for $r\leq q\leq m$, thus
    $\sigma_g^2 \geq 1 + \sum_{r=1}^{q} \{(1-p)^r - \pi(S)\}$.
    Assume $p\leq 1/2$, we have
    \begin{equation}
        \begin{split}
            \sum_{r=1}^{q} (1-p)^r &= \frac{(1-p)(1-(1-p)^{q})}{p} \\
            &\geq \frac{1}{2p}\min\left(\frac{pq}{2}, \frac{1}{2}\right) \\
            &\geq \frac{1}{8}\min(P-1, 1/p).
        \end{split}
    \end{equation}
    We have $q\pi(S)\leq  p^{-1} \pi(S) \leq 1/16$ (provided $\pi(S)\leq p/16$), and thus
    \begin{equation}
        \sigma_g^2 \geq 1 + \frac{1}{16}\min(P-1, 1/p).
    \end{equation}
    The constraint $\pi(S)\leq p/16 = 1.76 N' 2^{-d}/16$ holds when $d\geq 3$ since $\pi(S)\leq 6 \times d^{-d}4^{-d}$, and $6 \times 2^{-d} d^{-d}\leq 1.76/16$ for $d\geq 3$.
\end{proof}

\subsection{Other technical results}

\begin{lemma}
    \label{lemma:sg_to_mt}
    Let $K$ be a $\pi$-reversible Markov kernel with spectral gap $\gamma > 0$,
    and $\mu$ an $\omega$-warm probability measure with respect to $\pi$.
    Then $\textup{TV}(\mu K^p, \pi)\leq \frac{1}{2}\sqrt{\omega}(1-\gamma)^p$.
    In particular,
    \begin{equation}
        \label{eq:upperboundtau}
        \tau(\xi, \omega) \leq \log(\omega \xi^{-1})/\gamma, \qquad \xi \leq \omega.
    \end{equation}
    \begin{proof}
        Let $h_0 = \frac{\dd \mu}{\dd \pi}$ and $\bar{h}_0 = h_0 -1$, then
        $\bar{h}_0$ is a zero-mean $\pi$-square integrable function, and $\norm{\bar{h}_0}_{2}^2 = \chi^2(\mu\mid \pi)$.
        For any $p\geq 0$, the spectral gap assumption yields
        \begin{equation}
            \chi^2(\mu K^p\mid \pi) = \norm{K^p \bar{h}_0}_2^2 \leq
            (1-\gamma)^{2p} \norm{\bar{h}_0}^2_2 = (1-\gamma)^{2p}
            \chi^2(\mu \mid \pi).
        \end{equation}
        Remark that $\chi^2(\mu\mid\pi)= \int \bar{h}_0^2\dd \pi = \int h_0^2 \dd \pi - 1 \leq \omega \int h_0\dd\pi - 1 = \omega -1$.
        Since $\textup{TV}(\cdot, \cdot)\leq \frac{1}{2}\sqrt{\chi^2(\cdot\mid \cdot)}$, we obtain
        \begin{equation}
            \textup{TV}(\mu K^p, \pi)\leq \frac{1}{2}\sqrt{\omega-1} (1-\gamma)^{p}.
        \end{equation}
        Solving for the smallest $p\geq 0$ such that $\sqrt{\omega-1}(1-\gamma)^p \leq \xi$ and using $1\leq \sqrt{\omega}\leq \omega$ yields~\eqref{eq:upperboundtau}.
    \end{proof}
\end{lemma}

\begin{proposition}
    \label{prop:rho_pcn}
    Assume $V$ satisfies Assumption~\ref{ass:V}, and let $\beta=\beta_V$ for simplicity.
    Let $\pi_{\lambda}\propto \mathcal{N}(0, C)e^{-\lambda V}$ for any $\lambda \in [0, 1]$.
    The spectral gap of the pCN-MH kernel $\gamma(\rho)$ targeting $\pi_{\lambda}$ is lower bounded by
    \begin{equation}
        \label{eq:lower_bound_gamma_rho}
        \gamma(\rho) \geq C_0(1-\rho^2)\exp\left\{-2(1-\rho^2)\lambda
                                               \beta\textup{Tr}(C)\right\},
    \end{equation}
    for some numerical constant $C_0>0$.
    Furthermore, this lower bound is maximised (with respect to $\rho$, for a fixed
    $\lambda$) by taking
    \begin{equation}
        \label{eq:rholmbd}
        \rho(\lambda) =
        \begin{cases}
            \sqrt{1-\lambda^{\textup{C}}/\lambda}, & \lambda \geq \lambda^{\textup{C}}\\
            0, & \lambda < \lambda^{\textup{C}},
        \end{cases}
    \end{equation}
    with $\lambda^{\textup{C}} =  1/ (2\beta\textup{Tr}(C))$.
    Plugging $\rho(\lambda)$ inside the lower bound on the spectral gap gives the following lower bound:
    \begin{equation}
        \label{eq:spg_pcn}
        \gamma_{\textup{pCN}}(\rho(\lambda)) \geq \begin{cases}
                                                      C_0\frac{\lambda^{\textup{C}}}{\lambda} e^{-1}, &\lambda \geq \lambda^{\textup{C}}
                                                      \\
                                                      C_0\exp(-\frac{\lambda}{\lambda^{\textup{C}}}), & \lambda < \lambda^{\textup{C}}.
        \end{cases}
    \end{equation}

    \begin{proof}
        Let $\lambda \in [0, 1]$. Under Assumption~\ref{ass:V}, the potential $\lambda V$ has smoothness constant $\lambda \beta$.
        \citet[Lemma 58]{10.1214/24-AAP2058} state that the spectral gap of the pCN-MH kernel $\gamma(\rho)$ targeting $\pi_{\lambda}$ is lower bounded by~\eqref{eq:lower_bound_gamma_rho}.
        Let $f(\tau) = C_0 \tau^2e^{-2\tau^2 \lambda \beta\textup{Tr}(C)}$ for $\tau\in (0, 1]$, so that $f(\tau)$ is the lower bound in~\eqref{eq:lower_bound_gamma_rho} when $\tau^2 = 1-\rho^2$.

        If $\lambda < \lambda^{\textup{C}} = 1/(2\beta\textup{Tr}(C))$, then $f$ is a strictly increasing function ($f'(\tau)>0$).
        Consequently, the maximum of $f$ is obtained for $\tau = 1$, i.e., $\rho = 0$, and $\gamma(0)\geq C_{0}\exp(-\lambda/\lambda^{\textup{C}})$.
        If $\lambda \geq \lambda^{\textup{C}}$, then $f$ admits a maximum that satisfies the first-order condition $f'(\tau) = 0$, which yields $\tau^2 =\lambda^{\textup{C}}/\lambda \leq 1$, i.e., $\rho = \sqrt{1-\lambda^{\textup{C}}/\lambda }$.
        Plugging this value for $\rho$ yields $\gamma(\rho)\geq C_0 \lambda^{\textup{C}}/\lambda e^{-1}$.
        In both cases, $\gamma(\rho(\lambda))=\Omega(1/(\lambda \beta \textup{Tr}(C)))$.
    \end{proof}
\end{proposition}